\documentclass[12pt]{article} 
\usepackage[sectionbib]{natbib}
\usepackage{array,epsfig,fancyheadings,rotating}
\usepackage[]{hyperref}  
\usepackage{sectsty, secdot}
\sectionfont{\fontsize{12}{14pt plus.8pt minus .6pt}\selectfont}
\renewcommand{\theequation}{\thesection\arabic{equation}}
\subsectionfont{\fontsize{12}{14pt plus.8pt minus .6pt}\selectfont}

\RequirePackage[OT1]{fontenc}
\RequirePackage{amsthm,amsfonts,amsmath,amssymb,amsbsy,array,calrsfs,enumerate,
bbold,mathrsfs,color,graphicx,caption,hyperref,mwe,natbib,soul,titlesec,xr,soul}

\usepackage{mathrsfs}
\DeclareMathAlphabet{\mathscrbf}{OMS}{mdugm}{b}{n}

\newcommand\comment[1]{{\color{red} \fbox{Stuff commented out -- see the latex source.}}}

\newcommand{\bn}{\begin{enumerate}}
	\newcommand{\en}{\end{enumerate}}

\newcommand{\bc}{\begin{cases}}
	\newcommand{\ec}{\end{cases}}
\newcommand{\bse}{\begin{eqnarray*}}
\newcommand{\ese}{\end{eqnarray*}}
\newcommand{\ben}{\begin{eqnarray}}
\newcommand{\een}{\end{eqnarray}}
\newcommand{\var}{{\rm Var}}
\newcommand{\cov}{{\mathbb{C}}{\rm ov}}
\newcommand{\tr}{{\rm tr}}
\newcommand{\diag}{\mathrm{diag}}
\newcommand{\la}{\langle}
\newcommand{\ra}{\rangle}

\DeclareMathOperator*{\argmin}{argmin}

\newcommand{\bbR}{{\mathbb R}}
\newcommand{\bbH}{{\mathbb H}}
\newcommand{\bbL}{{\mathbb L}}

\newcommand{\bdlambda}{{\boldsymbol{\lambda}}}
\newcommand{\bdbeta}{{\boldsymbol{\beta}}}
\newcommand{\bdeta}{{\boldsymbol{\eta}}}
\newcommand{\bdomega}{{\boldsymbol{\omega}}}
\newcommand{\bdepsilon}{{\boldsymbol{\varepsilon}}}

\newcommand{\eqid}{\buildrel d \over =}

\newcommand{\Normal}{{\mathrm{Normal}}}

\newcommand{\wh}{\widehat}
\newcommand{\wt}{\widetilde}
\newcommand{\E}{\mathbb{E}}
\newcommand{\pr}{\mathbb{P}}

\newcommand{\CA}{\mathscr{A}}

\newcommand{\CC}{\mathcal{C}}

\newcommand{\BA}{\boldsymbol{A}}
\newcommand{\BB}{\boldsymbol{B}}

\newcommand{\BI}{{\boldsymbol{I}}}
\newcommand{\BK}{{\boldsymbol{K}}}
\newcommand{\BR}{{\boldsymbol{R}}}

\newcommand{\BX}{\boldsymbol{X}}
\newcommand{\BY}{\boldsymbol{Y}}

\newcommand{\bdf}{{\boldsymbol{f}}}
\newcommand{\bdg}{{\boldsymbol{g}}}

\renewcommand{\theenumi}{\roman{enumi}}

\usepackage[english]{babel}
\usepackage[utf8]{inputenc}
\usepackage{amsmath}
\usepackage{amssymb}
\usepackage{graphicx}
\usepackage[colorinlistoftodos]{todonotes}
\usepackage{caption}
\usepackage{enumerate}
\usepackage{lscape}
\usepackage{caption}
\usepackage{verbatim}
\usepackage{csvsimple}
\usepackage{multirow}
\usepackage{caption}
\usepackage{subcaption}
\usepackage[normalem]{ulem}
\usepackage{booktabs}
\usepackage{multibib}
\usepackage{setspace}

\newcites{supp}{Supplementary References}

\theoremstyle{plain}
\newtheorem{definition}{Definition}
\newtheorem{theorem}{Theorem}
\newtheorem{lemma}{Lemma}
\newtheorem*{lemma*}{Lemma}
\newtheorem{proposition}{Proposition}
\newtheorem{corollary}{Corollary}
\newtheorem{assum}{C.\ignorespaces}
\newtheorem*{assum*}{C}
\newtheorem{remark}{Remark}

\def\II{I\negthinspace I}
\def\III{I\negthinspace I\negthinspace I}

\newcommand{\scrI}{\mathscr{I}}
\newcommand{\scrL}{\mathscr{L}}
\newcommand{\scrT}{\mathscr{T}}
\newcommand{\scrbfL}{\mathscrbf{L}}
\newcommand{\scrbfC}{\mathscrbf{C}}

\newcommand{\scrbfT}{\mathscrbf{T}}

\newcommand{\calS}{ \mathcal{S} }
\newcommand{\gns}{ {\boldsymbol{g}_{\calS}} }
\newcommand{\gnsc}{ {\boldsymbol{g}_{\calS^c}} }
\newcommand{\ws}{ \boldsymbol{\omega}_{\calS} }
\newcommand{\wsc}{ \boldsymbol{\omega}_{\calS^c} }

\newcommand{\fzeros}{ \boldsymbol{f}_{0 \calS} }
\newcommand{\fs}{ \boldsymbol{f}_{\calS} }
\newcommand{\fhats}{ \boldsymbol{\widehat{f}}_{\calS} }
\newcommand{\fhatsc}{ \boldsymbol{\widehat{f}}_{\calS^c} }
\newcommand{\fchecks}{ \boldsymbol{\check{f}}_{\calS} }
\newcommand{\fhatts}{ \boldsymbol{\widehat{f}}_{\calS} }
\newcommand{\fhaths}{ \boldsymbol{\widehat{f}}_{\wh\calS} }
\newcommand{\ftildets}{ \boldsymbol{\widetilde{f}}_{\calS} }

\newcommand{\Tn}{ \mathscrbf{T}_n }
\newcommand{\TnSS}{  \mathscrbf{T}_n^{(\calS, \calS)} }
\newcommand{\TnScS}{  \mathscrbf{T}_n^{(\calS^c, \calS)} }
\newcommand{\TnjS}{  \mathscrbf{T}_n^{(j, \calS)} }
\newcommand{\TnSSc}{  \mathscrbf{T}_n^{(\calS, \calS^c)} }
\newcommand{\TnScSc}{  \mathscrbf{T}_n^{(\calS^c, \calS^c)} }
\newcommand{\TSS}{  \mathscrbf{T}^{(\calS, \calS)} }
\newcommand{\TjS}{  \mathscrbf{T}^{(j, \calS)} }
\newcommand{\TScS}{  \mathscrbf{T}^{(\calS^c, \calS)} }

\newcommand{\Tnjj}{  \mathscr{T}_n^{(j, j)} }
\newcommand{\Tjj}{  \mathscr{T}^{(j, j)} }
\newcommand{\T}{  \mathscrbf{T} }

\newcommand{\JnSS}{  \mathscrbf{T}_{n, \lambda_2}^{(\calS, \calS)} }
\newcommand{\JSS}{  \mathscrbf{T}_{\lambda_2}^{(\calS, \calS)} }
\newcommand{\QSS}{  \mathscrbf{Q}^{(\calS, \calS)} }
\newcommand{\ESS}{  \mathscrbf{E}^{(\calS, \calS)} }

\allowdisplaybreaks
\newcommand{\vertiii}[1]{{\left\vert\kern-0.25ex\left\vert\kern-0.25ex\left\vert #1 
    \right\vert\kern-0.25ex\right\vert\kern-0.25ex\right\vert}}
\newcommand{\GP}{\mathcal{G}\hspace{-.1cm}\mathcal{P}}

\usepackage{xcolor}

\begin{document}


\renewcommand{\baselinestretch}{2}

\markright{ \hbox{\footnotesize\rm Statistica Sinica
}\hfill\\[-13pt]
\hbox{\footnotesize\rm
}\hfill }

\markboth{\hfill{\footnotesize\rm Guo, Li and Hsing} \hfill}
{\hfill {\footnotesize\rm Variable Section in High-dimensional Functional Linear Models} \hfill}

\renewcommand{\thefootnote}{}
$\ $\par


\fontsize{12}{14pt plus.8pt minus .6pt}\selectfont \vspace{0.8pc}
\centerline{\large\bf Variable Selection and Minimax Prediction  }
\vspace{2pt} 
\centerline{\large\bf in High-dimensional Functional Linear Models}
\vspace{.4cm} 
\centerline{Xingche Guo and Yehua Li and Tailen Hsing} 
\vspace{.4cm} 
\centerline{\it University of Connecticut,  University of California Riverside}
\centerline{\it and University of Michigan}
 \vspace{.55cm} \fontsize{9}{11.5pt plus.8pt minus.6pt}\selectfont


\begin{quotation}
\noindent {\it Abstract:}
High-dimensional functional data have become increasingly prevalent in modern applications such as high-frequency financial data and neuroimaging data analysis. We investigate a class of high-dimensional linear regression models, where each predictor is a random element in an infinite-dimensional function space, and the number of functional predictors $p$ can potentially be ultra-high. Assuming that each of the unknown coefficient functions belongs to some reproducing kernel Hilbert space (RKHS), we regularize the fitting of the model by imposing a group elastic-net type of penalty on the RKHS norms of the coefficient functions. We show that our loss function is Gateaux sub-differentiable, and our functional elastic-net estimator exists uniquely in the product RKHS. Under suitable sparsity assumptions and a functional version of the irrepresentable condition, we derive a non-asymptotic tail bound for variable selection consistency of our method. Allowing the number of true functional predictors $q$ to diverge with the sample size, we also show a post-selection refined estimator can achieve the oracle minimax optimal prediction rate. The proposed methods are illustrated through simulation studies and a real-data application from the Human Connectome Project.

\par
\end{quotation}\par

\def\thefigure{\arabic{figure}}
\def\thetable{\arabic{table}}

\renewcommand{\theequation}{\thesection.\arabic{equation}}

\fontsize{12}{14pt plus.8pt minus .6pt}\selectfont

\section{Introduction}
 \vspace{-1em}
\label{s:intro}
Modern science and technology give rise to large data sets with high-frequency repeated measurements, resulting in random trajectories that can be modeled as functional data \citep{RamsaySilverman05}. There has been a large volume of literature on scalar-on-function regression models, where the most studied model is the functional linear model (FLM); see \cite{James2002, Muller2005, CaiHall06, ReissOgden2007, Crambes2009, CaiYuan2012, Lei2014, Shang-Cheng2015, Liu2021}, among others.
With functional data belonging to an infinite-dimensional function space \citep{hsing2015theoretical}, the sequence of eigenvalues of the covariance operator decays to zero, rendering the covariance operator non-invertible and hence
the inference of the FLM a challenging inverse problem. 

There has been a recent surge in applications of high-dimensional functional data analysis due to new developments in neuroimaging (e.g. fMRI and TDI), electroencephalogram (EEG), and high-frequency stock exchange data. For example, \cite{Qiao2019jasa} modeled EEG activity data from different nodes as high-dimensional functional data and proposed a functional Gaussian graphical model to study the connectivity between the nodes. \cite{lee2021conditional} considered a class of conditional functional graphical models to model the connectivity between different regions of interest (ROI) of the brain using fMRI data.

It is also natural to consider regression models with high-dimensional functional predictors. \cite{FanJamesRadchenko2015} studied variable selection procedures for linear and non-linear regression models with high-dimensional functional predictors. Their approach was to reduce the dimension of each functional predictor by representing it as a linear combination of some known basis functions and to apply a group-lasso type of penalty in model fitting. As pointed out in \cite{XueYao2021}, the results in \cite{FanJamesRadchenko2015} relied heavily on the assumption that the minimum eigenvalues of the design matrices are bounded away from zero, which ignored the infinite-dimensional nature of functional data and essentially limited their methods to functional data reside in a finite-dimensional function subspace. 
\cite{XueYao2021}, on the other hand, focused on hypothesis testing issues in high-dimensional FLMs rather than variable selection consistency.
As \cite{FanJamesRadchenko2015}, \cite{XueYao2021} also based their approach on representing functional predictors on pre-selected basis functions and minimizing a penalized least square loss function, where the group penalty can be flexibly chosen from lasso \citep{Tibshirani1996}, SCAD \citep{FanLi2001} or MCP \citep{Zhang2010aos}. To the best of our knowledge, the variable selection consistency property for the high-dimensional FLM in a general functional-data setting remains an open problem to date.

We propose to conduct variable selection in high-dimensional FLMs under the RKHS framework using a double-penalty approach, where the first penalty resembles the group-lasso type penalty in \cite{XueYao2021}, which encourages sparsity, and the second penalty is on the squared RKHS norms of the functional coefficients to regularize the smoothness of the fit. As shown in \cite{CaiYuan2012}, the RKHS approach can outperform the principal component regression approach when the coefficient functions are not directly spanned by the eigenfunctions of the functional predictors. Many of the existing high-dimensional functional regression approaches including \cite{FanJamesRadchenko2015} and \cite{XueYao2021} are similar in spirit to the principal component regression in which both the functional predictors and the coefficient functions are expressed using the same set of basis functions. Our approach offers the extra flexibility of picking the reproducing kernel based on the application and thus can outperform the existing methods when the coefficient functions are ``misaligned'' with the functional predictors as described by \cite{CaiYuan2012}.
Our double penalization method resembles a group-penalized version of the elastic-net \citep{zou2005regularization}, where the two penalties enforces sparsity and stabilizes the solution paths, respectively. It is well known that the lasso alone tends not to work well when the predictors are highly correlated, while the elastic-net may offer a more stable solution path and better prediction performance under high collinearity.

One of the main contributions of the present paper is providing a theory that addresses variable selection consistency for high-dimensional FLMs.
In the scalar case that they considered, \cite{ZouZhang2009aos} established a variable selection consistency result for the elastic-net. However, the noninvertibility of the design matrices of the functional predictors in our problem makes it necessary to create a completely new proof.
Another important contribution of our paper is that we develop the minimax optimal prediction rate for the high-dimensional FLMs, where the number of true functional predictors $q$ is allowed to grow to infinity with the sample size $n$. We show that a post-selection, refined estimation of the high-dimensional FLM using our RKHS approach can achieve such a minimax optimal rate.

\section{Functional Elastic-Net Regression}\label{sec:fnet}
 \vspace{-1em}
\subsection{Model Assumptions} \label{s:2.1}
\vspace{-0.8em}

Let $\mathbb{L}_2[0,1]$ be the $L_2$-space of square-integrable, measurable functions on $[0,1]$, equipped with the inner product $\langle f, g\rangle_{2} = \int_0^1 f(t)g(t)dt$ and functional norm $\|f\|_2=\langle f, f\rangle_2^{1/2}$, for any $f, g \in \mathbb{L}_2[0,1]$.
We will also be concerned with the $p$-fold product space of $\mathbb{L}_2^p[0,1]$ containing
elements $\boldsymbol{f}=(f_1,\ldots, f_p)^\top$ with each $f_j\in \mathbb{L}_2[0,1]$, $\|\boldsymbol{f}\|_2 \equiv (\sum_{j=1}^p \|f_j\|_2^2)^{1/2} <\infty $ and inner product 
$\langle \boldsymbol{f}, \boldsymbol{g}\rangle_{2} \equiv \sum_{j=1}^p \langle f_j, g_j\rangle_{2}$ for 
$\boldsymbol{f}=(f_1,\ldots, f_p)^\top, \boldsymbol{g} = (g_1, \ldots,g_p)^\top$. Let $\otimes$ be
the outer product associated with either inner product such that $f\otimes g$ defines an operator $(f\otimes g) h = f \langle g, h\rangle_2$. 
In this paper, we consider a high-dimensional FLM: 
\vspace{-1.5em}
\begin{align} \label{e:model} 
    Y_i =  \sum_{j=1}^p \langle X_{ij},  \beta_{j} \rangle_{2} + \varepsilon_i, \quad i = 1, \ldots, n,
\end{align}
where the functional predictors $X_{ij}(\cdot)$ are random elements in $\bbL_2[0,1]$, $\beta_{j}(\cdot)$ are unknown coefficient functions in $\mathbb{L}_2[0,1]$, and $\varepsilon_i$ are iid zero-mean random errors with variance $\sigma^2$.
\textcolor{black}{Without loss of generality, assume that both $Y_i$ and $X_{ij}(t)$ are centered at $0$, i.e., $\mathbb{E}Y_i = 0$ and $\mathbb{E}X_{ij}(t) = 0$ for $t \in [0, 1]$, $j = 1,\dots,p$, so that no intercept is needed in (\ref{e:model}).}

Consider $\BX_{i\bullet}=(X_{i1},\ldots, X_{ip})^\top$, $i=1,\ldots, n$, as iid zero-mean random vectors, with the covariance operator $\mathscrbf{C}$ defined as $\mathscrbf{C} = \E (X_{i1},\ldots, X_{ip})^\top \otimes (X_{i1},\ldots, X_{ip}).$
Note that we do not assume that the functional predictors are independent.
It is convenient to view $\mathscrbf{C}$ as a $p\times p$ operator-valued matrix 
$\{\mathcal{C}^{(j,j')}\}$ where $\mathcal{C}^{(j,j')}=\E (X_{ij}\otimes X_{ij'})$ 
is the cross covariance operators of $X_{ij}$ and $X_{ij'}$. 
Denote $\boldsymbol{Y}_n = (Y_1, \dots, Y_n)^{\top}$, 
$\boldsymbol{\varepsilon}_n = (\varepsilon_1, \dots, \varepsilon_n)^{\top}$ and 
$\BX_n=(\BX_{1\bullet}, \ldots, \BX_{n\bullet})^\top$
as the $n\times p$ matrix of functional 
predictors. Then, the sample covariance operator $\mathscrbf{C}_n$ is defined as
\vspace{-1.5em}
\begin{align} \label{e:sample_cov_op}
    \mathscrbf{C}_n = \frac{1}{n} \sum_{i=1}^n (X_{i1},\ldots, X_{ip})^\top \otimes (X_{i1},\ldots, X_{ip})  = \frac{1}{n} \boldsymbol{X}_n^{\top} \otimes \boldsymbol{X}_n.
\end{align}
We further assume that $\beta_{j}(\cdot) \in \mathbb{H}_j := \mathbb{H}(K_j)$, which is the reproducing kernel Hilbert space (RKHS) with kernel $K_j$ \citep{wahba1990spline}. Recall that a real, symmetric, square-integrable, and nonnegative definite function $K(\cdot, \cdot)$ on $[0,1]^2$ is called a reproducing kernel (RK) for a Hilbert space of functions $\mathbb{H}(K)$ on $[0,1]$ if $K(\cdot, t) \in \mathbb{H}(K)$ for any $t\in [0,1]$ and $\mathbb{H}(K)$ is equipped with the inner product such that $\langle \beta, K(\cdot, t)\rangle_{\mathbb{H}(K)} = \beta(t)$ for any $\beta\in \mathbb{H}(K)$ and any $t\in [0,1]$; the Hilbert space $\mathbb{H}(K)$ is then called an RKHS. With a proper choice of RK, an RKHS provides a flexible class of functions which can also be naturally regularized using the RKHS norm. As such, the RKHS is a useful framework in nonparametric estimation \citep{wahba1990spline} and functional data analysis \citep{CaiYuan2012, hsing2015theoretical, SunDuWangMa2018, lee2021conditional}.  

\begin{remark}
The choice of kernel $K$ determines the smoothness class. Sobolev kernels of order $m$ \citep{hsing2015theoretical} regulate the $m$-th derivative, whereas Gaussian kernels yield infinitely differentiable functions. 
In contrast, total-variation penalties, although successfully applied in scalar-on-image functional regression \citep{wang2017generalized} with the benefits of promoting piecewise structure and allowing jumps, are not induced by an RKHS norm and therefore lie outside our RKHS-based framework.
\end{remark}

We adopt the commonly assumed setting where the total number of functional predictors, $p$, can be much larger than the sample size $n$ but only a small portion of those have non-zero effects on the response.
Denote the signal set as $\mathcal{S}=\{  j \in \{1,\dots,p\}:  \var (\la X_{1j}, \beta_j \ra_2)$ $= \la \beta_{j}, \CC^{(j,j)} \beta_j \ra_2  \neq 0 \}$ and the non-signal set as $\calS^c = \{1, \dots, p\} \backslash \calS$, and write $q:= | \calS |$.

 \vspace{-1em}
\subsection{Functional Elastic-Net Based on RKHS} \label{s:2.2}
\vspace{-0.8em}

In order to regularize the solution as well as to enforce sparsity in $\bdbeta=(\beta_1, \ldots, \beta_p)^\top$, we assume $\bdbeta\in \mathbb{H} := \otimes_{j=1}^p \mathbb{H}_j$, which is the direct product of the RKHS \citep{hsing2015theoretical}, and estimate it by
\vspace{-0.8em}
\begin{align}
\label{equ:mini0}
    \wh \bdbeta=\argmin_{\bdbeta \in \mathbb{H}} \bigg\{ \frac{1}{2n} \sum_{i=1}^n \bigg(Y_i -  
\sum_{j=1}^p \langle X_{ij},  \beta_j \rangle_{2}  \bigg)^2 
+  
    \sum_{j=1}^p \mathrm{Pen}(\beta_j; \bdlambda)\bigg\} 
\end{align}
where 
$\mathrm{Pen}(\beta_j; \bdlambda)$ 
is the functional elastic-net penalty to be specified below with  $\bdlambda$ denoting a vector of tuning parameters. 

Following \cite{CaiYuan2012}, for any symmetric positive semi-definite kernel $R(\cdot, \cdot)$, denote $\mathscr{L}_R$ as the integral operator $(\scrL_R f)(\cdot)= \int_0^1 R(s, \cdot) f(s) ds$, $f\in \mathbb{L}_2[0,1]$. Suppose $R$ has a spectral decomposition $R(s,t)=\sum_{j=1}^\infty \theta_j^R\varphi_j^R(s) \varphi_j^R(t)$. Then its square root is defined as 
$R^{1/2}(s,t)=\sum_{j=1}^\infty (\theta_j^R)^{1/2} \varphi_j^R(s) \varphi_j^R (t)$, and  
$\mathscr{L}_{R^{1/2}}$ is the associated square-root integral operator. For a matrix of kernel functions $\BR=(R_{ij})_{i, j=1}^{k, m}$, let $\mathscrbf{L}_{\BR}: \mathbb{L}_2^m \to \mathbb{L}_2^k $ be the corresponding matrix of operators such that $\mathscrbf{L}_{\BR} \bdf =\left(\sum_{j=1}^m \scrL_{R_{ij}} f_j \right)_{i=1}^k$ for any $\bdf=(f_1, \ldots, f_m)^\top \in \mathbb{L}_2^m$. 
By \cite{wahba1990spline} and \cite{CaiYuan2012}, for any positive semi-definite 
kernel $K$ and any $\beta \in \mathbb{H}(K)$, there exists an $f \in \mathbb{L}_2[0,1]$ such that $\beta = \mathscr{L}_{K^{1/2}} f$. 
If $K$ is not strictly positive definte, then multiple $f$'s satisfy this relationship. However, there is always a unique $f$ satisfying $\|\beta\|_{\mathbb{H}(K)} = \|f\|_2$. The ridge regularization term in our objective (introduced later) guarantees the 
identifiability of this representative.
Without causing any confusion, we use $\|\cdot\|_2$ to denote the norm of $\mathbb{L}_2$ functions or vectors of $\mathbb{L}_2$ functions as well as the Euclidean norm in $\mathbb{R}^p$.

Let $\beta_{j}= \mathscr{L}_{K_j^{1/2}} f_{j}$ for all $j$ and denote $\bdf=(f_1, \ldots, f_p)^\top$. Then $\bdbeta= \scrbfL_{\BK^{1/2}} \bdf$ where $\boldsymbol{K}(s,t)=\mathrm{diag}(K_1, \dots, K_p)(s,t)$. Define
$\widetilde{X}_{ij} = \mathscr{L}_{K_j^{1/2}} X_{ij}$, $\wt \BX_{i \bullet}= (\wt X_{i1}, \ldots, \wt X_{ip})^\top$, and $\wt \BX_n=(\wt \BX_{1\bullet}, \ldots, \wt \BX_{n\bullet})^\top$. Thus, the theoretical and empirical covariance of $\wt \BX_{i\bullet}$ 
are $\scrbfT=\cov(\wt \BX_{i\bullet})= \mathscrbf{L}_{\boldsymbol{K}^{1/2}}\boldsymbol{\mathscrbf{C}}
\mathscrbf{L}_{\boldsymbol{K}^{1/2}}$ and 
$\Tn= \mathscrbf{L}_{\boldsymbol{K}^{1/2}}\boldsymbol{\mathscrbf{C}}_n 
\mathscrbf{L}_{\boldsymbol{K}^{1/2}} = n^{-1} \wt \BX_n^\top \otimes \wt \BX_n.$
Define $\mathbb{M}_{nj} = \mathrm{Span}\big\{ \widetilde{X}_{ij}(\cdot), i = 1,\dots,n \big\}$ and $\mathbb{M}_{nj}^{\perp}$ the orthogonal complement of $\mathbb{M}_{nj}$. 
With the above $\bbL_2$ representation $\bdf$ of $\bdbeta$, the loss function in (\ref{equ:mini0}) can be rewritten as
\vspace{-0.8em}
\begin{align} \label{equ:mini1}
	\ell(\boldsymbol{f})
	=\frac{1}{2} \langle  \mathscrbf{T}_n {\boldsymbol{f}}, \boldsymbol{f}  \rangle_{2} 
	- \left\langle {1\over n} \wt \BX_n^\top \BY_n,   \boldsymbol{f}  \right\rangle_{2}
	+\frac{1}{2n} \|\BY_n\|_2^2 +  \sum_{j=1}^p \mathrm{Pen}(f_j; \bdlambda).
\end{align}
We propose to use the following functional elastic-net penalty
\vspace{-0.6em}
\begin{align*}
	\mathrm{Pen}(f_j; \lambda_1, \lambda_2)= \lambda_1  \|\Psi_j f_j \|_{2} + \frac{\lambda_2}{2} \|f_j\|_{2}^2,  \quad \lambda_1,\lambda_2 >0,
\end{align*}
where $\Psi_j$ is an operator on $\mathbb{L}_2[0,1]$ satisfying the following condition. 
%
\begin{assum}
For $j=1,\dots,p$, $\Psi_j$ is a self-adjoint operator such that $\Psi_j f \in \mathbb{M}_{nj}$ for all $f \in \mathbb{M}_{nj}$. 
Assume that there exist positive constants $0<C_{\min} <C_{\max}< \infty$ such that, uniformly for all $j$,  the eigenvalues of $\Psi_j$ are in the interval $[C_{\min},C_{\max}]$.
\label{ass:a1}
\end{assum}
%
%
\begin{remark}
(i) The $\bbL_2$-norm $\|f_j\|_2$ in $\mathrm{Pen}(f_j; \lambda_1, \lambda_2)$ corresponds to the RKHS norm $\|\beta_j\|_{\bbH_j}$, a commonly used norm in functional regression problems \citep[cf.][]{CaiYuan2012}.

\noindent (ii) A simple choice for $\Psi_j$ is $\Psi_j = \mathscr{I}$, the identity operator, based on which the penalty $\mathrm{Pen}(f_j; \lambda_1, \lambda_2)$ includes both $\|f_j\|_2$ and $\|f_j\|_2^2$ and resembles an elastic-net \citep[cf.][]{zou2005regularization} version of the group lasso \citep{yuan2006model}. In the high-dimensional functional regression setting, \cite{XueYao2021} considered a penalty that focused on the amount of variation $X_j$ explains rather than the norm of $f_j$. Their penalty translates in our setting to  $\lambda_1 n^{-1/2} ( \sum_{i=1}^n  \langle X_{ij}, \beta_j\rangle_2^2 )^{1/2}  = \lambda_1 \| \{\Tnjj\}^{1/2} f_j\|_2$ where $\Tnjj$ is the empirical covariance of $\wt \BX_{\bullet j}= (\wt X_{1j}, \cdots, \wt X_{nj})^{\top}$ or the $(j,j)$th entry of $\Tn$. The approach in \cite{XueYao2021} does not penalize the squared norm, but both $X_j$ and $\beta_j$ are represented by a growing but finite number of basis functions, which effectively sets a lower bound on the smallest eigenvalue of $\Tnjj$. In our setting, we can achieve similar effects by setting $\Psi_j = (\Tnjj + \theta \mathscr{I})^{1/2}$, where $\theta>0$ provides a floor to the smallest eigenvalue of $\Psi_j$ and is treated as a tuning parameter. 
\end{remark}


Note that the functional estimator, $\widehat{\boldsymbol{f}}$, is defined as the solution that minimizes \eqref{equ:mini1} over an infinite-dimensional space $\mathbb{L}_2^p[0,1]$. The following proposition establishes that the minimization problem is indeed well defined and any minimizer must be in a finite-dimensional subspace. 

\begin{proposition}\label{lemma:solution_form}
Suppose that Condition C.\ref{ass:a1} holds. Then, for each $j=1,\ldots,p$, any minimizer $\widehat{f_j}$ of \eqref{equ:mini1} must be in the space $\mathbb{M}_{nj}$.
\end{proposition}
%
The proof of Proposition \ref{lemma:solution_form} uses the ideas of the well-known representer theorem for smoothing splines \citep{wahba1990spline}.
The fact that the minimizer of \eqref{equ:mini1} is in a finite-dimensional subspace allows us to establish its uniqueness in Proposition \ref{theorem:kkt} below.

Next, we develop the convex programming conditions in the functional space that characterize the optimizer of \eqref{equ:mini1}. 
For the classical lasso problem \citep{Tibshirani1996}, the Karush-Kuhn-Tucker (KKT) condition is used to characterize the solution \citep[cf.][]{zhao2006model,wainwright2009sharp}, where
subgradients are used in place of gradients due to the nondifferentiability of the lasso objective function. Similarly, in the function space, the objective function \eqref{equ:mini1} is not always differentiable because of the group-lasso-type penalty on $\| \Psi_j f_j \|_2$.
In Section \ref{s:KKT}, we review the definition of Gateaux differentiability and define the corresponding notion of sub-differential. With these in mind, we state the following result.

\begin{proposition}
\label{theorem:kkt}
Let $\bdbeta_0$ be the true value of $\bdbeta$ in Model  (\ref{e:model}), and $\boldsymbol{f}_0 = (f_{01}, \dots, f_{0p})^{\top}$ be the corresponding $\mathbb{L}_2^p$ surrogate such that $\bdbeta_0=\scrbfL_{\BK^{1/2}} \bdf_0$. 
Suppose Condition C.\ref{ass:a1} holds. Then, for all $\lambda_1, \lambda_2 > 0$, the solution $\widehat{\boldsymbol{f}}$ for \eqref{equ:mini1}
exists uniquely and satisfies 
\vspace{-1em}
\begin{align}
\label{equ:kkt}
	 \mathscrbf{T}_n( \widehat{\boldsymbol{f}} - \boldsymbol{f}_0 ) -  \boldsymbol{g}_n +     \lambda_2 \widehat{\boldsymbol{f}} +    \lambda_1  \boldsymbol{\omega} = 0,
\end{align}
where $\bdg_n= n^{-1} \wt \BX_n^\top \bdepsilon_n$, and $\omega_j =\frac{\Psi_j^2 \widehat{f}_j}{\| \Psi_j \widehat{f}_j \|_2}$ if $\widehat{f}_j\not =0$ and $\omega_j = \Psi_j\eta_j$ for some  $\eta_j$ with $\|\eta_j\|_2\le 1$ if $\widehat{f}_j=0$.
\end{proposition}
Equation \eqref{equ:kkt} is referred to as the functional KKT condition for 
the optimization problem \eqref{equ:mini1} and plays a central role in establishing Theorem \ref{thm:1} .

 %
 %
 %
 %
 %
 %
 %
 %
 %

 \vspace{-2em}
\section{Theoretical Results} \label{sec:theory}

 \vspace{-1.5em}
\subsection{Consistency property of variable selection} \label{sec:vs_theory}
In this section, we establish the consistency property of variable selection using our approach.
Even though the normality assumption is not essential to our methodology, 
in order to get sharp results that are comparable with those in
the literature, we assume that the rows of $\boldsymbol{X}_{i\bullet}$, $i=1,\ldots, n$, are iid zero-mean Gaussian 
random vectors with each element lies in $\mathbb{L}_2[0,1]$, and $\varepsilon_i \overset{iid}{\sim} \mathcal{N}(0, \sigma^2)$.
Note that Gaussianity is invoked only to obtain exponential concentration; 
our current theory does not cover genuinely heavy-tailed designs or errors.
Recall the definitions of $\mathcal{S}$ and $\widehat{\bdf}=(\wh f_1, \ldots, \wh f_p)^\top$ in Sections \ref{s:2.1} and \ref{s:2.2}, respectively, and define $\widehat{\mathcal{S}}=\big\{  j \in \{1,\dots,p\}:  \widehat{f}_{j} \neq 0 \big\}$. Then, variable selection consistency is achieved when $\widehat{\mathcal{S}} = \mathcal{S}$.

%

%

We collect here some notation used throughout the paper. 
Let $\mathbb{H}_1$ and $\mathbb{H}_2$ be two Hilbert spaces and $\CA: \mathbb{H}_1 \to \mathbb{H}_2$ be a compact linear operator mapping from $\mathbb{H}_1$ to $\mathbb{H}_2$. Then the $\mathbb{L}_2$ operator norm is defined as $\|\CA\|_2= \sup_{f\in \mathbb{H}_1} \|\CA f\|_2/ \|f\|_2$ which is the maximum singular value of $\CA$; if $\mathbb{H}_1=\mathbb{H}_2$ and $\CA$ is self-adjoint, the trace of $\CA$ is $\tr(\CA)=\sum_{j\ge 1} \Lambda_j(\CA)$, which is the sum of all eigenvalues.
For any $\boldsymbol{f} \in \mathbb{L}_2^p[0,1]$,  $\| \boldsymbol{f} \|_{\infty} := \max_j \| f_j \|_{2}$; 
for any $r \times s$ operator-valued matrix $\mathscrbf{A}=( \CA_{ij})_{i,j=1}^{r,s}$, where each $\CA_{ij}$ maps from $\mathbb{L}_2[0,1]$ to $\mathbb{L}_2[0,1]$, 
 define the norm $\vertiii{\mathscrbf{A}}_{a,b} := \sup_{\|\boldsymbol{f}\|_{a} \le 1} \|\mathscrbf{A} \boldsymbol{f} \|_{b}$ for $a, b \in \{2, \infty\}$. 
For any index sets $\calS_1$ and $\calS_2$, $\mathscrbf{A}^{(\calS_1, \calS_2)}$ is the submatrix of $\mathscrbf{A}$ with rows in $\calS_1$ and columns in $\calS_2$. This notation is used for matrices of operators, such as $\scrbfC$, $\scrbfT$, and $\Tn$. 
%
%
Consistent with this notation, $\mathscr{T}^{(j,j)}=\cov(X_j)$ is the $j$th diagonal element of $\scrbfT$, and define $\scrT^{(j,j)}_\lambda=\scrT^{(j,j)} + \lambda \scrI$ for any $\lambda>0$ where $\mathscrbf{I}$ is the identity operator. 
Let $\QSS=\diag\{\mathscr{T}^{(j,j)}, j \in \calS\}$ be the operator-valued matrix that only contains the diagonal terms of $\TSS$, and let $\QSS_{\lambda} = \QSS + \lambda \mathscrbf{I}$. 

In addition to Condition C.\ref{ass:a1}, we need the following conditions. 

\begin{assum}
Each $\mathscr{T}^{(j,j)}$ is standardized such that $\|\scrT^{(j,j)}\|_2=1$, 
with its trace uniformly bounded by a finite constant $\tau$, i.e.,
$\sup_{j \in \{1, \dots, p\}} \tr(\mathscr{T}^{(j,j)})  \le \tau$.
\label{ass:a2}
\end{assum}
\begin{assum}
Define
$\varkappa(\lambda_2) := \vertiii{ \TSS (\TSS_{\lambda_2})^{-1} }_{\infty, \infty}$.
Assume that for some $\gamma\in (0,1]$, we have
$\varkappa(\lambda_2) \cdot \vertiii{ \TScS (\TSS)^{-} }_{\infty, \infty} \le (C_{\min} / C_{\max}) (1 - \gamma)$,
where $(\TSS)^{-}$ is the Moore-Penrose generalized inverse of $\TSS$.
\label{ass:a3}
\end{assum}



\begin{assum}
$\aleph(\lambda_2):= \vertiii{ \left( \mathscrbf{T}^{(\calS, \calS)} - \mathscrbf{Q}^{(\calS, \calS)} \right)  \big(\mathscrbf{Q}^{(\calS, \calS)}_{\lambda_2}\big)^{-1} }_{\infty, \infty} < 1$.
\label{ass:a6}
\end{assum}
%
\noindent
\begin{remark}
(i)
Condition C.\ref{ass:a2} places a mild constraint on the decay rate of the eigenvalues for $\mathscr{T}^{(j,j)}$, which is equivalent to $\sup_{j \in \{1, \dots, p\}} \E \|\widetilde X_j\|_2^2 \le \tau$.

\noindent (ii)
Condition C.\ref{ass:a3} controls the correlation between functional predictors in the true signal set $\calS$ and those in the non-signal set $\calS^c$. This assumption is related to 
the so-called ``irrepresentable condition'' on model selection consistency of the classical lasso \citep{zhao2006model, wainwright2009sharp}, the classical elastic-net \citep{JiaYu2010}, and the sparse additive models \citep{ravikumar2009sparse}. Condition
C.\ref{ass:a3} becomes harder to fulfill when $\varkappa(\lambda_2)$ is large  or when $C_{\min}/C_{\max}$ is small. However, when the predictors in $\calS$ and in $\calS^c$ are uncorrelated, then $\vertiii{ \TScS (\TSS)^{-} }_{\infty, \infty}=0$ and the assumption holds trivially.

\noindent(iii)
Condition C.\ref{ass:a6} puts constraints on the correlations between the predictors in the true signal set $\calS$, so that none of the true predictors can be represented by other predictors in $\calS$. When the predictors in $\calS$ are uncorrelated, then $\aleph(\lambda_2) =0$ and C.\ref{ass:a6} trivially holds.
\end{remark}

To gain a deeper understanding of Conditions C.\ref{ass:a2}-C.\ref{ass:a6}, an example will be provided in Section \ref{s:partially_separable} where the functional predictors have a partially separable covariance structure \citep{Zapata2021}.
To state the variable selection consistency properties of our approach, we further assume without loss of generality that $\| \fzeros \|_{\infty} = 1$ below. 
Also, the symbol $D^*$ and similar symbols below will denote universal constants in $(0, \infty)$ that arise from inequalities, whose values change from line to line but do not depend on the model parameters, sample size, or regularization parameters. The specific expressions of universal constants may be complicated and do not add to the understanding of the results. With these in mind, define the following conditions on $\lambda_1,\lambda_2$:
\vspace{-1em}
\begin{align}
        &\lambda_1/\lambda_2 > \left({3\over\gamma} - 2\right) C_{\max}^{-1}, \quad D_{1,1}^* >  \lambda_1 > D_{1,2}^* { \tau^{1/2} (1+\sigma)  \over C_{\min} \gamma} \sqrt{{\log(p-q) \over n}}, \nonumber \\
    & D_{2,1}^* > \lambda_2 > D_{2,2}^* \frac{\tau (1+\sigma) (\rho_1 + 1)}{(C_{\min} / C_{\max})^2 \gamma^2}  \max\left({q\log(p-q)\over n},\sqrt{q^2\over n} \right).
\label{equ:lambda_12_bound_1}
\end{align}
where $\rho_1$ denotes the largest eigenvalue of $\TSS$ and $D_{1,1}^*, D_{1,2}^*, D_{2,1}^*, D_{2,2}^*$ are universal constants.
It is worth emphasizing that by carefully separating the model/regularization parameters with universal constants, our nonasymptotic results below can be readily used to state asymptotic results for which some or all of the parameters could change with $n$. An example of that is provided in Corollary \ref{cor:R(f)} below.

Finally, define the signal set containing predictors with ``substantial" predictive power $\calS_G:=\{j\in\calS:\left\| (\Tjj)^{1/2} f_{0j} \right\|_2>G\}$,
where $G\in (0,\infty)$; recall $\| (\Tjj)^{1/2} f_{0j}\|_2^2 = \E\langle X_j,\beta_j\rangle_2^2$.
The variable selection consistency of our functional elastic-net approach is given in the following result. 

\begin{theorem} \label{thm:1}
Consider the functional elastic-net problem \eqref{equ:mini1}.
Suppose that Conditions C.\ref{ass:a1}-C.\ref{ass:a3} and \eqref{equ:lambda_12_bound_1} hold. 
Then $\widehat{\calS}$ exists uniquely, and (i) and (ii) below 
hold with probability at least
\vspace{-1em}
\begin{align}
    1 -   \exp\left( -D{\lambda_2^2 n \over q } \right), \quad \mbox{where} \quad D = D^* \left((C_{\min}/C_{\max}) \gamma \over \tau^{1/2} (\rho_1+1) (\sigma + 1)  \right)^2,
\label{equ:vs_consistency_rate}
\end{align}
for some universal constant $D^*$.

\noindent (i) 
    The estimated signal set is contained in the true signal set, i.e. $\widehat{\calS} \subset \calS$. 
    (ii) 
    Under the additional Condition C.\ref{ass:a6}, we have $\widehat\calS\supset\calS_G$ for 
    \vspace{-0.8em}
\begin{align*}
    G = \frac{12 - 8\aleph(\lambda_2)}{1-\aleph(\lambda_2)} \left( C_{\max} \sqrt{ \lambda_1^2/\lambda_2} + 2 \sqrt{\lambda_2} \right),
\end{align*}
and, in particular, if $\calS_G=\calS$, then $\widehat\calS=\calS$ and variable selection consistency is achieved. 
\end{theorem}
%
%
%
%
%
%
%
%



\begin{remark}
(i)
Part (i) of Theorem \ref{thm:1} guarantees a sparse solution for the functional elastic-net where all predictors in the non-signal set are eliminated.
By examining \eqref{equ:lambda_12_bound_1} and \eqref{equ:vs_consistency_rate}, we can see that increasing $\lambda_2$ (and, consequently, $\lambda_1$) leads to a higher probability of eliminating the non-signals. 
Condition \eqref{equ:lambda_12_bound_1} also implies that, as the correlation of predictors between the signal and non-signal sets increases (i.e., decreasing value of $\gamma$), larger values of $\lambda_1,\lambda_2,\lambda_1 / \lambda_2$ are required.
Moreover, larger values of $\gamma$, smaller values of $\tau$, and reduced $\sigma^2$ (resulting in a decreased correlation between $\calS$ and $\calS^c$, faster eigenvalue decay for each $\Tjj$, and a higher signal-to-noise ratio, respectively) enhance the functional elastic-net's ability to accurately identify the signal set.

\noindent (ii)
Part (ii) of Theorem \ref{thm:1} provides conditions that prevent the functional elastic-net from removing the true signals and thus guarantees that the predictors identified by the functional elastic-net are not overly sparse.
Large values of $\lambda_1$, $\lambda_1 / \lambda_2$, and $\aleph(\lambda_2)$ result in a larger gap $G$, making signal detection more challenging.
This is understandable because a large sparsity penalty can lead to the removal of true signals, especially when there is a strong correlation.

\noindent (iii) 
Condition \eqref{equ:lambda_12_bound_1} requires that the lower bound of $\lambda_1$ must be of the rate $\sqrt{\frac{\log(p-q)}{n}}$ to control sparsity. This is similar to the lower bound of the regularization parameter of the lasso \citep[see\@ Theorem 3 of][]{wainwright2009sharp}. Our theory also requires a lower bound for $\lambda_2$ to control both the smoothness and variance of $\widehat{f}_j$. The roles of $\lambda_2$ in functional linear regression have been discussed by many \citep[see, e.g.,][]{CaiYuan2012}. The classical (finite-dimensional) elastic-net optimization \citep{zou2005regularization} includes lasso as a special case, with $\lambda_2=0$. However, this is not feasible in the infinite-dimensional functional setting.
To understand it, consider 
classical high-dimensional data (in the scalar setting) and let $\boldsymbol{\Sigma}_{\calS}$ be the $q \times q$ covariance matrix of the true predictors. 
A common assumption to avoid collinearity in that setting is to bound the minimum eigenvalue of $\boldsymbol{\Sigma}_{\calS}$ away from zero \citep{zhao2006model, wainwright2009sharp}, which is why $\lambda_2$ could be taken as zero. We cannot bound the eigenvalues of $\TSS$ that way in the functional setting because it contradicts the intrinsic infinite dimensionality of functional data; in fact, the sequence of eigenvalues for $\TSS$ shrinks to zero even if all the predictors in $\calS$ are uncorrelated. 
\end{remark}

Following \cite{CaiYuan2012}, we also study the excess risk as a metric to measure the prediction accuracy of the estimator $\mathcal{R}\left( \boldsymbol{f}  \right)
	 = \mathbb{E}\left( \sum_{j=1}^p \langle {\wt X_{j}}^*,  f_{0j} - f_j \rangle_{2}\right)^{2},$
%
where $\boldsymbol{\wt X}^*_{\bullet}$ is a copy of $\boldsymbol{\wt X}_{i\bullet}$. 
The excess prediction risk of our estimator, $\widehat{\boldsymbol{f}}$, is obtained by plugging $\widehat{\boldsymbol{f}}$ in $\mathcal{R}\left( \boldsymbol{f} \right)$. 
The following result describes the excess prediction risk of the functional elastic-net estimator.

\begin{theorem} \label{c:excess}
Assume that Conditions C.\ref{ass:a1}-C.\ref{ass:a3} and  \eqref{equ:lambda_12_bound_1} hold. 
Then, the excess risk satisfies $\mathcal{R}(\widehat{\boldsymbol{f}}) <  q \left(4 C_{\max} \lambda_1 + 4 \lambda_2 + C_{\max}^2 \lambda_1^2 / \lambda_2 \right) $ with probability bounded below by the expression in \eqref{equ:vs_consistency_rate}.
\end{theorem}



Next, we discuss asymptotic results readily derived from Theorems \ref{thm:1} and \ref{c:excess} by allowing $p,q$ as well as the model/regularization parameters to vary with the sample size $n$. 
To facilitate the discussion, denote $a_k \asymp b_k$ for two positive sequences $\{a_k\}_{k=1}^\infty$ and $\{b_k\}_{k=1}^\infty$, if $c_1 < a_k / b_k < c_2$ for some $0 < c_1 < c_2 < \infty$ and for all $k$. 
The following corollary is a direct result of Theorem \ref{c:excess}, the proof of which is in the Supplementary Material.
\begin{corollary} \label{cor:R(f)}
Assume that Conditions C.\ref{ass:a1}-C.\ref{ass:a3} and \eqref{equ:lambda_12_bound_1} hold, where
$C_{\min}$ and $\gamma$ are bounded away from $0$, and $\rho_1$, $\sigma^2$, $\tau$, and $C_{\max}$ bounded away from $\infty$.
Let 
$
\alpha(p,q,n) := \max\left(q, \sqrt{\log (p-q)}, \sqrt{q\log n}\right)
$
and assume that $q \alpha(p,q, n) = o(n^{1/2})$.
Then, for some sufficiently large constant $D$, the probability that $\mathcal{R}\left(\widehat{\boldsymbol{f}}\right) > D n^{-1/2} q \alpha(p,q,n)$ infinitely often is $0$.
\end{corollary}
%
%
%
%
%
%
\begin{remark}
    Consider a high dimension FLM setting where $q\asymp n^{\varsigma}$ for some $0<\varsigma<1/4$, and suppose all functional predictor in the signal set have about the same contribution to the variation of the response such that $G= \min_{j\in \calS} \| (\Tjj)^{1/2} f_{0j}\|_2 \asymp 1/\sqrt{q}$. By Theorem \ref{thm:1} (ii), we can choose $\lambda_1 \asymp \lambda_2 \asymp (1/q)$ to guarantee recovery of the signal set $\calS_G$. Condition \eqref{equ:lambda_12_bound_1} is also satisfied if $\log p =O (n^{1-2\varsigma})$, which is an ultra-high-dimensional FLM setting. Under this setting and with the choice of tuning parameters described above, the probability bound in \eqref{equ:vs_consistency_rate} goes to 1 which ensures variable selection consistency; the condition $q \alpha(p,q, n) = o(n^{1/2})$ in Corollary \ref{cor:R(f)} is also satisfied, and we can conclude $\mathcal{R}\left(\widehat{\boldsymbol{f}}\right) \to 0$ almost surely.  
\end{remark}

%
%
%
%
%
%
%
%
%
%
%
%

\vspace{-1em}
\subsection{Oracle minimax optimal rate and a post-selection refined estimator}\label{sec:minimax}
\vspace{-1em}

\cite{CaiYuan2012} established the minimax lower bound of the excess prediction risk for univariate FLM with $q=1$. Such a lower bound is yet to be established for high-dimensional FLMs. In this subsection, we first investigate the minimax lower bound of the excess prediction risk under the oracle model, where $\calS$ is known and the true number of functional predictors $q$ is allowed to diverge with the sample size $n$. 
We need the following conditions for our results.
\begin{assum}
For each $j\in {\cal S}$, the $k$-th eigenvalue of ${\cal T}^{(j,j)}$ is bounded by $ck^{-2r}$ for some $c\in (0,\infty)$ and $r>1/2$. For some $b \in (0, \infty)$, 
\vspace{-1em}
\begin{align}
    \sup_{\alpha > 0} \vertiii{ \left(\QSS_{\alpha} \right)^{-1/2} \TSS_{\alpha} \left(\QSS_{\alpha} \right)^{-1/2} }_{2,2} \le b  \label{equ:bound_corr}
\end{align}

\label{ass:a8}
\end{assum}

\vspace{-1em}

Condition C.\ref{ass:a8} requires that the eigenvalues of each ${\cal T}^{(j,j)}$, $j\in {\cal S}$, to decay in a polynomial rate, which is the same assumption made in \cite{CaiYuan2012}. By requiring $r>1/2$, each ${\cal T}^{(j,j)}$ is a linear operator that belongs to the trace class, which includes the Hilbert-Schmidt operators. Condition C.\ref{ass:a8} trivially holds when $\TSS = \QSS$ meaning that the functional predictors are uncorrelated. 
When the functional predictors have a partially separable covariance structure, \eqref{equ:bound_corr} holds in mild conditions (see supplementary materials). The following proposition and its corollary further illustrate what Condition C.\ref{ass:a8} entails.
\begin{proposition}
\label{theorem:eigenvalue_rela}
Assume \eqref{equ:bound_corr} holds, we have $\Lambda_k(\TSS) \le b \Lambda_k(\QSS)$, where $\Lambda_k(\TSS)$ and $\Lambda_k(\QSS)$ denote the $k$-th largest eigenvalues of $\TSS$ and $\QSS$, respectively.
\end{proposition}

\begin{corollary}
\label{corollary:eigenvalue_rela}
Assume Condition C.\ref{ass:a8} holds, let $\{\rho_l =\Lambda_l (\TSS)\}_{l \ge 1}$ be the eigenvalues of $\TSS$ in a decreasing order, then $\rho_{q(k-1)+j} \le bc \cdot k^{-2r}$ for any $k \ge 1$ and $j=1, \dots, q$.
\end{corollary}
%
%
%
%
%
%
Corollary \ref{corollary:eigenvalue_rela} is a direct result of Proposition \ref{theorem:eigenvalue_rela} and is essential in deriving the minimax lower bound in the following theorem.


\begin{theorem} \label{thm:minimax_lower}
Let $\mathscrbf{P}(r)$ be the class of covariance operators satisfying Condition C.\ref{ass:a8}. 
Then   
\vspace{-1.5em}
\begin{align*}
    \lim_{a \rightarrow 0} \lim_{n \rightarrow \infty} \inf_{ \wt f_\calS } \ \sup_{\TSS \in \mathscrbf{P}(r)} \ \sup_{\fzeros \in \mathbb{{L}}_2^q }  \mathbb{P}\left( \mathcal{R}( \wt f_\calS ) \ge a (n / q)^{-\frac{2r}{2r+1}}  \right) = 1,
\end{align*}
where the infimum is taken over all possible predictors $\wt f_\calS$ based on the training data $\{(\BX_{i\calS}, Y_i), i =1, \dots, n \}$.
\end{theorem}

Theorem \ref{thm:minimax_lower} provides the oracle minimax lower bound for the excess prediction risk of the high dimensional FLM, which reduces to the lower bound of \cite{CaiYuan2012} if $q=1$. 
By comparing this result with Corollary \ref{cor:R(f)}, we can see that the excess risk of the functional elastic-net, $\mathcal{R}( \wh{\boldsymbol{f}} )$, is at a rate slower than $(n/q)^{-1/2}$, which in turn is slower than the oracle minimax rate in Theorem \ref{thm:minimax_lower} when $r > 1/2$. 
This is understandable, since the primary goal of functional elastic-net is to perform variable selection.
Suppose all assumptions in Theorem \ref{thm:1} hold and $\mathcal{S} = \mathcal{S}_G$, the functional elastic-net estimator enjoys variable selection consistency and can help us find an estimated signal set $\wh \calS$ that satisfies the following condition.
\begin{assum}
$\lim_{n \rightarrow \infty}\sup_{\TSS \in \mathscrbf{P}(r)} \sup_{\fzeros \in \mathbb{L}_2^q } \mathbb{P}\left(\widehat{\calS}\not=\calS\right) = 0.$
\label{ass:a10}
\end{assum}
%
%
%
This motivates us to refine our FLM estimator within the selected signal set with the goal of improving the excess prediction risk,
\vspace{-1em}
\begin{align}
    \fhaths =\argmin_{f_j \in \mathbb{L}_2} \left\{ \frac{1}{n} \sum_{i = 1}^n \left(Y_i -  
\sum_{j \in { \wh\calS}} \langle \wt{X}_{ij},  f_j \rangle_{2}  \right)^2 
+  
    \lambda_3 \sum_{j \in { \wh\calS}} \| f_j \|_2^2  \right\}. \label{equ:object_ridge} 
\end{align}
The refined estimator \eqref{equ:object_ridge} is a special case of the functional elastic-net estimator in Section \ref{s:2.2} by including functional predictors in $\wh \calS$ only and setting the $\ell_1$ penalty to 0, as the focus has shifted away from variable selection. As such, $\fhaths$ can be calculated the same way as the functional elastic-net with a minimum modification to the algorithm. 


\begin{theorem} \label{thm:minimax_upper}
Assume Conditions C.\ref{ass:a8}-C.\ref{ass:a10} hold, and the number of true signals satisfies $q = o \left(n^{\frac{2r-1}{4r}} \right)$.  Then 
\vspace{-1em}
\begin{align*}
    \lim_{A \rightarrow \infty} \lim_{n \rightarrow \infty} \ \sup_{\TSS \in \mathscrbf{P}(r)} \ \sup_{\fzeros \in \mathbb{{L}}_2^q }  \mathbb{P}\left( \mathcal{R}( \fhaths ) \ge A (n/q)^{-\frac{2r}{2r+1}} \right) = 0,
\end{align*}
provided that $\lambda_3 \asymp (n/q)^{-2r/(2r+1)}$. 
\end{theorem}

Theorem \ref{thm:minimax_upper} shows that our refined estimator \eqref{equ:object_ridge} achieves the oracle the minimax rate in Theorem \ref{thm:minimax_lower}, which is determined by the rate of decay of the eigenvalues of the operator $\TSS$. 
When $q$ is a constant that does not grow with $n$, the minimax rate for the excess risk is on the order of $n^{-2r/(2r+1)}$, consistent with the findings in \cite{CaiYuan2012}. 
\section{Implementation and Numerical Studies}\label{sec:numerical_study}

 %
 %
 %
 %
 %
 %
 %
 %
 %
 %
 %
 %
 %
 %
 
\vspace{-1em}
\subsection{Practical Implementation}\label{sec:algorithm}
\vspace{-1em}

Proposition \ref{lemma:solution_form} provides an expression for the exact solution to the optimization problem \eqref{equ:mini1}, where each $\widehat{f}_j$ is a linear combination of $\wt \BX_{\bullet j}$. However, such a solution is not scalable to big data and ultra-high dimensions, since there are a total of $np$ parameters to estimate. In this subsection, we propose a computationally-efficient algorithm to fit the model based on the idea of reduced-rank approximations, which has been widely used in semiparametric regression \citep{ruppert2003semiparametric} and spline smoothing \citep{ma2015efficient}. Our low-rank approximation shares a similar spirit as the eigensystem truncation approach proposed by \cite{xu2021low} for a low-rank approximation of smoothing splines.

Since $\wh f_j$ falls in the subspace spanned by $\wt \BX_{\bullet j}$, it can be well approximated by the eigenfunctions of $\Tnjj$, which is the empirical covariance of $\wt \BX_{\bullet j}$. Let $\boldsymbol{\varphi}_j(t) = (\varphi_{j1}, \dots, \varphi_{j M_j})^{\top}(t)$ be the first $M_j$ eigenfunctions of $\Tnjj$, such that $\int_0^1 \boldsymbol{\varphi}_j(t) \boldsymbol{\varphi}_j^{\top}(t) dt = \boldsymbol{I}_{M_j}$, and we approximate $f_j$ with $\widetilde{f}_j(t) = \boldsymbol{\varphi}_j^{\top}(t) \boldsymbol{c}_j$. 
Define $\boldsymbol{\Gamma}_j = \int_0^1 \wt \BX_{\bullet j}(t) \boldsymbol{\varphi}_j^{\top}(t) dt$ and
$\boldsymbol{H}_j = \int_0^1 (\Psi_j \boldsymbol{\varphi}_j)(t) (\Psi_j \boldsymbol{\varphi}_j)^{\top}(t) dt$. 
We reparameterize the coefficient vectors as $\boldsymbol{d}_j = \boldsymbol{H}_j^{1/2} \boldsymbol{c}_j$, and solve the group elastic-net problem (\ref{equ:mini1}) iteratively using a block coordinate-descent algorithm. At coordinate $j$, we fix $\boldsymbol{d}_{j'}$ for $j' \neq j$, define $\wt{\boldsymbol{Y}}_n^{(j)} =  \boldsymbol{Y}_n -  \sum_{j' \neq  j} \boldsymbol{\Gamma}_{j'} \boldsymbol{H}_{j'}^{-1/2}   \boldsymbol{d}_{j'}$, and update $\boldsymbol{d}_{j}$ by
\vspace{-1.5em}
\begin{align}
    \widehat{\boldsymbol{d}_j} &= \argmin_{\boldsymbol{d}_j \in \mathbb{R}^{M_j}} \left\{ \frac{1}{2} \boldsymbol{d}_j^{\top} \boldsymbol{\Omega}_{j} \boldsymbol{d}_j -  \boldsymbol{\varrho}_j^{\top} \boldsymbol{d}_j + \lambda_1 \| \boldsymbol{d}_j \|_2 \right\}, 
\label{equ:mini_algo}
\end{align}
where $\boldsymbol{\Omega}_{j} =  \boldsymbol{H}_j^{-1/2} \left( \frac{1}{n} \boldsymbol{\Gamma}_j^{\top} \boldsymbol{\Gamma}_j + \lambda_2 \boldsymbol{I}_{M_j} \right) \boldsymbol{H}_j^{-1/2}$ and $\boldsymbol{\varrho}_j = n^{-1} \boldsymbol{H}_j^{-1/2} \boldsymbol{\Gamma}_j^{\top} \wt{\boldsymbol{Y}}_n^{(j)}$.

Proposition \ref{theorem:kkt_algoriothm} provides the solution to the minimization problem \eqref{equ:mini_algo}.
\begin{proposition}
\label{theorem:kkt_algoriothm}
For $\lambda_1 > 0$, the solution $\widehat{\boldsymbol{d}_j}$ for \eqref{equ:mini_algo} exists. Furthermore, if $\| \boldsymbol{\varrho}_j \|_2 \le \lambda_1$, then $\widehat{\boldsymbol{d}_j} = {0}$; if $\| \boldsymbol{\varrho}_j \|_2 > \lambda_1$, then $\widehat{\boldsymbol{d}_j} \neq 0$ and $\widehat{\boldsymbol{d}_j}$ is the solution to 
$\boldsymbol{\Omega}_{j} \boldsymbol{d}_j - \boldsymbol{\varrho}_j + \lambda_1 \boldsymbol{d}_j\| \boldsymbol{d}_j \|_2^{-1} = \boldsymbol{0}.$
\end{proposition}
We can solve $\widehat{\boldsymbol{d}_j}$ by iteratively updating $\boldsymbol{d}_j \leftarrow \left( \boldsymbol{\Omega}_{j} + \lambda_1 \| \boldsymbol{d}_j \|_2^{-1} \boldsymbol{I}_{M_j} \right)^{-1} \boldsymbol{\varrho}_j$ until convergence. 
Since the objective function \eqref{equ:mini_algo} is the combination of a convex and differentiable least squares loss and a convex penalty, the block coordinate-wise algorithm is guaranteed to converge to the global minimum \citep{Friedman2007Coordinate}.

For the refined estimator in \eqref{equ:object_ridge}, no iteration is needed since there is no $\ell_1$ penalty involved. Write $\widehat{f}_j(t) = \boldsymbol{\varphi}_j^{\top}(t) \widehat{\boldsymbol{c}}_j$ for each $j \in \widehat{\mathcal{S}} \equiv \{j_1, j_2, \dots, j_{\wh{q}}\}$.
Then, the coefficient vectors can be calculated as
\vspace{-1.5em}
\begin{align*}
    \left( \wh{\boldsymbol{c}}_{j_1}^{\top}, \dots, \wh{\boldsymbol{c}}_{j_{\hat{q}}}^{\top} \right)^{\top} = \frac{1}{n}  \left( \frac{1}{n} \boldsymbol{\Gamma}_{\wh\calS}^{\top} \boldsymbol{\Gamma}_{\wh\calS} + \lambda_3 \boldsymbol{I} \right)^{-1} \boldsymbol{\Gamma}_{\wh\calS}^{\top} \boldsymbol{Y}_n,
\end{align*}
where $  \boldsymbol{\Gamma}_{\wh\calS} = \left(\boldsymbol{\Gamma}_{j_1}, \dots, \boldsymbol{\Gamma}_{j_{\wh{q}}} \right)$ is the design matrix for functional predictors in the estimated signal set.

In most applications, the functional predictors are observed on $N$ equally spaced points, the kernel functions $K_j$ are evaluated as $N \times N$ matrices, $K_j^{1/2}$ are computed via the spectral decomposition of $K_j$, and all integrals can be approximated using Riemann sums on the observed discrete points. As discussed in \cite{zhou2023functional}, the incurred errors by these approximations are negligible when $N$ is sufficiently large.

\vspace{-1em}
\subsection{Simulation Studies}\label{sec:simulations}
\vspace{-1em}

We simulate the functional predictors as
$
X_{ij}(t)= \sqrt{2} \sum_{k \ge 1 } z_{ijk} \sqrt{\nu_k} \cos (k \pi t)$, $(i=1,\ldots,n, j=1,\ldots, p),
$
where $\boldsymbol{z}_{i \cdot k} =(z_{i1k}, \ldots, z_{ipk})^\top \sim $ i.i.d. $\Normal(\boldsymbol{0}, \boldsymbol{\Sigma}_p)$, and $\boldsymbol{\Sigma}_p$ is an autoregressive correlation matrix with the $(j,k)$th entry being $\rho^{|j-k|}$, $1 \leq k, j \leq p$. We generate the response $Y$ by the high-dimensional functional linear regression model \eqref{e:model}, using coefficient functions under one of the three scenarios described below and setting $\epsilon_i \sim \Normal(0, \sigma^2=0.5^2)$. For each scenario, we consider three correlation levels between the functional predictors, $\rho = 0$, $0.3$ and $0.75$, and three settings for the problem size: a high dimension and high sample size setting with $(n, p, q) = (500, 50, 5)$, a high dimension and low sample size setting with $(n, p, q) =(200, 100, 5)$, and an ultra-high dimension setting with $(n, p, q) = (100, 200, 10)$. 
For simplicity, we set the signal set to be $\calS=\{1,\ldots, q\}$, and set $\beta_{0j}(t)=4 \sum_{k \ge 1}   (-1)^{u_{jk}} r_k \phi_k(t)$, for $j\in \calS$,
where the basis functions $\phi_k(t)$ and coefficients $r_k$ are to be specified below, $u_{jk}$ are i.i.d. Bernoulli random variables with $P(u_{jk}=1)=0.5$. Inspired by \cite{CaiYuan2012}, we consider the following three scenarios for $\{ \phi_k(t), r_k, \nu_k\}$:

\noindent{\bf Scenario I:} $\phi_k(t)= \sqrt{2} \cos (k \pi t)$, and $\nu_k= r_k=\exp(-k/4)$, for $k\ge 1$;

\noindent{\bf Scenario \II:} $\phi_k(t)= \sqrt{2} \sin (k \pi t)$, and $\nu_k= r_k=\exp(-k/4)$, for $k\ge 1$;

\noindent{\bf Scenario \III:} $\phi_k(t)= \sqrt{2} \cos (k \pi t)$, $r_k= k^{-2}$, and $\nu_{k}=(|k-k_0|+1)^{-2}$ for $k\ge 1$, where we set $k_0 = 10$.

Scenario I represents a case where the functional predictors and the coefficient functions are perfectly aligned. Not only they are spanned by the same set of cosine functions, but the eigenvalues $\nu_k$ and the coefficients $r_k$ both monotonically decay with $k$. In other words, the signals most important to $X_{ij}$ also contribute the most to $Y_i$. As shown by \cite{CaiYuan2012}, $\beta_{0j}$ under this scenario belong to an RKHS with the RKHS norm $
\|\beta\|_{\mathbb{H}}=\{\int\left(\beta^{\prime \prime}\right)^{2} 
\}^{1/2}$, and the reproducing kernel
$K(s,t)  =-\frac{1}{3}\big[ B_{4}(|s-t| / 2)+B_{4}\{(s+t) / 2\} \big]$,
where $B_k$ is the $k$th Bernoulli polynomial. 

Scenarios \II~ and \III~ represent various cases of misalignment. Under Scenario \II, $X_{ij}$ and $\beta_{0j}$ are spanned by different bases. Using similar derivations as \cite{CaiYuan2012}, we can show $\beta_{0j}$ belong to an RKHS with the reproducing kernel $
K(s,t)  =-\frac{1}{3}\big[ B_{4}(|s-t| / 2) - B_{4}\{(s+t) / 2\} \big]$. Under Scenario \III, the maximum mode of variation in $X_{ij}$ is contributed from a high-frequency cosine function with $k=k_0$, however, these high-frequency signals do not contribute much to the response because the corresponding $r_k$'s are small.  
Even though the polynomial decay of the coefficient $r_k= k^{-2}$ in Scenario \III~ is slower than the exponential series $r_k= \exp(-k/4)$ in the asymptotic sense, as it turns out $\exp(-k/4) \ge k^{-2}$ for $k\le 26$. As such, there are practically more random components that contribute to the variations in $X_{ij}$ and the response $Y_i$ under Scenarios I and \II.

We repeat the simulation 200 times for each scenario, each level of correlation, and each problem size. For each simulated data set, we also simulate an additional sample of 100 data pairs of $(\BX, Y)$ as testing data to evaluate the prediction performance.
We apply our proposed functional elastic-net (fEnet) method to each simulated data set and make a comparison with the method proposed by \cite{XueYao2021}, which is to equip high-dimensional functional linear regression with a SCAD penalty \citep{FanLi2001} and thus termed FLR-SCAD. 
For FLR-SCAD, there are two tuning parameters, the SCAD penalty parameter $\lambda$ and the number of basis functions $s_1$ to represent both the functional predictor and the coefficient functions. For a fair comparison, we set the basis of FLR-SCAD to be the true basis $\phi_{k}(t)$ as described above. 
For the proposed fEnet, we set $\Psi_j = (\Tnjj + \theta \mathscr{I})^{1/2}$ and hence end up with four tuning parameters ($\lambda$, $\alpha$, $s$, and $\theta$), where $\lambda_1 = \alpha  \lambda$, $\lambda_2 = (1-\alpha)  \lambda$, and $s$ is the number of eigenfunctions used in the reduced rank approximation described in Section \ref{sec:algorithm}. 
For both methods, the tuning parameters are selected based on a grid search that minimizes the averaged mean square prediction error using the testing sample so that the results reported here represent the best possible performance of the two.
For a single tuning configuration, fEnet matches FLR–SCAD in runtime (around $3$ seconds at the optimal tuning parameter in the ultra–high–dimensional case). Since fEnet is insensitive to the basis size $s$, we fix $s$ and parallelize cross-validation to keep total cost low.
We use false positive rate (FPR) and false negative rate (FNR), defined as FPR$=|\wh \calS \cap \calS^c| / |\calS^c| $ and FNR$=|\wh \calS^c \cap \calS| / |\calS| $, to assess the variable selection performance, and we use the maximum norm difference (MND) to gauge the signal recovery performance, where MND is defined as
the maximum of the $\mathbb{L}_2$ norm of $\widehat{\beta}_j - \beta_{0j}$ for $j=1,\dots,p$. In order to make results from the three scenarios more comparable, we measure prediction error by the relative excess risk (RER): 
$
 { \E \{ \sum_{j=1}^p  \langle X_j^\ast, (\hat \beta_j- \beta_{j0}) \rangle\}^2 / \E \{ \sum_{j=1}^p  \langle X_j^\ast, \beta_{j0} \rangle\}^2 },
$
which is a standardized version of the excess risk. 

\begin{table}[!htbp]
\caption{Simulation Scenario I: summary of estimation, prediction, and variable selection performance of the proposed fEnet method versus FLR-SCAD under different problem sizes. }
\centering
\scriptsize
\begin{tabular}{cccccccc}
\toprule
$n$ & $p$ & $q$ & Method & FPR $(\%)$ & FNR $(\%)$ & MND & RER \\
\midrule
\multicolumn{8}{c}{$\rho=0$} \\
500 & 50 & 5 & fEnet & 0 (0, 0) & 0 (0, 0) & 0.36 (0.30, 0.45) & 0.0006 (0.0003, 0.0009) \\
& & & FLR-SCAD & 0 (0, 0) & 0 (0, 0) & 0.54 (0.37, 0.82) & 0.0009 (0.0005, 0.0019) \\ 
200 & 100 & 5 & fEnet & 0 (0, 0) & 0 (0, 0) & 0.53 (0.42, 0.68) & 0.0018 (0.0011, 0.0029) \\
& & & FLR-SCAD & 0 (0, 0) & 0 (0, 0) & 0.75 (0.58, 1.19) & 0.0035 (0.0017, 0.0106) \\
100 & 200 & 10 & fEnet & 0 (0, 1.1) & 0 (0, 0) & 1.31 (1.06, 1.65) & 0.0179 (0.0094, 0.0399) \\
& & & FLR-SCAD & 4.7 (1.6, 8.4) & 0 (0, 30)  & 4.89 (3.97, 5.00) & 0.5280 (0.3206, 0.7734) \\
\midrule
\multicolumn{8}{c}{$\rho=0.3$} \\
500 & 50 & 5 & fEnet & 0 (0, 0) & 0 (0, 0) & 0.37 (0.31, 0.47) & 0.0007 (0.0004, 0.0011) \\
& & & FLR-SCAD & 0 (0, 0) & 0 (0, 0) & 0.59 (0.41, 1.03) &  0.0012 (0.0006, 0.0027) \\
200 & 100 & 5 & fEnet & 0 (0, 0) & 0 (0, 0) & 0.58 (0.45, 0.73) & 0.0025 (0.0015, 0.0044)  \\
& & & FLR-SCAD & 0 (0, 0) & 0 (0, 0) & 0.78 (0.58, 1.51) & 0.0044 (0.0021, 0.0146) \\
100 & 200 & 10 & fEnet & 0 (0, 1.6) & 0 (0, 0) & 1.39 (1.08, 1.92) & 0.0192 (0.0103, 0.0441)  \\
& & & FLR-SCAD & 4.7 (1.6, 9.5) & 10 (0, 40) & 5.00 (4.37, 5.05)  & 0.5319 (0.3665, 0.7523)  \\
\midrule
\multicolumn{8}{c}{$\rho=0.75$} \\
500 & 50 & 5 & fEnet & 0 (0, 0) & 0 (0, 0) & 0.53 (0.42, 0.67) & 0.0012 (0.0007, 0.0019) \\
& & & FLR-SCAD & 0 (0, 0) & 0 (0, 0) & 0.98 (0.67, 1.78) &  0.0018 (0.0008, 0.0049) \\
200 & 100 & 5 & fEnet & 0 (0, 0) & 0 (0, 0) & 0.85 (0.72, 1.03) & 0.0035 (0.0021, 0.0056)  \\
& & & FLR-SCAD & 0 (0, 0) & 0 (0, 0) & 1.28 (0.76, 4.61) & 0.0066 (0.0029, 0.1287) \\
100 & 200 & 10 & fEnet & 0 (0, 4.2) & 0 (0, 10) & 2.04 (1.49, 5.00) & 0.0175 (0.0078, 0.1329)  \\
& & & FLR-SCAD & 2.1 (0, 4.2) & 50 (30, 70) & 5.86 (5.00, 7.91)  & 0.2895 (0.1932, 0.3894)  \\
\bottomrule
\end{tabular}
\label{table:compare}
\end{table}

Simulation results under Scenario I are summarized in Table \ref{table:compare}, where we compare the median FPR, FNR, MND, and RER as well as their $2.5\%$ and $97.5\%$ quantiles for the two competing methods. As we can see, both methods accurately choose the correct model under the first two problem sizes and for all correlation levels, although our method shows some small advantages in terms of estimation (MND) and prediction (RER). We now focus on the ultra-high dimension setting with $(n, p, q) = (100, 200, 10)$, where our method shows an overwhelming advantage over FLR-SCAD in all criteria considered for variable selection, estimation, and prediction. Note that under the high correlation setting ($\rho=0.75$), not only $\{X_{ij}, \ j\in \calS\}$ are strongly correlated among themselves, but they are also strongly correlated with some of the predictors in $\calS^c$. In this case, even though FLR-SCAD mistakes some of the non-signals with some real signals, its prediction performance may not be as bad as when $\rho=0$ or $0.3$.


To further investigate the variable selection performance under the ultra-high dimension setting, we plot the receiver operating characteristic (ROC) curves for the two methods in Figure \ref{img:ROC}, where the false positive rate and true positive rate (TPR), i.e. $1-$FNR, are calculated under different values of $\lambda$ while holding other tuning parameters fixed at their optimal values. As such, both FPR and TPR become functions of $\lambda$. As $\lambda$ increases, all coefficient functions are shrunk to 0 and hence both FPR and TPR decrease to 0. The ROC of our method yielding a higher area under the curve (AUC) than FLR-SCAD, especially when there is a high correlation between the functional predictors, means that our method has a better variable selection performance.


\begin{figure}[!htbp]
    \begin{subfigure}{0.31\textwidth}
        \centering
        \includegraphics[width=0.95\linewidth]{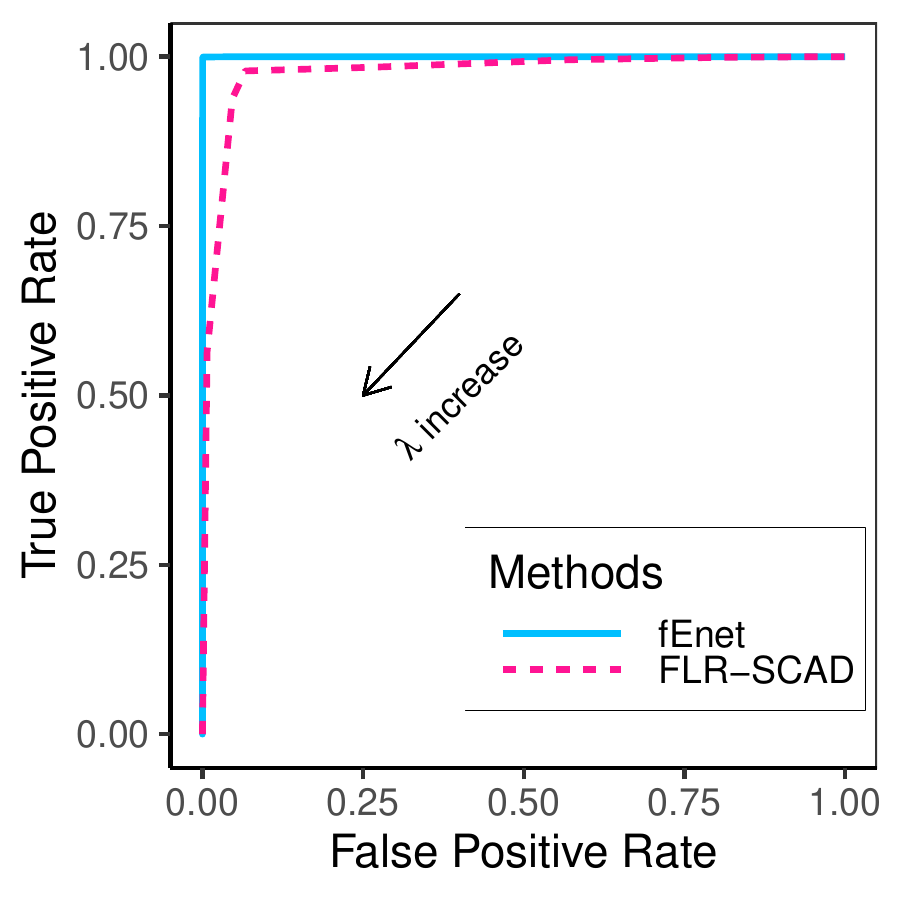}
        \vspace{-0.5em}
        \caption{$\rho = 0$}
    \end{subfigure}
    \begin{subfigure}{0.31\textwidth}
        \centering
        \includegraphics[width=0.95\linewidth]{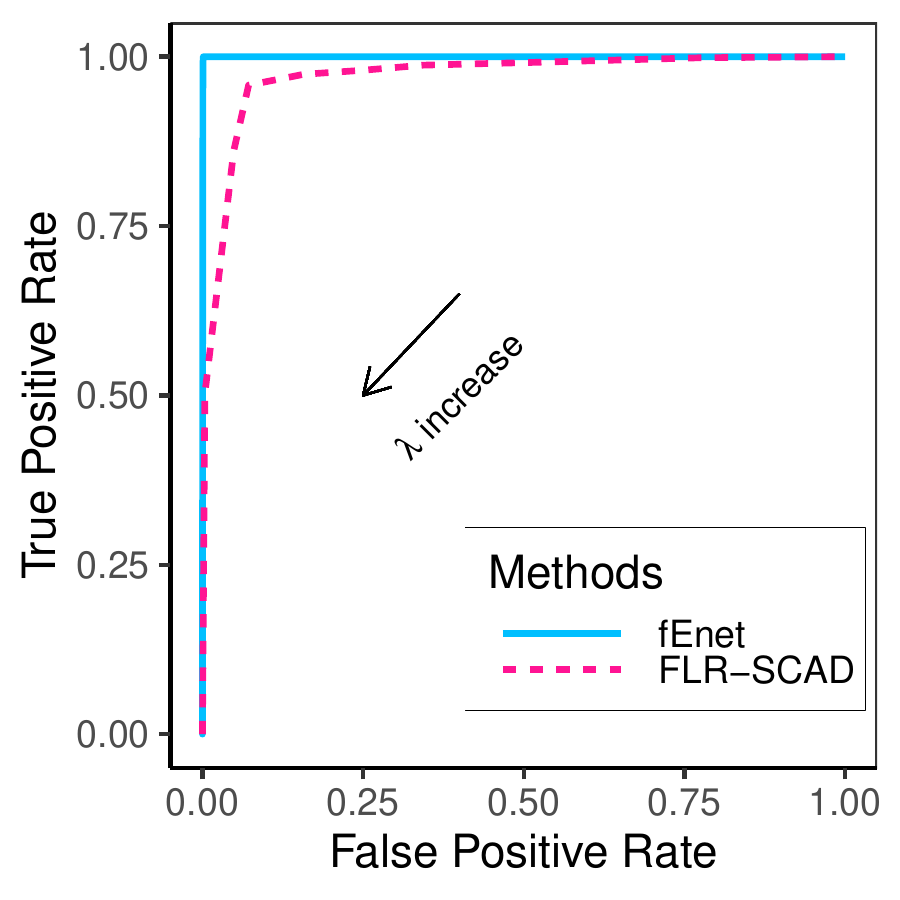}
        \vspace{-0.5em}
        \caption{$\rho = 0.3$}
    \end{subfigure}
    \begin{subfigure}{0.31\textwidth}
        \centering
        \includegraphics[width=0.95\linewidth]{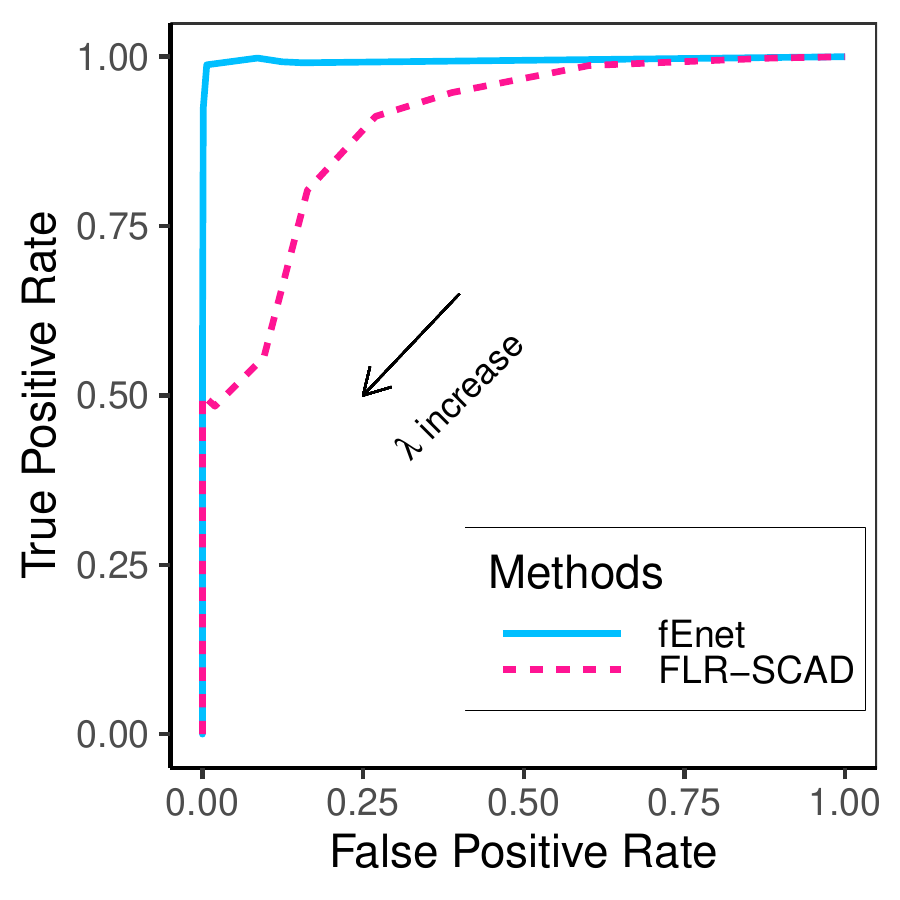}
        \vspace{-0.5em}
        \caption{$\rho = 0.75$}
    \end{subfigure}
  	\caption{Simulation Scenario I: The ROC curves of fEnet and FLR-SCAD under the ultra-high-dimension setting $(n, p, q) = (100, 200, 10)$. The ROC curves are obtained by changing the value of $\lambda$ and holding other hyperparameters at optimal.}
  	  	\label{img:ROC}
\end{figure}


\begin{figure}[h]
    \begin{subfigure}{0.31\textwidth}
        \centering
        \includegraphics[width=0.95\linewidth]{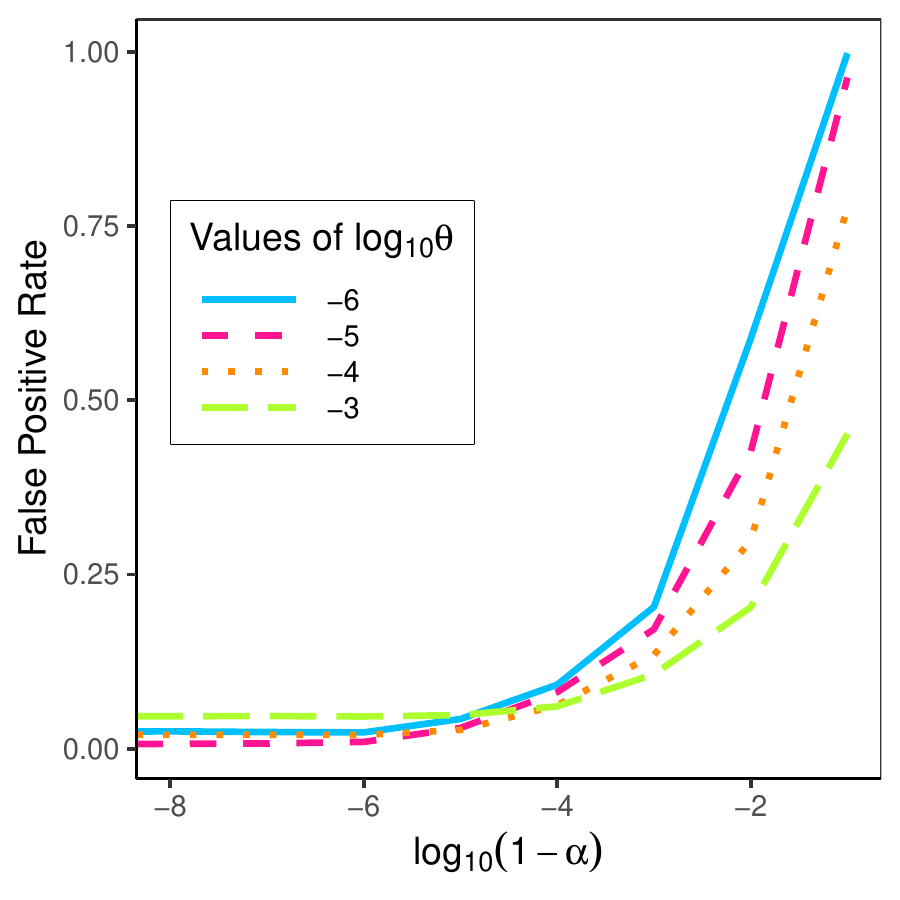}
        \vspace{-0.5em}
        \caption{FPR}
    \end{subfigure}
    \begin{subfigure}{0.31\textwidth}
        \centering
        \includegraphics[width=0.95\linewidth]{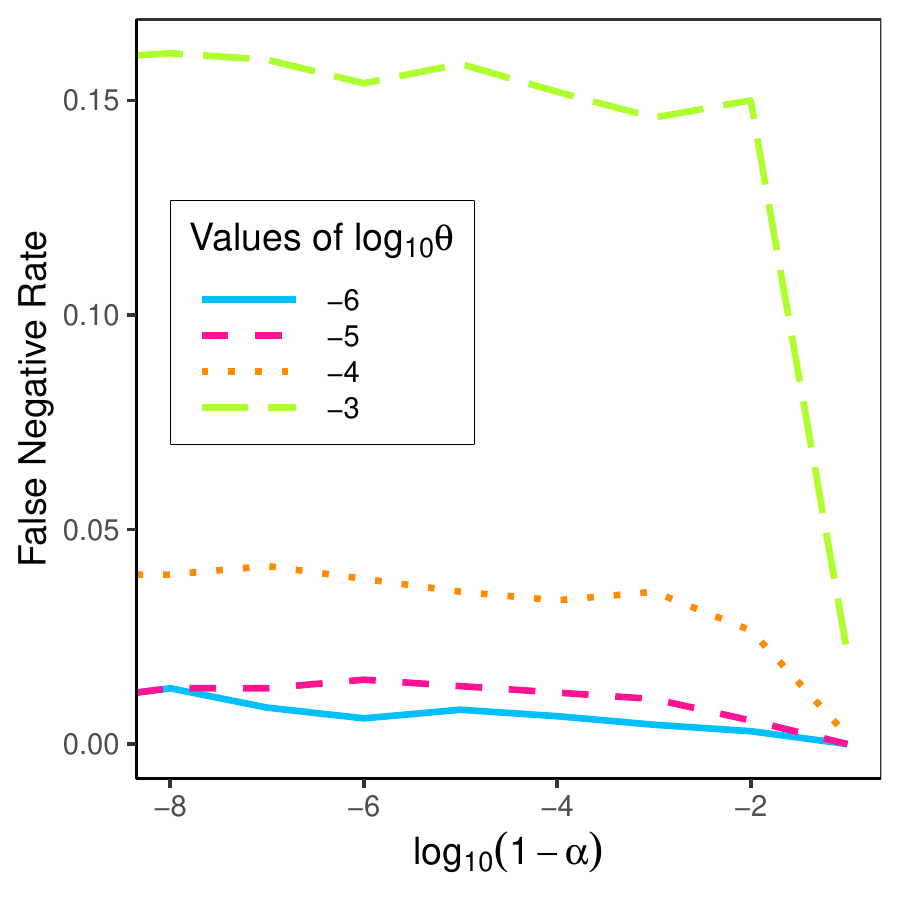}
        \vspace{-0.5em}
        \caption{FNR}
    \end{subfigure}
    \begin{subfigure}{0.31\textwidth}
        \centering
        \includegraphics[width=0.95\linewidth]{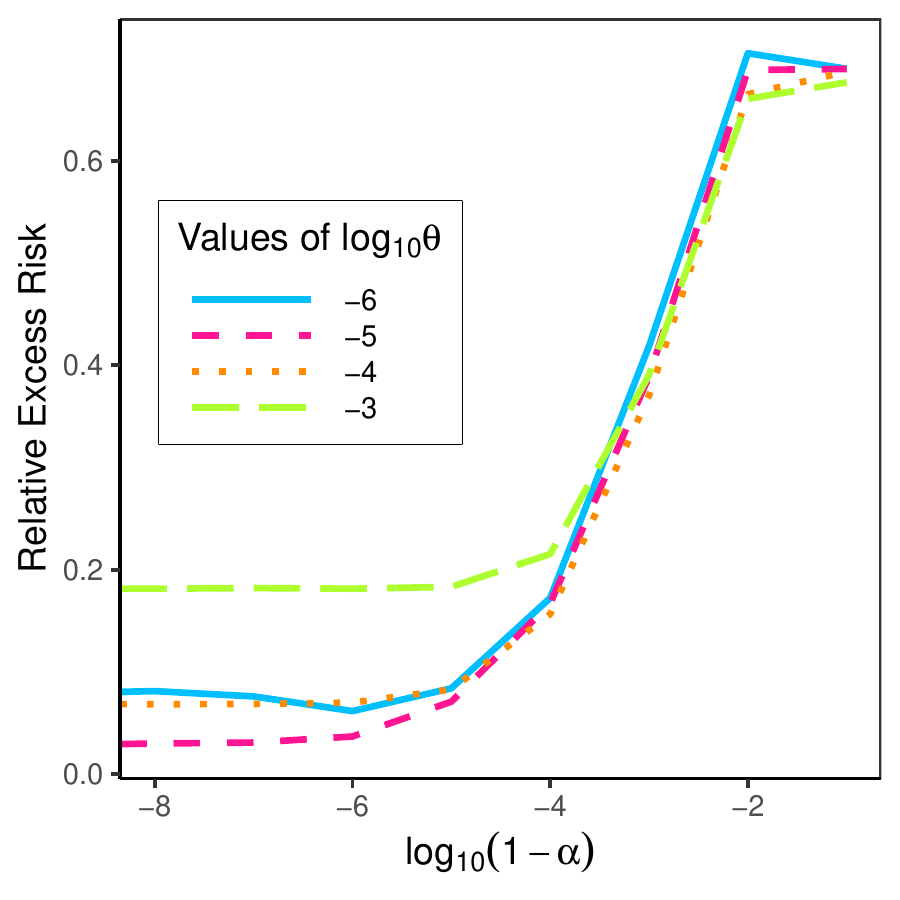}
        \vspace{-0.5em}
        \caption{RER}
    \end{subfigure}
  	\caption{Simulation Scenario I: The plots of FPR, FNR, and RER versus $\log_{10}(1-\alpha)$ for different values of $\theta$ under the ultra-high-dimensional case and $\rho=0.75$.}
  	  	\label{img:change_alpha_theta}
\end{figure}

To investigate the effect of $\alpha=\lambda_1/ (\lambda_1+\lambda_2)$ and $\theta$ on the variable selection and prediction performance, we revisit the ultra-high dimension setting with $\rho = 0.75$. We calculate the average FPR, FNR, and RER at various values of $\alpha$ and $\theta$ while keeping $\lambda$ and $s$ fixed at their optimal values. In Figure \ref{img:change_alpha_theta} we plot the averaged FPR, FNR, and RER against $\log_{10}(1-\alpha)$ for different values of $\theta$. 
These plots suggest that for any fixed $\theta$, FPR is a decreasing function of $\alpha$ while FNR increases with $\alpha$. This observation corroborates our remarks for Theorem \ref{thm:1} that a larger ratio between $\lambda_1$ and $\lambda_2$ means more predictors will be removed from the model and hence the decreased FPR and increased FNR. There should be an optimal $\alpha$, which is neither 0 nor 1, providing the best trade-off between FPR and FNR. The plot of RER against $\log (1-\alpha)$ also suggests the existence of a non-trivial optimal value for $\alpha$, which in turn suggests that we need both components in the elastic-net penalty for the best performance. 
By comparing curves across different values of $\theta$, we can see that FPR decreases with $\theta$, FNR increases with $\theta$, and RER is not monotone with $\theta$. All of these point to the conclusion that there is non-zero optimal value for $\theta$.

To save space, results under Scenarios \II ~and \III ~are deferred to the supplementary material. When there is a misalignment between the functional predictor and the coefficient functions, particularly under Scenario \III ~with a high correlation between the functional predictors, we observe better FPR and FNR from the proposed fEnet method not only for the ultra-high dimension setting but all the other problem sizes as well.

\begin{table}[h]
\caption{Relative efficiency (RE) between the functional elastic-net estimate and the two-stage estimate under Scenario I}
\centering
\footnotesize
\begin{tabular}{cccccc}
\toprule
$n$ & $p$ & $q$ & $\rho=0$ & $\rho=0.3$ & $\rho=0.75$ \\
\midrule
500 & 50 & 5 & 1.04  & 1.06  &  1.29  \\
200 & 100 & 5 & 1.30  & 1.44  & 1.51  \\
100 & 200 & 10 & 1.63  & 1.68  & 1.95 \\
\bottomrule
\end{tabular}
\label{table:two_stage}
\end{table}

Next, we demonstrate the efficiency gain of the refined estimator \eqref{equ:object_ridge} in prediction performance. Focusing on Scenario I, we refit FLM to the simulated data as described in \eqref{equ:object_ridge} using the predictors selected by fEnet only. The tuning parameter $\lambda_3$ is selected by a grid search that minimizes the averaged mean square prediction error using the testing sample. 
Table \ref{table:two_stage} presents a summary of the relative efficiency (RE) between the fEnet estimator $\widehat{\boldsymbol{f}}$ and the refined estimator $\fhaths$, where $\mbox{RE}(\wh{\boldsymbol{f}}, \fhaths) =  \mbox{RER}( \wh{\boldsymbol{f}} ) / \mbox{RER}( \fhaths)$. The reported REs are based on the average over 200 replicates, and a value of RE greater than 1 indicates an improved prediction performance in the refined estimator. These results demonstrate improved prediction performance of the refined estimator across all problem sizes and correlation levels, particularly in the case of ultra-high-dimension and high correlation between functional predictors, where the refined estimator is almost twice as efficient as the original fEnet. 

\subsection{Real Data Application}\label{sec:data}

We now demonstrate our methodology using a dataset obtained from the Human Connectome Project (HCP) \citep{van2013wu}. The data comprise resting-state fMRI scans from $n=549$ individuals, where each brain was repeatedly scanned over 1200 time points. These 3-dim fMRI images were pre-processed and parcellated into 268 brain regions-of-interest (ROI) using a whole-brain, functional atlas defined in \cite{finn2015functional}. Since the raw ROI level fMRI time series are quite noisy, we instead treat the smoothed periodograms at different ROI's as high-dimensional functional data. Specifically, we apply Fast Fourier Transform to the fMRI time series at each ROI, smooth the resulting periodogram using the `smooth.spline' function in \textit{R}, and keep the most informative segment from 1 to 300 Hz as a functional predictor. 
In addition to the fMRI, each subject in the study also undertook the Penn Progressive Matrix (PPM) test, the score of which is commonly used as a surrogate for fluid intelligence \citep{greene2018task}. 

This dataset was previously analyzed by \cite{lee2021conditional}, who used the raw fMRI time series as functional data and the PPM score as a covariate to study functional connectivity between the ROI's. We instead treat the smoothed periodograms from the 268 ROI's as high-dimensional functional predictors and the PPM score as the response.  
By fitting a high-dimensional functional linear model using the proposed fEnet method, our goal is to identify brain regions that are associated with fluid intelligence.

\begin{figure}[t]
    \begin{subfigure}{0.30\textwidth}
        \centering
        \includegraphics[width=0.75\linewidth]{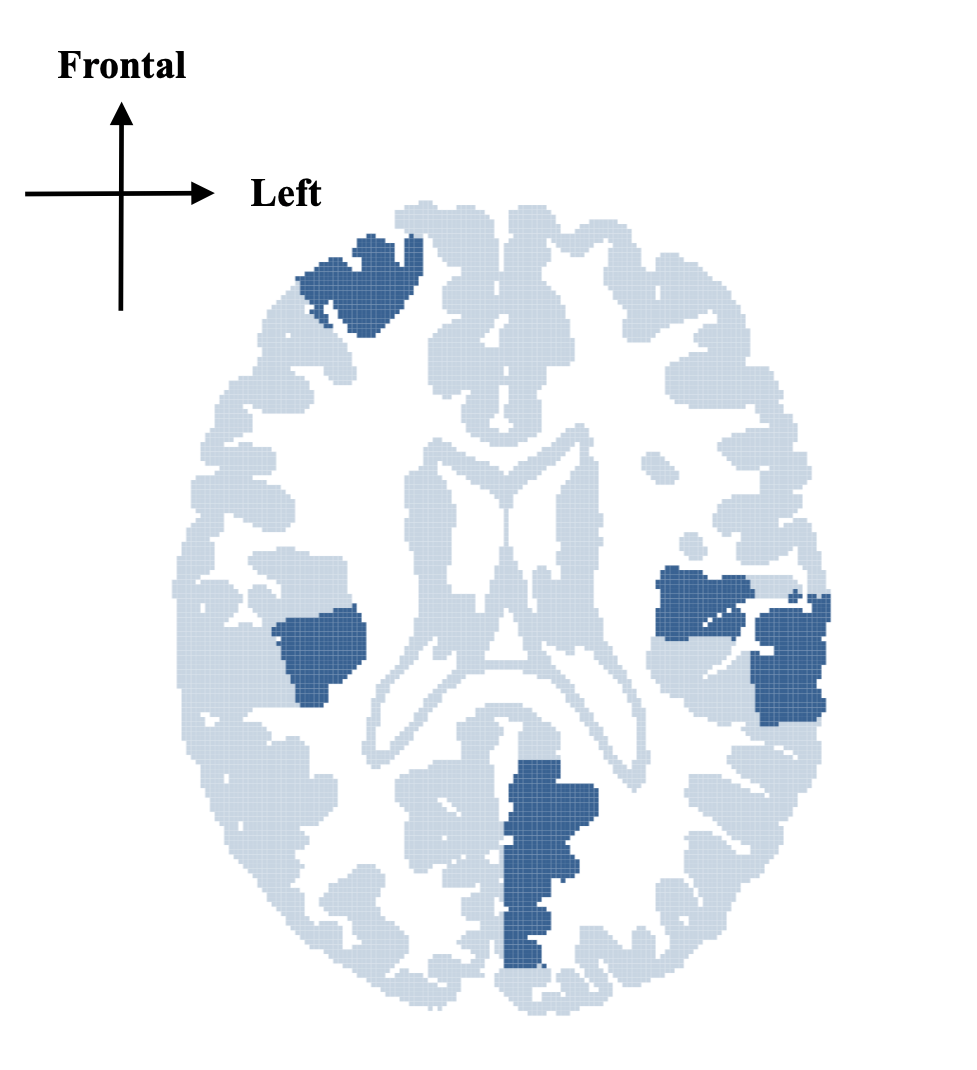}
        \vspace{-0.5em}
        \caption{Top view}
    \end{subfigure}
    \begin{subfigure}{0.33\textwidth}
        \centering
        \includegraphics[width=0.81\linewidth]{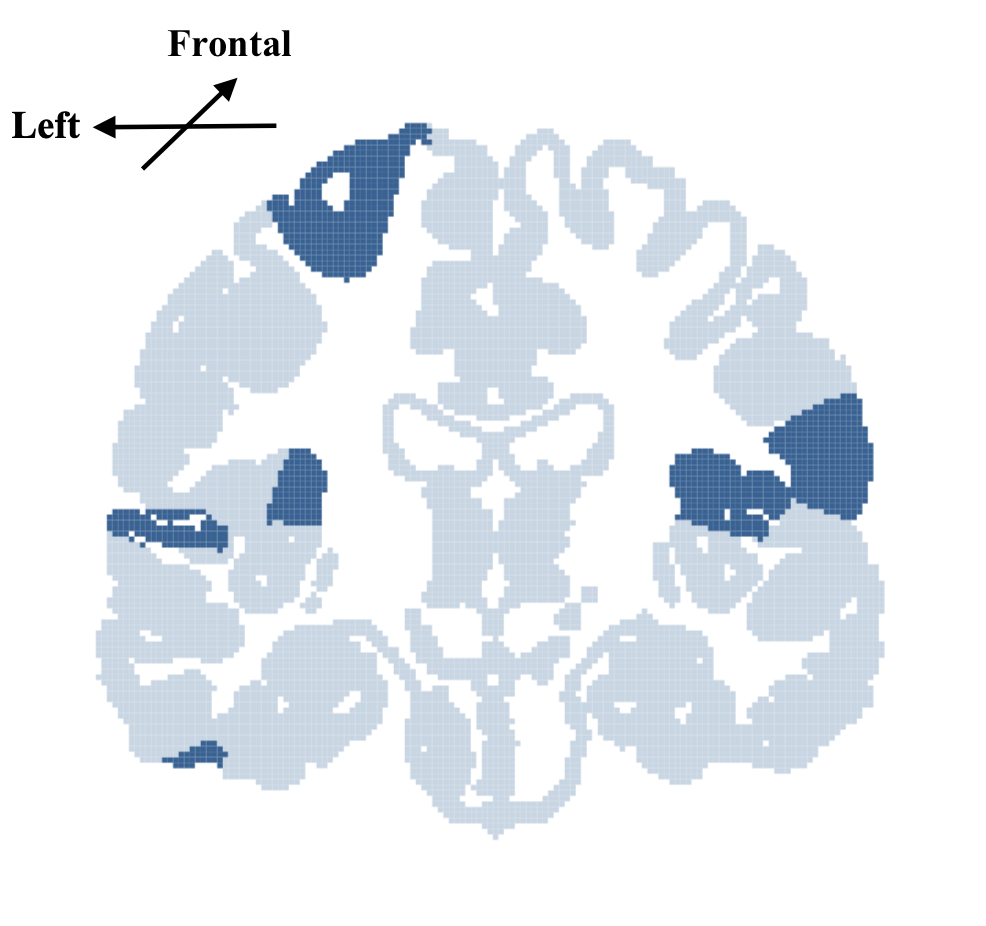}
        \vspace{-0.5em}
        \caption{Front view}
    \end{subfigure}
    \begin{subfigure}{0.33\textwidth}
        \centering
        \includegraphics[width=0.81\linewidth]{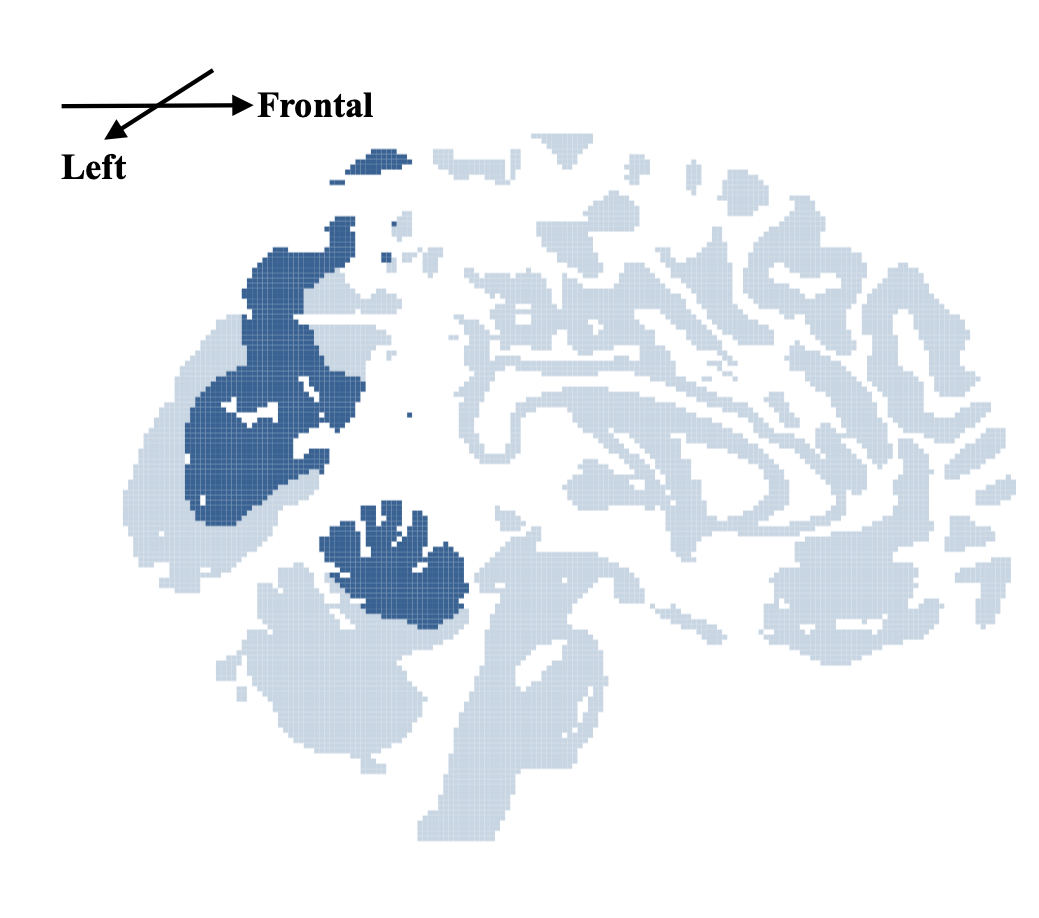}
        \vspace{-0.5em}
        \caption{Side view}
    \end{subfigure}
  	\caption{The orthographic projections of a brain (light blue), where the 33 selected ROIs using the HCP data are marked in dark blue.}
  	  	\label{img:real_roi}
\end{figure}


To ensure the robustness of our results, we randomly divide the 549 individuals into a training set ($80\%$) and a validation set ($20\%$) for a total of 200 times. We select the optimal tuning parameters of our model by minimizing the averaged mean squared prediction error (MSPE) on the 200 validation sets. 
We find 33 ROIs that are consistently selected by our proposed method across all 200 repetitions. 
In Figure \ref{img:real_roi}, we provide three projection views of the brain and mark the physical locations of the selected ROIs. Our results suggest that fluid-intelligence-related ROIs are distributed in multiple brain regions, including those on the prefrontal and parietal cortices. These findings agree with the literature \citep{duncan2000neural, jung2007parieto} that fluid intelligence, considered a complex cognitive ability that involves various cognitive processes, is typically associated with multiple brain regions.

\section{Summary}\label{sec:summary}

Our RKHS-based functional elastic-net method is different from existing high-dimensional functional linear regression methods in two important ways. First, we do not express the functional predictors and the coefficient functions using the same set of basis functions, which offers the extra flexibility to choose the reproducing kernel based on the application and better numerical performance when the functional predictors and the coefficient functions are misaligned. Second, our penalty consists of two parts: a lasso-type penalty on the normal of the prediction error to enforce sparsity and a ridge penalty that regularizes the smoothness of the coefficient function for better prediction. Our simulations show that both penalties are important and that the best performance in terms of variable selection, estimation, and prediction is achieved by finding the best trade-off between the two penalties. 
We also derived a sharp non-asymptotic probability bound on the event of our method achieving variable selection consistency, while assuming the functional predictors are non-degenerative random elements in infinite-dimensional Hilbert spaces. Our theory also suggests a bound for the smallest signal size that can be detected by the functional elastic-net method. Our investigation of the minimax optimal rate for high-dimensional FLM is completely new, and we show that our post-selection refined RKHS estimator achieves the oracle minimax optimal excessive risk. The efficiency gain from using the refined estimator is also demonstrated through simulation studies.

Handling sparsely or irregularly observed functional covariates is an important yet nontrivial extension. There has been some recent work addressing functional linear regression with a single discretely observed predictor under an FPCA framework \citep{zhou2023functional}. However, to the best of our knowledge, analogous results within an RKHS framework, particularly for high-dimensional settings, remain unavailable. We note this gap as a promising direction for future research.

\section*{Acknowledgement}
The authors thank the editor, the associate editor, and two anonymous referees for their many helpful and constructive comments, which led to significant improvements to our paper.

\vspace{-0.6em}
\begingroup
  \setstretch{1.6}          
  \setlength{\bibsep}{5pt}
  \setlength{\parskip}{2pt}
  \bibliographystyle{apalike}
  \bibliography{flm_vs_ref}
\endgroup

\clearpage\pagebreak\newpage
\setcounter{section}{0}
\renewcommand{\thesection}{S.\arabic{section}} 
\renewcommand{\thesubsection}{S.\arabic{section}.\arabic{subsection}}
\setcounter{page}{1}
\renewcommand{\theequation}{S.\arabic{equation}}
\renewcommand{\thefigure}{S.\arabic{figure}}
\renewcommand{\thetable}{S.\arabic{table}}
\setcounter{table}{0}
\setcounter{equation}{0} \setcounter{figure}{0} 
		
\renewcommand{\theenumi}{\arabic{enumi}}
\renewcommand{\thelemma}{S.\arabic{lemma}}
\setcounter{lemma}{0}
		
\begin{center}
{\Large{\bf Supplementary Material for ``Variable Selection and Minimax Prediction in High-Dimensional Functional Linear Models"}}\\ 
			\vskip1cm
\end{center}
\begin{center}
    \large{Xingche Guo, Yehua Li and Tailen Hsing}
\end{center}
	
\vskip1cm

This Supplemental Paper is organized as follows. We provide technical details in Section \ref{s:proofs}, where Section \ref{s:KKT} includes details for deriving the Karush-Kuhn-Tucker conditions in function spaces, Section \ref{s:partially_separable} introduces a special case for high-dimensional functional predictors with partially separable covariance structure.
All technical proofs are provided in Section \ref{sec:proofs}, where Section \ref{s:propositions} contains the proofs of the propositions, Section \ref{sec:proof_theorems_corollary} contains the proofs of Theorems \ref{thm:1}-\ref{thm:minimax_upper} as well as Corollary \ref{cor:R(f)}, Section \ref{s:lemmas} provides the proofs of the lemmas used in the main proof of Theorem \ref{thm:1}, and Section \ref{s:addi_lemmas} provides some additional technical lemmas. Substantiating examples are provided in Section \ref{s:MA_AR_cases} to support the technical assumptions made in the paper, and some additional simulation results are provided in Section \ref{sec:add_simulation}.


\section{Technical Details} \label{s:proofs}

\subsection{Karush-Kuhn-Tucker Conditions in Function Spaces} \label{s:KKT}

In this section, we introduce the Karush-Kuhn-Tucker (KKT) condition in function spaces 
and specialize it for \eqref{equ:mini1}.
First, we review the notion of Gateaux differentiability. For convenience, let $\mathcal{J}$ denote a mapping from some 
Hilbert space $\mathbb{H}$ to $\bbR$, where $\mathcal{J}$ is not necessarily linear. 
We note that the Hilbert space
assumption in the definition below could be relaxed depending on the context of the application.
 
\begin{definition}
(Gateaux differentiability)
For $f,\psi\in\mathbb{H}$, we say that  $\mathscr{J}$ is Gateaux differentiable at $f$ in the direction 
of $\psi$ if~ $\mathrm{lim}_{\tau \downarrow 0^{+}} \frac{\mathscr{J}(f + \tau \psi)  - \mathscr{J}(f) }{\tau}$ 
and  $\mathrm{lim}_{\tau \uparrow 0^{-}} \frac{\mathscr{J}(f + \tau \psi)  - \mathscr{J}(f) }{\tau}$ exist and are equal. The common limit in this case is denoted by $\mathcal{D}_{\mathcal{J}}(f; \psi)$ 
and is referred to as the Gateaux derivative of $\mathscr{J}$at $f$ in the direction of $\psi$.
If $\mathcal{D}_{\mathcal{J}}(f; \psi)$ is defined for all $\psi\in\mathbb{H}$, we
say that $\mathcal{J}$ is Gateaux differentiable at $f$.
 \end{definition}
 
Clearly, if $\mathcal{J}$ is Gateaux differentiable at $f$  then $\mathcal{D}_{\mathcal{J}}(f; \cdot)
\in \mathfrak{B}(\mathbb{H}, \mathbb{R})$, 
the space of continuous linear functionals on $\mathbb{H}$. 
On the other hand, if $\mathcal{J}$ is convex but not necessarily Gateaux differentiable, then the useful notions of sub-derivative and sub-differential can be defined as follows.

\begin{definition}
(Sub-derivative and sub-differential) The Gateaux sub-differential of a convex functional $\mathcal{J}$ at $g$ 
is defined as the collection 
$\partial_{\mathcal{J}(g)} = \{\mathscr{A}
\in \mathfrak{B}(\mathbb{H}, \mathbb{R}): \mathcal{J}(f) \ge \mathcal{J}(g) +  \mathscr{A} (f-g) \mbox{ for all } f \in \mathbb{H}\}$ of linear functionals, where the
elements in $\partial_{\mathcal{J}(g)}$ are referred to as sub-derivatives.
\end{definition}

\begin{proposition}\label{prop:convex}
Any Gateaux differentiable mapping $\mathcal{J}$ from $\mathbb{H}$ to $\mathbb{R}$ is convex if and only if $\mathcal{J}(f) \ge \mathcal{J}(g) + \mathcal{D}_{\mathcal{J}}(g; f-g)$ for all $f, g\in \mathbb{H}$, in which case
$\mathcal{J}(g)$ is the global minimum of $\mathcal{J}(\cdot)$ if and only if $\mathcal{D}_{\mathcal{J}}(g; \cdot) \equiv 0$.
Suppose, on the other hand, that $\mathcal{J}$ is convex but not Gateaux differentiable. Then $\mathcal{J}(g)$ is the global minimum of $\mathcal{J}$ if and only if $0 \in \partial_{\mathcal{J}(g)}$.
\end{proposition}

With $\mathscrbf{T}_n$ and $\boldsymbol{g}_n$ defined in \eqref{equ:kkt}, the objective function $\ell(\boldsymbol{f})$ can be expressed as
\begin{align} \label{equ:mini1_1}
\ell(\boldsymbol{f})
= \sum_{i=1}^4 \ell_i(\boldsymbol{f}) + \frac{1}{2n} \|\boldsymbol{\varepsilon}_n\|_2^2, 
\end{align}
where 
$\ell_1(\boldsymbol{f}) = {1\over 2} \langle  \mathscrbf{T}_n(\boldsymbol{f} - {\boldsymbol{f}_0}), \boldsymbol{f} - {\boldsymbol{f}_0} \rangle_{2}$, 
$\ell_2(\boldsymbol{f}) = -\langle \boldsymbol{g}_n,  \boldsymbol{f} - {\boldsymbol{f}_0}  \rangle_{2}$, 
$\ell_3(\boldsymbol{f}) = \frac{\lambda_2}{2} \|\boldsymbol{f}\|_{2}^2$, $ \ell_4(\boldsymbol{f}) = \lambda_1\sum_{j=1}^p\left\| \Psi_j {f_j} \right\|_{2}, \ \boldsymbol{f} \in\mathbb{L}_2^p$.
The following proposition contains the key elements in minimizing $\ell(\boldsymbol{f})$ based on \eqref{equ:mini1_1}.

\begin{proposition} \label{pp:L1_L3}
The functionals $\ell_i,i=1,2,3$, are Gateaux differentiable at all $\boldsymbol{f}\in\mathbb{L}_2^p$, where
$\mathcal{D}_{ \ell_1 }(\boldsymbol{f}; \boldsymbol{\psi})$ $ = \langle \mathscrbf{T}_n(\boldsymbol{f} - \boldsymbol{f}_0), \boldsymbol{\psi} \rangle_{2}$, 
$\mathcal{D}_{\ell_2 }(\bdf; \boldsymbol{\psi}) =  -\langle  \boldsymbol{g}_n , \boldsymbol{\psi} \rangle_{2}$, and $\mathcal{D}_{\ell_3}(\boldsymbol{f}; \boldsymbol{\psi}) = \lambda_2 \langle \boldsymbol{f}, \boldsymbol{\psi} \rangle_{2}$.
The sub-differential of $\ell_4$ at $\boldsymbol{f}$ contains all functionals of the form $\lambda_1\left\langle\boldsymbol{\omega}, \cdot\right\rangle_2$, $\boldsymbol{\omega}\in\mathbb{L}_2^p$, such that $\omega_j =\frac{\Psi_j^2 f_j}{\| \Psi_j f_j \|_2}$ if $f_j\not =0$ and $\omega_j = \Psi_j\eta_j$ for any arbitrary $\eta_j$ with $\|\eta_j\|_2\le 1$ if $f_j=0$.

 \end{proposition}

Note that the KKT condition \eqref{equ:kkt} can be easily derived from Propositions \ref{prop:convex} and \ref{pp:L1_L3}. The proofs for Propositions \ref{prop:convex} and \ref{pp:L1_L3} are given in the Supplementary Material.

\subsection{Partially Separable Covariance Structure}
\label{s:partially_separable}

To gain a deeper understanding of Conditions C.\ref{ass:a2}-C.\ref{ass:a6}, we consider functional predictors with a partially separable covariance structure \citep{Zapata2021}, namely,
\ben
    \scrbfT^{(\calS,\calS)}=\sum_{k=1}^\infty \BA_k \psi_k \otimes \psi_k,\label{equ:partial_separable}
\een
where $\{\psi_k, k\ge 1\}$ are orthonormal functions in $\bbL_2[0,1]$ and $\{\BA_k, k\ge 1\}$ are a sequence of $q \times q$ covariance matrices. Further, consider $\BA_k = \nu_k \BR$, with $\nu_1 \ge \nu_2 \ge \cdots >0$ a sequence of eigenvalues and $\BR$ a $q\times q$ correlation matrix. In this setting, $\{X_{j}, j\in \calS\}$ share the same eigenvalues and eigenfunctions, and their principal component scores have the same correlation structure across different order $k$. To satisfy Condition C.\ref{ass:a2}, we must have $\nu_1=1$ and $\sum_{k\ge 1} \nu_k \le \tau <\infty$. To find the upper bound for $\varkappa(\lambda_2)$, first note that
\bse
    \scrbfT^{(\calS,\calS)} (\scrbfT^{(\calS,\calS)}_{\lambda_2})^{-1} = \sum_{k=1}^\infty \BA_k (\BA_k +\lambda_2 \BI)^{-1} \psi_k \otimes \psi_k =: \sum_{k=1}^\infty \BB_k \psi_k \otimes \psi_k,
    %
%
\ese
where $\BB_k = \BR (\BR +\vartheta_k \BI)^{-1}$ and $\vartheta_k= \lambda_2/ \nu_k \to \infty$ as $k\to \infty$.
Writing $\BB_k= \{B_{k,jj'}\}_{j,j'=1}^q$, it follows that
\ben\label{eq:infty_norm_bound}
    \vertiii{\scrbfT^{(\calS,\calS)} (\scrbfT^{(\calS,\calS)}_{\lambda_2})^{-1}}_{\infty, \infty} \le \max_{1\le j \le q} \sum_{j'=1}^q \max_k |B_{k,jj'}|.
\een
In Section \ref{s:MA_AR_cases} of the supplementary material, we examine two specific scenarios where $\BR$ is either a $MA(1)$ or $AR(1)$ correlation matrix. We find that the upper bound of $\varkappa(\lambda_2)$ is equal to some constant independent of $\lambda_2$ and the true signal size $q$. 
Furthermore, we find that Condition C.\ref{ass:a6} holds for all legitimate $MA(1)$ correlation matrices and for $AR(1)$ correlation matrices characterized by an autoregressive coefficient less than $1/3$.

\section{Technical Proofs}\label{sec:proofs}

\subsection{Proof of Propositions} \label{s:propositions}

\textbf{Proof of Proposition \ref{lemma:solution_form}}
\begin{proof}
Rewrite the minimization function (\ref{equ:mini1}),
\begin{align*} 
	\ell(\boldsymbol{f} )
	:=\frac{1}{2n} \sum_{i=1}^n \left(Y_i -  \sum_{j=1}^p \langle \widetilde{X}_{ij},  f_j \rangle_{2}  \right)^2 
	+  \lambda_1 \sum_{j=1}^p \|\Psi_j f_j \|_{2} + \frac{\lambda_2}{2} \sum_{j=1}^p \|f_j\|_{2}^2.
\end{align*}
The minimizer $\widetilde{f}_j(t)$ can always be represented in the form
\begin{align*}
	\widetilde{f}_j(t) = \widehat{f}_j(t) + \eta_j(t),
\end{align*}
where $\widehat{f}_j(\cdot) = \sum_{i=1}^n c_{ij} \widetilde{X}_{ij}(\cdot) \in \mathbb{M}_{nj}$, and $\eta_j(\cdot)  \in \mathbb{M}_{nj}^{\perp}$. 
Therefore, we have $ \langle \widetilde{X}_{ij},  \widetilde{f}_j \rangle_{2} = \langle \widetilde{X}_{ij},  \widehat{f}_j \rangle_{2}$,  $\| \widetilde{f}_j \|_2^2  = \| \widehat{f}_j \|_2^2 + \left\| \eta_j \right\|_2^2$, and 
$ \| \Psi_j \widetilde{f}_j \|_2^2  = \| \Psi_j \widehat{f}_j \|_2^2 + \| \Psi_j \eta_j \|_2^2$.
The last equation holds by Condition (C.\ref{ass:a1}).
Therefore, $\widetilde{f}_j(t)$ is the minimizer when $\eta_j \equiv 0$.
\end{proof}

\noindent {\bf Proof of Proposition \ref{theorem:kkt}} 

\begin{proof}
    \noindent
The KKT condition \eqref{equ:kkt} follows readily from Propositions \ref{prop:convex} and \ref{pp:L1_L3}.
We can show the existence of functional KKT solution by showing that the minimizer of (\ref{equ:mini1_1}) exists. Note that \eqref{equ:mini1_1} can be reformulated as a constrained quadratic programming problem:
\begin{align*}
\mbox{min}_{\boldsymbol{f}} \left\{\ell_1(\boldsymbol{f}) + \ell_2(\boldsymbol{f})
\right\} \mbox{ such that } \ell_3(\boldsymbol{f})
\le C_1 \mbox{ and } \ell_4(\boldsymbol{f})
\le C_2.
\end{align*}
where $(C_1,C_2)$ here have a one-to-one correspondence with the regularization parameters $(\lambda_1, \lambda_2)$ via the Lagrangian duality. It follows from
Proposition \ref{lemma:solution_form} that the solution can be found in a finite-dimensional subspace. 
Therefore, the above minimization problem involves a continuous finite-dimensional quadratic objective function over a compact set. By Weierstrass' extreme value theorem, the minimum is always achieved.
To show uniqueness, first note that there is either a unique solution or an (uncountably) infinite number of solutions. This is because if $\boldsymbol{f}_1$ and $\boldsymbol{f}_2$ are two minimizers, then by convexity $\ell(\alpha \boldsymbol{f}_1 + (1-\alpha) \boldsymbol{f}_2)  \le \alpha  \ell(\boldsymbol{f}_1) + (1-\alpha)  \ell(\boldsymbol{f}_2)$, and hence
\begin{align} \label{e:convex_ell}
\ell(\alpha \boldsymbol{f}_1 + (1-\alpha) \boldsymbol{f}_2) = \ell(\boldsymbol{f}_1) = \ell(\boldsymbol{f}_2) \mbox{ for all $\alpha \in (0,1)$}.
\end{align}
If $\boldsymbol{f}_1\not=\boldsymbol{f}_2$, then by the strict convexity of $\ell_3$ we have $\ell_3(\alpha \boldsymbol{f}_1 + (1-\alpha) \boldsymbol{f}_2)  < \alpha  \ell_3(\boldsymbol{f}_1) + (1-\alpha)  \ell_3(\boldsymbol{f}_2)$. Since $\ell_1,\ell_2$ and $\ell_4$ are all convex and in view of \eqref{equ:mini1_1}, the relationsip \eqref{e:convex_ell} cannot hold. Thus, $\boldsymbol{f}_1=\boldsymbol{f}_2$.

\end{proof}

\noindent {\bf Proof of Proposition \ref{theorem:eigenvalue_rela}} 
\begin{proof}
    Write the spectral decomposition of $\Tjj$ as $\Tjj = \sum_{k \ge 1} \nu_{jk} \eta_{jk} \otimes \eta_{jk}$ where $\{ \nu_{jk} \}_{k\ge 1}$ are the eigenvalues of $\Tjj$ in decreasing order, and $\{ \eta_{jk} \}_{k \ge 1}$ are the corresponding eigenfunctions. Define 
    \begin{align*}
        \Tjj_{m} = \Pi_{j,m} \Tjj \Pi_{j,m} = \sum_{k=1}^m \nu_{jk} \eta_{jk} \otimes \eta_{jk}
    \end{align*}
    where $\Pi_{j,m} = \sum_{k=1}^m \eta_{jk} \otimes \eta_{jk}$ is the projection operator onto the m-dimensional principal components of $\Tjj$. Recall that $\QSS = \mbox{diag}( \Tjj )_{1 \le j \le q}$. It is straightforward that 
    \begin{align*}
        \QSS_{\alpha, m} = \boldsymbol{\Pi}_{m} \QSS_{\alpha} \boldsymbol{\Pi}_m = \boldsymbol{\Pi}_{m} \left(\QSS + \alpha \mathscrbf{I} \right) \boldsymbol{\Pi}_m,
    \end{align*}
    where $\boldsymbol{\Pi}_m = \mbox{diag}(\Pi_{j,m})_{1 \le j \le q}$. 
    We know that $\QSS_{\alpha, m} \rightarrow \QSS_{\alpha}$ as $m \rightarrow \infty$.
    Define 
    \begin{align*}
        \TSS_{\alpha,m} = \boldsymbol{\Pi}_m \TSS_{\alpha} \boldsymbol{\Pi}_m = \boldsymbol{\Pi}_{m} \left(\TSS + \alpha \mathscrbf{I} \right) \boldsymbol{\Pi}_m.
    \end{align*}
    Note that
    \begin{align}
        \mathscr{T}^{(j_1,j_2)} &= \mathscr{T}^{(j_1,j_2)}_m + \E \left( \Pi_{j_1, m} \wt{X}_{j_1} \otimes \Pi_{j_2, m}^c \wt{X}_{j_2} \right) + \E \left( \Pi_{j_1, m}^c \wt{X}_{j_1} \otimes \Pi_{j_2, m} \wt{X}_{j_2} \right) \label{equ:approx_Tjj} \\
        & \quad\quad + \E \left( \Pi_{j_1, m}^c \wt{X}_{j_1} \otimes \Pi_{j_2, m}^c \wt{X}_{j_2} \right), \nonumber
    \end{align}
    where $\Pi_{j, m}^c = \sum_{k > m} \eta_{jk} \otimes \eta_{jk}$. By Cauchy–Schwarz inequality, for any $f_1, f_2 \in \mathbb{L}_2$,
    \begin{align*}
        \E \left| \left\langle \Pi_{j_1, m} \wt{X}_{j_1}, f_1 \right\rangle_2 \left\langle \Pi_{j_2, m}^c \wt{X}_{j_2} , f_2 \right\rangle_2 \right| &\le \left\| \mathscr{T}^{(j_1, j_1)}_m f_1 \right\|_2 \left\| \left( \mathscr{T}^{(j_2, j_2)} - \mathscr{T}^{(j_2, j_2)}_m \right) f_2 \right\|_2.
    \end{align*}
    As $m$ approaches infinity, the second term on the right-hand side of \eqref{equ:approx_Tjj} converges to 0. Similarly, the third and fourth terms also converge to 0. As a result, we show that $\TSS_{\alpha, m} \rightarrow \TSS_{\alpha}$ as $m \rightarrow \infty$.

    Note that $\TSS_{\alpha, m}$ and $\QSS_{\alpha, m}$ have one-to-one mapping to a vector space of at most $mq$ dimensions. According to \citesupp{lu2000some}, there exists a relationship between the eigenvalues of $\TSS_{\alpha, m}$ and $\QSS_{\alpha, m}$ as follows:
    \begin{align*}
        \Lambda_k \left(\TSS_{\alpha,m} \right) \le \Lambda_k \left(\QSS_{\alpha,m} \right) \vertiii{ \left( \QSS_{\alpha,m} \right)^{-1/2} \TSS_{\alpha,m} \left( \QSS_{\alpha,m} \right)^{-1/2} }_{2,2}.
    \end{align*}
    By the definition of operator norm, 
    \begin{align*}
        &\vertiii{ \left( \QSS_{\alpha,m} \right)^{-1/2} \TSS_{\alpha,m} \left( \QSS_{\alpha,m} \right)^{-1/2} }_{2,2}  \\
        & = \vertiii{ \boldsymbol{\Pi}_{m} \left( \QSS_{\alpha} \right)^{-1/2} \boldsymbol{\Pi}_{m} \TSS_{\alpha} \boldsymbol{\Pi}_{m} \left( \QSS_{\alpha} \right)^{-1/2} \boldsymbol{\Pi}_{m} }_{2,2}  \\
        & \le \vertiii{ \left( \QSS_{\alpha} \right)^{-1/2}  \TSS_{\alpha} \left( \QSS_{\alpha} \right)^{-1/2} \ }_{2,2} \\
        & \le b.
    \end{align*}
    The last inequality holds due to Condition C.\ref{ass:a8}. 
    Finally, let $m \rightarrow \infty$ and $\alpha \rightarrow 0$, we have
    \begin{align*}
        \Lambda_k \left(\TSS_{\alpha,m} \right) \rightarrow \Lambda_k \left(\TSS \right), \quad \Lambda_k \left(\QSS_{\alpha,m} \right) \rightarrow \Lambda_k \left(\QSS \right).
    \end{align*}
\end{proof}

\noindent{\bf Proof of Proposition \ref{theorem:kkt_algoriothm}}
\begin{proof}
 The convex program (\ref{equ:mini_algo}) can be reformulated as a constrained quadratic program
 \begin{align*}
         \min_{\boldsymbol{d}_j \in \mathbb{R}^{M_j}} \left\{ \frac{1}{2} \boldsymbol{d}_j^{\top} \boldsymbol{\Omega}_{j} \boldsymbol{d}_j -  \boldsymbol{\varrho}_j^{\top} \boldsymbol{d}_j \right\}, \quad \mbox{such that} \quad \| \boldsymbol{d}_j \|_2 \le C_1,
 \end{align*}
 where the regularization parameter $\lambda_1$ and constraint level
$C_1$ are in one-to-one correspondence via Lagrangian duality. 
As a result, the above minimization problem involves a continuous finite-dimensional quadratic objective function over a compact set. The Weierstrass’ extreme value theorem guarantees that the minimum is always achieved. According to the Karush-Kuhn-Tucker (KKT) condition to \eqref{equ:mini_algo}
\begin{align}
 \boldsymbol{\Omega}_{j} \boldsymbol{d}_j - \boldsymbol{\varrho}_j + \lambda_1 \boldsymbol{r}_j = \boldsymbol{0}, \label{equ:kkt_algo}
\end{align}
where $\boldsymbol{r}_j$ denotes the sub-gradient of $\| \boldsymbol{d}_j \|_2$ such that $\|\boldsymbol{r}_j \|_2 \le 1$ and $\boldsymbol{r}_j =  \| \boldsymbol{d}_j \|_2^{-1} \boldsymbol{d}_j$ holds for $\boldsymbol{d}_j \neq 0$. When $\| \boldsymbol{\varrho}_j \|_2 \le \lambda_1$, suppose $\boldsymbol{d}_j \neq 0$, according to \eqref{equ:kkt_algo}, we have
\begin{align*}
    \lambda_1 + \Lambda_{\min}\left( \boldsymbol{\Omega}_j \right) \| \boldsymbol{d}_j \|_2 \le \| \boldsymbol{\varrho}_j \|_2 \le  \lambda_1 + \Lambda_{\max}\left( \boldsymbol{\Omega}_j \right) \| \boldsymbol{d}_j \|_2,
\end{align*}
where $\Lambda_{\min}\left( \boldsymbol{\Omega}_j \right)$ and $\Lambda_{\max}\left( \boldsymbol{\Omega}_j \right)$ represent the smallest and largest eigenvalues of the $\boldsymbol{\Omega}_j$, respectively. In order words, when $\| \boldsymbol{\varrho}_j \|_2 \le \lambda_1$, we must have $\boldsymbol{d}_j = 0$. On the other hand, when $\| \boldsymbol{\varrho}_j \|_2 > \lambda_1$, suppose $\boldsymbol{d}_j = 0$, according to \eqref{equ:kkt_algo}, we have $\boldsymbol{\varrho} = \lambda_1 \boldsymbol{r}_j$, and hence $\| \boldsymbol{\varrho}_j \|_2 \le \lambda_1$. This statement presents a contradiction, therefore, 
\begin{align*}
    \left\{
    \begin{aligned}
        \boldsymbol{d}_j = 0, & \quad \mbox{if} \ \| \boldsymbol{\varrho}_j \|_2 \le \lambda_1, \\ 
        \boldsymbol{d}_j \neq 0, & \quad \mbox{if} \ \| \boldsymbol{\varrho}_j \|_2 > \lambda_1.
    \end{aligned}
    \right.
\end{align*}
\end{proof}

\noindent{\bf Proof of Proposition \ref{prop:convex}}
\begin{proof}
To begin with, assume $\mathcal{J}$ is convex and Gateaux differentiable.
Suppose $\mathcal{J}(f) \ge \mathcal{J}(g) + \mathcal{D}_{\mathcal{J}}(g; f-g)$ for all $f, g\in \mathbb{H}$.
Define $h = \lambda f + (1-\lambda) g$, then
$\mathcal{J}(f) \ge \mathcal{J}(h) + \mathcal{D}_{\mathcal{J}}(h; f-h)$ and 
$\mathcal{J}(g) \ge \mathcal{J}(h) + \mathcal{D}_{\mathcal{J}}(h; g-h)$, by the linear combination of the two inequalities, we have:
\begin{align*}
\lambda \mathcal{J}(f) + (1-\lambda) \mathcal{J}(g) \ge
\mathcal{J}(h) + \mathcal{D}_{\mathcal{J}}(h; 0) = \mathcal{J}(\lambda f + (1-\lambda) g),
\end{align*}
which shows convexity. On the other hand, by convexity, for all $f, g \in \mathbb{H}$, $\lambda \in (0,1)$, we have
\begin{align*}
\mathcal{J}(f) - \mathcal{J}(g) \ge \frac{\mathcal{J}(g + \lambda (f-g)) - \mathcal{J}(g)}{\lambda},
\end{align*}
let $\lambda \downarrow 0^{+}$, then the right-hand side will go to $\mathcal{D}_{\mathcal{J}}(g; f-g)$.

To find the global minimum of $\mathcal{J}(\cdot)$, suppose $\mathcal{D}_{\mathcal{J}}(g; \psi) = 0$ for all $\psi \in \mathbb{H}$, then $\mathcal{J}(g) \le \mathcal{J}(f) $ for all $f\in \mathbb{H}$.
On the other hand, by setting $f_1 = g + \tau \psi$, $f_2 = g - \tau \psi$, we have
\begin{align*}
\frac{\mathcal{J}(g) -  \mathcal{J}(g-\tau \psi)}{\tau}  \le  \mathcal{D}_{\mathcal{J}}(g; \psi) \le  \frac{\mathcal{J}(g+\tau \psi) -  \mathcal{J}(g)}{\tau}.
\end{align*}
suppose $\mathcal{J}(g)$ is the global minimum, the left side is smaller than 0 and the right side is greater than 0. By the definition of Gateaux differentiability, the limits on both sides exist and are equal when $\tau \rightarrow 0$. Therefore, $\mathcal{D}_{\mathcal{J}}(g; \psi)=0$ for all $\psi$.

Now assume $\mathcal{J}$ is convex but not Gateaux differentiable.
Then we can easily show $\mathcal{J}(g)$ is the global minimum of $\mathcal{J}$ if and only if $0 \in \partial_{\mathcal{J}(g)}$ using a similar derivation as above.

\end{proof}


\subsection{Proofs of Theorems and Corollary}\label{sec:proof_theorems_corollary}

\noindent \textbf{Proof for Theorem \ref{thm:1}} 

Recall that $\widehat{\boldsymbol{f}}$ is the solution of KKT condition \eqref{equ:kkt} and $\widehat{\mathcal{S}}=\big\{  i \in \{1,\dots,p\}:  \widehat{f}_i \neq 0 \big\}$.
Write $\wt\BX_n = (\wt\BX_{\calS}, \wt\BX_{\calS^c})$ by grouping the columns in $\calS$ and $\calS^c$. 
For $j \in \calS^c$, in the scenario where $\Tjj$ possesses finitely many nonzero eigenvalues, there exist infinitely many $f_{j}$ such that $\la f_{j}, \Tjj f_j \ra_2 = 0$, and those $f_j$ do not make contributions to the response. Without loss of generality, we assume that $\boldsymbol{f}_{0 \calS^c} = \boldsymbol{0}$, and we have $\bdf_0=(\bdf_{0\calS}^\top, {\pmb 0}^\top)^\top$. Similarly, partition $\wh \bdf= (\wh\bdf_\calS^\top, \wh \bdf_{\calS^c}^\top)^\top$, $\bdg_n=(\bdg_\calS^\top, \bdg_{\calS^c}^\top)^\top$ and $\bdomega=(\ws^\top, \wsc^\top)^\top$.
With the partitions defined above and those in Section \ref{sec:theory}, the KKT condition in (\ref{equ:kkt}) can be rewritten as
\begin{align}
\label{equ:kkt0}
\left(
\begin{matrix} 
\TnSS & \TnSSc \\
\TnScS & \TnScSc
\end{matrix}
\right)
 \left(
\begin{matrix} 
\fhats - \fzeros  \\
\fhatsc
\end{matrix}
\right) - 
 \left(
\begin{matrix} 
\gns  \\
\gnsc
\end{matrix}
\right) 
+ \lambda_2  \left(
\begin{matrix} 
\fhats  \\
\fhatsc
\end{matrix}
\right) + \lambda_1
 \left(
\begin{matrix} 
\ws  \\
\wsc
\end{matrix}
\right) = \boldsymbol{0}.
\end{align}

\noindent \textbf{Proof of (i) of Theorem \ref{thm:1}} 

{To utilize the Primal-Dual Witness argument in \cite{wainwright2009sharp}, let $\fchecks$ be the solution of the functional elastic-net problem knowing the true signal set $\calS$. In other words, $\fchecks$ is the value of $\fs$ that minimizes
\begin{align*}
    \frac{1}{2} 
 \left\langle  \TnSS(\fs-\fzeros),  \fs-\fzeros  \right\rangle_{2} -  \left\langle  \gns , \fs-\fzeros \right\rangle_{2} + \sum_{j\in \calS} \mathrm{Pen}(f_j; \lambda_1, \lambda_2).
\end{align*}
Using similar arguments as for Proposition \ref{theorem:kkt} , 
\begin{align}
	\TnSS(\fchecks - \fzeros) -  \gns + \lambda_2\fchecks + \lambda_1  \ws = 0,
 \label{equ:kkt1}
\end{align}
where $\ws=(\Psi_j \eta_j,j\in \calS)$ 
is the functional subgradient of $\ell_4$ for this problem described in Proposition \ref{theorem:kkt} and \ref{pp:L1_L3}. For convenience, let 
$\bdeta_{\cal W}=(\eta_j, j\in {\cal W})$ for any set ${\cal W}$.
By Proposition \ref{theorem:kkt}, the solution to the functional elastic-net problem is unique and satisfies the KKT equation \eqref{equ:kkt0}. If we can show that $ \left(\fchecks^{\top}, \boldsymbol{0}^\top\right)^{\top}$ solves \eqref{equ:kkt0}, then $\widehat{\boldsymbol{f}}=\left(\fchecks^{\top}, \boldsymbol{0}^\top\right)^{\top}$ and $\wh \calS\subset \calS$. It remains to show
\begin{align}
	& \TnScS(\fchecks - \fzeros) -  \gnsc + \lambda_1  \wsc = 0,  \label{equ:kkt2}
\end{align}
for some $\wsc$ satisfying $\wsc=(\Psi_j \eta_j, j\in \calS^c)$ where $\|\boldsymbol{\eta}_{\mathcal{S}^c} \|_{\infty}\le 1$.  
However, by (\ref{equ:kkt1}), 
\begin{align}
	\fchecks - \boldsymbol{f}_{0\mathcal{S}} =   \left(\JnSS\right)^{-1}  \left( \gns - 
 \lambda_2 \fzeros- \lambda_1 \ws \right),
 \label{equ:f_diff}
\end{align}
and hence, upon combining (\ref{equ:kkt2}) and (\ref{equ:f_diff}), any $\wsc$ that solves \eqref{equ:kkt2} must satisfy
\begin{align}\label{equ:omega_S_c1}
\begin{split}
	\wsc  &:= \frac{1}{\lambda_1} \left\{ \gnsc - \TnScS  \left(\JnSS\right)^{-1} \gns \right\} \\
& \hspace{.6cm} + \TnScS \left(\JnSS\right)^{-1} \left(\frac{\lambda_2}{\lambda_1} \fzeros+\ws\right). 
\end{split}
\end{align}
By Condition C.\ref{ass:a1}, the existence of $\wsc$ satisfying \eqref{equ:kkt2} is guaranteed by
\begin{align} \label{e:thm1_i}
	\|\wsc\|_{\infty} \le C_{\min}.
\end{align}
}
The rest of this subsection will be focusing on \eqref{e:thm1_i}.

It is easy to see that, for any $\boldsymbol{f} \in \mathbb{L}_2^q[0,1]$, $	(\TnSS \boldsymbol{f})(t) = {1\over n} \int \wt \BX_{\calS}^{\top}(t) \wt\BX_{\calS}(u) \boldsymbol{f}(u)du$. 
The first term on the right-hand side of \eqref{equ:omega_S_c1} can be rewritten as
$	(\lambda_1 {n})^{-1} \wt \BX_{\calS^c}^{\top} ( \mathrm{I} - \boldsymbol{\Delta}_n ) \boldsymbol{\varepsilon}_n $,
where
\begin{align}
 \boldsymbol{\Delta}_n =  {1\over n} \int \wt \BX_{\calS}(u) \left\{ \left(\JnSS\right)^{-1} \wt \BX_{\calS}^{\top}\right\} (u) du.
 \end{align}
Thus, for all $ j \in \calS^c$,
\begin{align}
 \label{equ:ineq1}
\begin{split}
	\| \omega_j \|_{2} 
	& =  \bigg\| \frac{\sigma}{\lambda_1 {n}} \wt \BX_{\bullet j}^{\top} \bigg( \mathrm{I} - \boldsymbol{\Delta}_n \bigg) \boldsymbol{z}_n  + \TnjS \left(\JnSS\right)^{-1} \left(\frac{\lambda_2}{\lambda_1} \fzeros+\ws\right) \bigg\|_{2} \\
	& \le  \bigg\| \frac{\sigma}{\lambda_1 {n}} \wt \BX_{\bullet j}^{\top} \bigg( \mathrm{I} - \boldsymbol{\Delta}_n \bigg) \boldsymbol{z}_n \bigg\|_{2} 
	+ \bigg\| \TnjS \left(\JnSS\right)^{-1} \left(\frac{\lambda_2}{\lambda_1} \fzeros+\ws\right) \bigg\|_{2},
 \end{split}
\end{align}
where $\boldsymbol{z}_n = \sigma^{-1} \boldsymbol{\varepsilon}_n$ has covariance matrix equal to an identity matrix. 
If $\widehat{\calS} \not\subset \calS$ then \eqref{e:thm1_i} fails, and, by Lemmas \ref{lm:(i)} and \ref{lm:proof(ii)} below, 
\begin{align} 
\begin{split}
\mathbb{P} \left( \widehat{\calS} \not\subset \calS \right) 
& \le \mathbb{P} \left( \|\wsc\|_{\infty} > \left(1 - \frac{\gamma}{9}\right)C_{\min} \right) \\
& \le \exp\left( - D^{(1)} \lambda_1^2 n  \right)
+\exp\left( - D^{(2)} {\lambda_2^2 n \over q} \right).
\end{split}
    \label{equ:tail_bound_final_1}
\end{align}
Note that $\exp\left( - D^{(1)} \lambda_1^2 n  \right) \le \exp\left( - D^{(1)} C_{\max}^{-2} q {\lambda_2^2 n \over q}  \right)$ since $\lambda_1> C_{\max}^{-1} \lambda_2$. Applying Lemma \ref{lm:tail_sum} with $\epsilon=1/2$, we can bound the rhs of \eqref{equ:tail_bound_final_1} by the probability in \eqref{equ:vs_consistency_rate}, provided $\lambda_2^2 n / q > (2 \log 2) D^{-1}$, which is guaranteed by Condition \eqref{equ:lambda_12_bound_1} for sufficiently large $D_{2,2}^*$ in the condition.

To conclude the proof of part (i) of Theorem \ref{thm:1}, it remains to establish the following three lemmas, the proofs of which are in the Supplemental Material. 

\begin{lemma} \label{lm:tail_sum}
For $a_k, b_k > 0$, $k = 1, \dots, K$,
\begin{align*}
    \sum_{k=1}^K a_k \exp(-b_k x) \le \exp\left\{ - (1-\epsilon) b x \right\}
\end{align*}
for $x > (\epsilon b)^{-1} \log(Ka)$, where $\epsilon \in (0, 1)$, $a = \max_k a_k$ and $b = \min_k b_k$.
\end{lemma}

\begin{lemma} \label{lm:(i)} 
Let $\gamma$ be as in Condition C.\ref{ass:a3}. Suppose $\lambda_1 >  D_1^* (\sigma+1) \tau^{1/2}  (C_{\min} \gamma)^{-1} \sqrt{{\log(p-q) \over n}}$ for some constant $D_1^*$.
We have
\begin{align*}
	& \pr\left(\mathrm{max}_{j \in \calS^c} \bigg\| \frac{\sigma}{\lambda_1 {n}} \wt \BX_{\bullet j}^{\top} \bigg( \mathrm{I} - \boldsymbol{\Delta}_n \bigg) \boldsymbol{z}_n  \bigg\|_{2} \ge \frac{ \gamma C_{\min}}{9}\right) \le \exp\left( - D^{(1)} \lambda_1^2 n  \right)
\end{align*}
where $D^{(1)} = D_2^* C_{\min}^2 \gamma^2 (\sigma + 1)^{-2} \tau^{-1}$ and $D_2^*$ is a universal constant.
\end{lemma}

\begin{lemma} \label{lm:proof(ii)} 
Let $\gamma$ be as in Condition C.\ref{ass:a3}. Suppose, for some constant $D_1^*$,
\begin{align*}
    \lambda_2 > D_1^* \frac{\tau (\rho_1 + 1 )}{(C_{\min}/C_{\max})^2 \gamma^2}  \max\left({q\log(p-q)\over n},\sqrt{q^2\over n} \right) 
    \quad \mbox{and} \quad \frac{\lambda_1}{\lambda_2} > \left({3\over\gamma} - 2 \right) C_{\max}^{-1}.
\end{align*}
Then
\begin{align*} 
	& \pr\left\{   \max_{j \in \calS^c} \left\| \TnjS \left(\JnSS\right)^{-1} \left(\frac{\lambda_2}{\lambda_1} \fzeros+\ws\right) \right\|_{2}
 \ge \left(1-{2\gamma \over 9}\right) C_{\min} \right\} \le \exp\left( - D^{(2)} {\lambda_2^2 n \over q} \right)
\end{align*}
where $D^{(2)} = D_2^*  (C_{\min}/C_{\max})^2 \gamma^2 (\rho_1+1)^{-2} \tau^{-1}$ and $D_2^*$ is a universal constant.
\end{lemma}

%
%
%
%
%
%
%
%
%
%
%
%
%
%
%

\noindent \textbf{Proof of (ii) of Theorem \ref{thm:1}} 

We need to show that $\| \widehat{f}_j \|_{2} > 0$ for all $j \in \calS_G$ with the probability lower bound stated in the theorem. For simplicity, assume that $\calS_G = \calS$. The same arguments hold if $\calS$ is replaced by $\calS_G$ below.

Note that $\pr(\wh \calS \supset \calS)= \pr( \min_{ j \in \calS} \| \widehat{f}_j \|_{2} > 0) \ge \pr(  \min_{j\in \calS} \| (\mathcal{T}^{(j,j)} )^{1/2} \widehat{f}_j \|_{2} > 0)$. By the triangle inequality,
\begin{align*}
	\min_{j \in \calS} \left\| (\Tjj)^{1/2} \widehat{f}_j \right\|_2 
	& \ge \min_{j \in \calS} \left\| (\Tjj)^{1/2} f_{0j} \right\|_2 - \max_{j \in \calS} \left\| (\Tjj)^{1/2} (\widehat{f}_j - f_{0j}) \right\|_2 \\
	& \ge G - \max_{j \in \calS} \left\| (\Tjj)^{1/2} (\widehat{f}_j - f_{0j}) \right\|_2.
\end{align*}
Thus, it suffices to provide an upper bound for $\pr\left(\max_{j \in \calS} \| (\Tjj)^{1/2} (\widehat{f}_j - f_{0j}) \|_2>G\right)$.
By \eqref{equ:f_diff}, we have
\begin{align*}
	\fchecks - \fzeros =&  (\JSS)^{-1} \bigg( \gns -  \lambda_2 \fzeros - \lambda_1 \ws \bigg) \\
	& + \bigg\{ (\JnSS)^{-1} - (\JSS)^{-1} \bigg\} \bigg( \gns - \lambda_2 \fzeros - \lambda_1 \ws \bigg).
\end{align*}
Since
$(\JnSS)^{-1} - (\JSS)^{-1} = (\JSS)^{-1} \left(\TSS - \TnSS \right)(\JnSS)^{-1}$,
\begin{align*}
	& \max_{j \in \calS} \left\|(\mathcal{T}^{(j,j)} )^{1/2} (\check{f}_j - f_{0j})\right\|_{2}  = \left\|(\QSS)^{1/2}(\check{\bdf}_{\calS}-\boldsymbol{f}_{0 \calS})\right\|_\infty \\
 & \le  \vertiii{ (\QSS)^{1/2}  (\JSS)^{-1}}_{\infty,\infty} \left\{ \left\| \gns \right\|_{\infty} + \lambda_2 \left( \|\fzeros\|_{\infty} + \frac{\lambda_1}{\lambda_2} C_{\max} \right) \right\} \\
    & \hspace{.5cm} \times \left( 1 + \frac{\sqrt{q}}{\lambda_2} \vertiii{\TSS - \TnSS}_{2,2} \right),
\end{align*}
where we applied the inequality 
\begin{align*}
\vertiii{ (\TSS - \TnSS) (\JnSS)^{-1}}_{\infty,\infty}  \le \frac{\sqrt{q}}{\lambda_2} \vertiii{\TSS - \TnSS}_{2,2}.
\end{align*}
By Lemma \ref{lm:(iii)} with $\| \fzeros \|_{\infty} = 1$, 
\begin{eqnarray}\label{equ:f_norm_diff}
    && \max_{j \in \calS} \left\|(\mathcal{T}^{(j,j)} )^{1/2} (\widehat{f}_j - f_{0j})\right\|_{2} \\
&& \le \frac{6-4\aleph(\lambda_2)}{(1-\aleph(\lambda_2))\sqrt{\lambda_2}} \left( \left\| \gns \right\|_{\infty} + \lambda_2 + C_{\max}\lambda_1 \right) \left( 1 + \frac{\sqrt{q}}{\lambda_2} \vertiii{\TSS - \TnSS}_{2,2} \right). \nonumber
\end{eqnarray}
Thus, with $G$ as given in the theorem,
\begin{align}
\begin{split}
    & \pr\left(\max_{j \in \calS} \left\| (\Tjj)^{1/2} (\widehat{f}_j - f_{0j}) \right\|_2>G\right) \\
    & \le \pr\left( \left\| \gns \right\|_{\infty} > \lambda_2\right) + \pr\left(\frac{\sqrt{q}}{\lambda_2} \vertiii{\TSS - \TnSS}_{2,2}>1\right).
\end{split}
\label{equ:tail_bound_final_2}
\end{align}
Finally, bound the rhs of \eqref{equ:tail_bound_final_2} using Lemmas \ref{lm:(iv)} and \ref{lemma:tail_bound} and note that it is dominated by the expression in \eqref{equ:vs_consistency_rate}
under Condition \eqref{equ:lambda_12_bound_1}.

\begin{lemma} \label{lm:(iii)}
Under Condition C.\ref{ass:a6}, for any $\lambda_2 > 0$
\begin{align}
\vertiii{ (\QSS)^{1/2}  (\JSS)^{-1}}_{\infty,\infty} < \frac{6-4\aleph(\lambda_2)}{1-\aleph(\lambda_2)} \frac{1}{\sqrt{\lambda_2}}
\label{equ:bound_4}
\end{align}
\end{lemma}

\begin{lemma} \label{lm:(iv)}
Suppose $\lambda_2 >  D_1^* (\sigma+1) \tau^{1/2}  \sqrt{\log q \over n}$, we have
\begin{align}
\mathbb{P} \left( \left\|  \gns \right\|_{\infty} \ge \lambda_2 \right)
\le   \exp \left( - D^{(3)} \lambda_2^2 n  \right) 
\label{equ:tail_bound_3}
\end{align}
holds for some $D^{(3)} < D_2^* \left( (\sigma + 1)^{2} \tau \right)^{-1} $ where $D_1^*$ and $D_2^*$ are universal constants.
\end{lemma}

\begin{lemma} \label{lemma:tail_bound}
Suppose $\rho_1$ is the largest eigenvalue of $\TSS$, then
\begin{align*}
\pr \left( \sqrt{q} \vertiii{\TSS - \TnSS}_{2,2} > u \right) \le \exp\left\{ -\frac{u^2 n}{C^2 \rho_1^2 q}  \right\}
\end{align*}
holds for some constant $C>0$, as long as $C$ and $q$ satisfy
\begin{align} \label{e:q_cond}
\frac{u^2}{C^2 \rho_1^2} < q \le  \sqrt{\frac{u^2 n }{\tau C^2 \rho_1}}.
\end{align}
\end{lemma}

The proofs for Lemmas \ref{lm:(iii)} - \ref{lemma:tail_bound} are in the Supplementary Material.

\noindent{\bf Proof of Theorem \ref{c:excess}}
\begin{proof}
   When $\widehat{\calS} \subset \calS$, the excess risk has the form
\begin{align*}
\mathcal{R}(\widehat{\boldsymbol{\beta}}) &= \mathbb{E}^{*}\left[ \sum_{j \in \calS } \langle X_{j}^*,  \beta_{0j} - \widehat{\beta}_j \rangle_{2}\right]^{2} = \left\| \left( \TSS \right)^{1/2} (\fzeros - \fhats)  \right\|_2^2.
\end{align*}
Following a similar derivation as in \eqref{equ:f_norm_diff},
\begin{align*}
& \left\| \left( \TSS \right)^{1/2} (\fzeros - \fhats)  \right\|_2 \\
& \le \sqrt{q} \vertiii{ (\TSS)^{1/2}  (\JSS)^{-1}}_{2,2} \left\{ \left\| \gns \right\|_{\infty} + \lambda_2 \left( \|\fzeros\|_{\infty} + \frac{\lambda_1}{\lambda_2} C_{\max} \right) \right\} \\
& \hspace{.5cm} \cdot \left( 1 + \frac{1}{\lambda_2} \vertiii{\TSS - \TnSS}_{2,2} \right). 
\end{align*}
Similar to \eqref{equ:bound_b},
\begin{align*}
\vertiii{ (\TSS)^{1/2}  (\JSS)^{-1}}_{2,2} \le \frac{1}{2\sqrt{\lambda_2}},
\end{align*}
together with a similar derivation as for \eqref{equ:f_norm_diff} with $\| \fzeros \|_{\infty} = 1$, we have 
\begin{align*}
 \mathcal{R}(\widehat{\boldsymbol{f}}) 
\le q (2 + C_{\max} \lambda_1 / \lambda_2)^2 \lambda_2 = q \left(4 C_{\max} \lambda_1 + 4 \lambda_2 + C_{\max}^2 \lambda_1^2 / \lambda_2 \right) 
\end{align*}
with probability greater than \eqref{equ:vs_consistency_rate}. 
\end{proof}

\noindent{\bf Proof of Corollary \ref{cor:R(f)}}
\begin{proof}
    Recall
\begin{align*} 
\alpha(p,q,n) = \max\left(q,\sqrt{\log (p-q)}, \sqrt{q\log n} \right),
\end{align*}
and define
$$
\ell_n = C q^{-1} \alpha^2(p,q,n)
$$
where $C$ is a large enough constant.
Let 
$$
\lambda_2=(\ell_n q/n)^{1/2} = C^{1/2} {1\over\sqrt{n}} \alpha(p,q,n)
$$ 
and $\lambda_1 = b\lambda_2$ for some suitable constant $b$. 
If $q^2\log(p-q)\le n$ for large $n$, which is guaranteed by the assumption $q\alpha(p,q,n) = o(n^{1/2})$, we have for all large $n$,
\begin{align*} 
\alpha(p,q,n) = \max\left(q, {q\log(p-q)\over\sqrt{n}},
\sqrt{\log (p-q)}, \sqrt{q\log n} \right),
\end{align*}
from which it is easy to see that \eqref{equ:lambda_12_bound_1} holds for $b,C$ sufficiently large.
Note that $\ell_n \ge C\log n$.
By Theorem \ref{c:excess}, the excess risk is bounded by a constant multiple of 
\begin{align*}
\lambda_2q 
= C^{1/2}{q\over\sqrt{n}} \alpha(p,q,n)
\end{align*}
where probability at least $1-n^{-D}$ for some constant $D$. The claim of the corollary follows from the Borel-Cantelli Lemma by choosing a large enough $C$ and hence $D>1$.
\end{proof}

%
%
%
%
%
%
%
%
%
%
\noindent{\bf Proof of Theorem \ref{thm:minimax_lower}}
\begin{proof}

Recall that the excess risk for an estimator $\ftildets$ has the form 
\begin{eqnarray*}
    \mathcal{R}(\ftildets) = \left\| \left( \TSS \right)^{1/2} (\fzeros - \ftildets)  \right\|_2^2.
\end{eqnarray*}
For any $\boldsymbol{f}_1, \boldsymbol{f}_2 \in \mathbb{{L}}_2^q$, define 
 \begin{align*}
       \mathcal{D}(\boldsymbol{f}_1, \boldsymbol{f}_2) = \left\| \left( \TSS \right)^{1/2} (\boldsymbol{f}_1 - \boldsymbol{f}_2)  \right\|_2, 
\end{align*}
which is a proper metric in $\mathbb{{L}}_2^q$.
Write the spectral decomposition of $\TSS$ as $\TSS = \sum_{k \ge 1} \rho_k \boldsymbol{\phi}_k \otimes \boldsymbol{\phi}_k$, where $\rho_1 \ge \rho_2 \ge \dots \ge 0$.  
By Corollary \ref{corollary:eigenvalue_rela}, for any covariance operator $\TSS \in \mathscrbf{P}(r) $, its eigenvalues satisfy $\rho_{q(k-1)+j} \le C k^{-2r}$ for some constant $C>0$. Consider sub-class of covariance operators, denoted as $\mathscrbf{P}(r,C, C')$ for some $0<C'<C<\infty$, which include all $\TSS$ with $C' k^{-2r} \le \rho_{q(k-1)+j} \le C k^{-2r}$.
It is straightforward to show that for $k > q$, $c_1 (k/q)^{-2r} \leq \rho_k \leq c_2 (k/q)^{-2r}$ for some $0 < c_1 < c_2 < \infty$.

As noted in \cite{CaiYuan2012} in the proof of their Theorem 1, any lower bound derived under a specific case yields a lower bound for the general case. For the rest of the proof, we will consider a special case
where $\TSS \in \mathscrbf{P}(r,C, C')$ and the functional coefficient in the oracle model has the form
\begin{align}
    \boldsymbol{\beta}_{\theta} = \scrbfL_{\BK_{\calS}^{1/2}} \boldsymbol{f}_{\theta}, \quad \boldsymbol{f}_{\theta} = M^{-1/2} \sum_{k=M+1}^{2M} \theta_{k} \boldsymbol{\phi}_k. \label{equ:f_theta}
\end{align}
where $\boldsymbol{\theta} = \left(\theta_{M+1}, \ldots, \theta_{2 M}\right) \in \{0,1\}^{M}$ for some large integer $M$.
The Varshamov–Gilbert bound (Lemma 2.9, \citesupp{Tsybakov2009}) shows that for any $M \ge 8$, there exists a subset $\boldsymbol{\Theta}_0 = \{ \boldsymbol{\theta}^{(0)}, \boldsymbol{\theta}^{(1)}, \dots, $ $\boldsymbol{\theta}^{(N)} \} \in \{0, 1\}^{M}$ such that 
(a) $\boldsymbol{\theta}^{(0)} = (0, \dots, 0)^{\top}$;  
(b) $ H(\boldsymbol{\theta}^{(j)}, \boldsymbol{\theta}^{(k)}) \ge M/8$  for any $\ 0 \le j < k \le N $, $H(\cdot, \cdot)$ is the Hamming distance; and (c) $N \ge 2^{M/8}$.
Because $\{\boldsymbol{f}_{\theta}: \boldsymbol{\theta} \in \boldsymbol{\Theta}_0\} \subset \mathbb{{L}}_2^q$, it is clear that $\forall B > 0$
\begin{align}\label{eq:special_case}
\begin{split}
    & \sup_{\TSS \in \mathscrbf{P}(r) } \ \sup_{\fzeros \in \mathbb{{L}}_2^q}  \mathbb{P}\left( \mathcal{D}(\ftildets, \fzeros) \ge B \right) \\ 
    &\quad \ge  \sup_{\TSS \in \mathscrbf{P}(r,C, C')} \ \max_{\boldsymbol{\theta} \in \boldsymbol{\Theta}_0}  \mathbb{P}_{\theta}\left( \mathcal{D}(\ftildets, \boldsymbol{f}_{\theta}) \ge B \right).
    \end{split}
\end{align}
Here, $\mathbb{P}_{\theta}$ is the probability measure when the function coefficient has the form given in \eqref{equ:f_theta}. 

Next, we proceed to establish the lower bound under the special case using results in Theorem 2.5 of \citesupp{Tsybakov2009}.
To that end, for any $\boldsymbol{\theta}, \boldsymbol{\theta}' \in \boldsymbol{\Theta}_0$ such that $\boldsymbol{\theta} \neq \boldsymbol{\theta}'$, the Kullback–Leibler distance between $\mathbb{P}_{\theta}$ and $\mathbb{P}_{\theta'}$ is given by
\begin{align*}
    & \mathcal{K}\left(\mathbb{P}_{\theta} \| \mathbb{P}_{\theta'} \right) = \frac{n}{2 \sigma^2} \mathcal{D}^2(\boldsymbol{f}_{\theta}, \boldsymbol{f}_{\theta'}) =  \frac{n}{2\sigma^2 M} \sum_{k=M+1}^{2M} (\theta_{k} - \theta'_{k})^2 \rho_k \\
    & \quad\quad\quad \le \frac{n \rho_M}{2\sigma^2 M} H(\boldsymbol{\theta}, \boldsymbol{\theta}') \le \frac{n \rho_M}{2\sigma^2 }. 
\end{align*}
For any $0 < \alpha < 1/8$, let $M = \lceil c_{0} n^{1 /(2 r+1)} q^{2r/(2r+1)} \rceil$ and $c_0 = D \alpha^{-1/(2r+1)}  $ for some large enough $D > 0$, then  
\begin{align*}
    \frac{1}{N}\sum_{k=1}^N \mathcal{K}\left(\mathbb{P}_{\theta^{(k)}} \| \mathbb{P}_{\theta^{(0)}} \right) \le \frac{c_2}{2 \sigma^2} n \left( \frac{M}{q} \right)^{-2r} \le \frac{c_2 c_0^{-(2r+1)}}{2 \sigma^2}  M \le \alpha \log N.
\end{align*}
On the other hand, 
\begin{align*}
    \mathcal{D}^2(\boldsymbol{f}_{\theta}, \boldsymbol{f}_{\theta'})  \ge  \frac{\rho_{2M}}{M}  H(\boldsymbol{\theta}, \boldsymbol{\theta}') \ge \frac{\rho_{2M}}{8} \ge \frac{c_1}{8} \left(\frac{2M}{q} \right)^{-2r} \ge 4 d {\alpha}^{\frac{2r}{2r+1}} \left( n/q \right)^{-\frac{2r}{2r+1}},
\end{align*}
for some small enough $d > 0$. By Theorem 2.5 in \citesupp{Tsybakov2009} we have
\begin{align*}
    & \inf_{ \wt{\boldsymbol{f}}_{\calS} } \ \sup_{\TSS \in \mathscrbf{P}(r,C, C')} \ \max_{\boldsymbol{\theta} \in \boldsymbol{\Theta}_0}  \mathbb{P}_{\theta}\left( \mathcal{D}^2(\ftildets, \boldsymbol{f}_{\theta}) \ge d {\alpha}^{\frac{2r}{2r+1}} (n / q)^{-\frac{2r}{2r+1}} \right) \\
    & \quad\quad \ge \frac{\sqrt{N}}{1+\sqrt{N}}\left(1-2 \alpha-\sqrt{\frac{2 \alpha}{\log N}}\right).
\end{align*}
Letting $a = d {\alpha}^{2r/(2r+1)}$, we have
\begin{align}\label{eq:lower_bound_special_case}
\lim _{a \rightarrow 0} \lim _{n \rightarrow \infty} \inf_{ \wt{\boldsymbol{f}}_{\calS} } \ \sup_{\TSS \in \mathscrbf{P}(r,C, C')} \ \max_{\boldsymbol{\theta} \in \boldsymbol{\Theta}_0}  \mathbb{P}_{\theta}\left( \mathcal{D}^2(\ftildets, \boldsymbol{f}_{\theta}) \ge a (n / q)^{-\frac{2r}{2r+1}} \right) = 1.
\end{align}
The minimax lower bound result in the theorem is derived by combining \eqref{eq:special_case} and \eqref{eq:lower_bound_special_case}.
\end{proof}

%
%
%
%
%
%
%
%
%
%
\noindent{\bf Proof of Theorem \ref{thm:minimax_upper}}
\begin{proof}
We first note that
\begin{align*}
    & \mathbb{P}\left( \mathcal{R}( \fhaths ) \ge B  \right) - \mathbb{P}\left( \mathcal{R}( \widehat{\boldsymbol{f}}_{\calS})\ge B \right) \\
    & =  \mathbb{P}\left( \mathcal{R}( \fhaths ) \ge B, \widehat{\calS}\not=\calS\right)
    - \mathbb{P}\left( \mathcal{R}( \widehat{\boldsymbol{f}}_{\calS})\ge B, \widehat{\calS}\not=\calS\right).
\end{align*}
Thus, as long as $\mathbb{P}\left(\widehat{\calS}\not=\calS\right)\to 0$, we have
\begin{align*}
    & \lim_{n \rightarrow \infty} \ \sup_{\TSS} \ \sup_{\fzeros \in \mathbb{L}_2^q}  \mathbb{P}\left( \mathcal{R}( \fhaths ) \ge B  \right) 
    = \lim_{n \rightarrow \infty}  \ \sup_{\TSS} \ \sup_{\fzeros \in \mathbb{L}_2^q}   \mathbb{P}\left( \mathcal{R}( \widehat{\boldsymbol{f}}_{\calS})\ge B \right).
\end{align*}
From \eqref{equ:object_ridge}, we can easily derive that
    \begin{align*}
        \fhatts = \left(\TSS_{n,\lambda_3} \right)^{-1} \left\{ \TnSS \fzeros + \gns \right\},
    \end{align*}
where $\fzeros$ and $\gns$ are defined in \eqref{equ:kkt0}.
Define $\ftildets = \left(\TSS_{\lambda_3} \right)^{-1} \TSS \fzeros$,
then
\begin{align}
       \mathcal{R}^{1/2}(\fhatts) &= \left\| \left( \TSS \right)^{1/2} (\fzeros - \fhatts)  \right\|_2 \nonumber \\
       &\le \left\| \left( \TSS \right)^{1/2} (\fzeros - \ftildets)  \right\|_2 + \left\| \left( \TSS \right)^{1/2} (\ftildets - \fhatts)  \right\|_2. \label{equ:thm4_bound}
\end{align}
By Lemma \ref{lemma:minimax_upper_1}, the first term in \eqref{equ:thm4_bound} can be bounded by $\frac{1}{2} \lambda_3^{1/2} \left\| \fzeros \right\|_2$. In order to bound the second term, note that
\begin{align*}
    \ftildets - \fhatts &= \left(\TSS_{\lambda_3} \right)^{-1} \TSS_{n,\lambda_3} \left( \ftildets - \fhatts \right) \\
    & \quad + \left(\TSS_{\lambda_3} \right)^{-1} \left( \TSS - \TnSS  \right) \left( \ftildets - \fhatts \right) \\
    &= \left(\TSS_{\lambda_3} \right)^{-1} \TnSS \left( \ftildets - \fzeros \right) + \lambda_3 \left(\TSS_{\lambda_3} \right)^{-1} \ftildets  \\
    & \quad - \left(\TSS_{\lambda_3} \right)^{-1} \gns + \left(\TSS_{\lambda_3} \right)^{-1} \left( \TSS - \TnSS  \right) \left( \ftildets - \fhatts \right) \\
    &= \left(\TSS_{\lambda_3} \right)^{-1} \TSS \left( \ftildets - \fzeros \right) \\ 
    & \quad + \left(\TSS_{\lambda_3} \right)^{-1} \left( \TnSS - \TSS  \right) \left( \ftildets - \fzeros \right) \\
    & \quad + \lambda_3 \left(\TSS_{\lambda_3} \right)^{-2} \TSS \fzeros 
 - \left(\TSS_{\lambda_3} \right)^{-1} \gns \\
    & \quad + \left(\TSS_{\lambda_3} \right)^{-1} \left( \TSS - \TnSS  \right) \left( \ftildets - \fhatts \right) .
\end{align*}
Therefore, by the triangular inequality,
\begin{align} \label{e:longdisplay}
\begin{split}
    & \left\| \left( \TSS \right)^{\nu_1} (\ftildets - \fhatts)  \right\|_2 \\
    & \le \left\| \left( \TSS \right)^{\nu_1} \left(\TSS_{\lambda_3} \right)^{-1} \TSS \left( \ftildets - \fzeros \right)  \right\|_2 \\
    & \quad + \left\| \left( \TSS \right)^{\nu_1} \left(\TSS_{\lambda_3} \right)^{-1} \left( \TnSS - \TSS  \right) \left( \ftildets - \fzeros \right)  \right\|_2 \\
    & \quad + \left\| \lambda_3 \left( \TSS \right)^{\nu_1} \left(\TSS_{\lambda_3} \right)^{-2} \TSS \fzeros \right\|_2  \\
    & \quad + \left\| \left( \TSS \right)^{\nu_1} \left(\TSS_{\lambda_3} \right)^{-1} \gns \right\|_2 \\
    & \quad + \left\| \left( \TSS \right)^{\nu_1} \left(\TSS_{\lambda_3} \right)^{-1} \left( \TSS - \TnSS  \right) \left( \ftildets - \fhatts \right) \right\|_2.
    \end{split}
\end{align}
For convenience, define
\begin{align*}
    A(\nu) &= \left\| \left( \TSS \right)^{\nu} (\ftildets - \fhatts)  \right\|_2, \\
    B_1(\nu) &= \left\| \left( \TSS \right)^{\nu} (  \ftildets - \fzeros )  \right\|_2, \\
    B_2(\nu_1,\nu_2) &=\vertiii{ \left( \TSS \right)^{\nu_1} \left(\TSS_{\lambda_3} \right)^{-1} \left( \TnSS - \TSS  \right) \left( \TSS \right)^{-\nu_2}  }_{2,2}, \\
    B_3(\nu) &= \vertiii{ \lambda_3 \left( \TSS \right)^{\nu} \left(\TSS_{\lambda_3} \right)^{-1}}_{2,2} \left\| \fzeros \right\|_2, \\
    B_4(\nu) &=\left\| \left( \TSS \right)^{\nu} \left(\TSS_{\lambda_3} \right)^{-1} \gns \right\|_2.
    \end{align*}
Then, \eqref{e:longdisplay} may be further developed as
\begin{align} \label{e:longdisplay_1}
A(\nu_1) \le B_1(\nu_1) + B_2(\nu_1,\nu_2) B_1(\nu_2) + 
B_3(\nu_1) + B_4(\nu_1) + B_2(\nu_1,\nu_2) A(\nu_2).
\end{align}
According to Lemma \ref{lemma:minimax_upper_1}-\ref{lemma:minimax_upper_3}, for $0 < \nu \le 1/2$,
\begin{align*}
    B_1(\nu), 
    \ B_3(\nu), 
    \ B_4(\nu) = O_p(\lambda_3^{\nu}), \ 
     B_4(\nu) = O_{p}\left( \left( \frac{n}{q}  \lambda_3^{1-2 \nu + \frac{1}{2r} }\right)^{-\frac{1}{2}} \right).
\end{align*}
First, let $\nu_1 = \nu_2 = \nu$ in \eqref{e:longdisplay_1}, where $0 < \nu < 1/2 - 1/(4r)$. 
According to Lemma \ref{lemma:minimax_upper_4}, 
\begin{align} \label{equ:thm4_order}
\begin{split}
   B_2(\nu,\nu) & = O_{p}\left( q^{\frac{1}{2}}\left( \frac{n}{q}  \lambda_3^{1-2 \nu + \frac{1}{2r} }\right)^{-\frac{1}{2}}\right) =  O_p \left(q^{\frac{1}{2}} \lambda_3^{\nu} \right)
   = o_p(1), \\
   B_4(\nu) &= O_p(\lambda_3^\nu),
   \end{split}
\end{align}
provided that $\lambda_3 \asymp (n/q)^{-2r/(2r+1)}$ and $q = o\left(n^{\frac{4r\nu}{2r+1+4r\nu}} \right)$. In this case, combining the last term on the rhs of \eqref{e:longdisplay_1} with the lhs, we obtain
\begin{align} \label{e:first_estimate}
\left\| \left( \TSS \right)^{\nu} (\ftildets - \fhatts)  \right\|_2  = O_p(\lambda_3^{\nu})
\end{align}
provided that $\lambda_3 \asymp (n/q)^{-2r/(2r+1)}$.

Next, we let $\nu_1 = 1/2$, $\nu_2 = \nu\in (0,1/2 - 1/(4r))$ in \eqref{e:longdisplay_1}. According to Lemma \ref{lemma:minimax_upper_4}, 
\begin{align*}
B_2(1/2,\nu) = O_{p}\left( q^{\frac{1}{2}}\left( \frac{n}{q}  \lambda_3^{\frac{1}{2r} }\right)^{-\frac{1}{2}}\right)
    =  O_p \left(q^{\frac{1}{2}} \lambda_3^{\frac{1}{2}} \right)
    = o_p \left(q^{\frac{1}{2}} \lambda_3^{\nu} \right)
\end{align*}
provided that $\lambda_3 \asymp (n/q)^{-2r/(2r+1)}$.

\begin{align*}
    & \vertiii{ \left( \TSS \right)^{1/2} \left(\TSS_{\lambda_3} \right)^{-1} \left( \TnSS - \TSS  \right) \left( \TSS \right)^{-\nu}  }_{2,2} \\
    &= O_{p}\left( q^{\frac{1}{2}}\left( \frac{n}{q}  \lambda_3^{\frac{1}{2r} }\right)^{-\frac{1}{2}}\right)
    =  O_p \left(q^{\frac{1}{2}} \lambda_3^{\frac{1}{2}} \right)
    = o_p \left(q^{\frac{1}{2}} \lambda_3^{\nu} \right)
\end{align*}
provided that $\lambda_3 \asymp (n/q)^{-2r/(2r+1)}$.
When $q = o\left(n^{\frac{4r\nu}{2r+1+4r\nu}} \right)$, the above expression has an order of $o_p(1)$. In this case, again
\begin{align*}
    & \left\| \left( \TSS \right)^{1/2} (\ftildets - \fhatts)  \right\|_2 = O_p \left(\lambda_3^{\frac{1}{2}} + \left( \frac{n}{q}  \lambda_3^{\frac{1}{2r} }\right)^{-\frac{1}{2}}  + q^{\frac{1}{2}}\lambda_3^{\frac{1}{2} + \nu}  \right).
\end{align*}
Thus, $\left\| \left( \TSS \right)^{1/2} (\ftildets - \fhatts)  \right\|_2 = O_p \left(\lambda_3^{\frac{1}{2}} \right)$. As a result, $\mathcal{R}(\fhatts) = O_p \left(\lambda_3 \right)$ provided that $\lambda_3 \asymp (n/q)^{-2r/(2r+1)}$.
Finally, let $\nu \rightarrow 1/2 - 1/(4r)$, we have $q = o\left(n^{\frac{2r-1}{4r}} \right)$.
\end{proof}

\subsection{Proofs of Lemmas} \label{s:lemmas}

\noindent{\bf Proof of Lemma \ref{lm:tail_sum}}
\begin{proof}
    Note that
    \begin{align*}
         \sum_{k=1}^K a_k \exp(-b_k x)  \ \le \ 
         Ka \exp(-bx)  \ = \ 
         \exp\left[ - \{ b - x^{-1} \log(Ka) \} x \right]. 
    \end{align*}
We established the Lemma by noting that $b - x^{-1} \log(Ka) > (1 - \epsilon) b$. 
\end{proof}

\noindent{\bf Proof of Lemma \ref{lm:(i)}} 
\begin{proof}
First of all, we claim that, for any $\xi \in (0,1)$, 
\begin{align} \label{equ:tail_bound_1}
\begin{split}
	& \pr\left(\mathrm{max}_{j \in \calS^c} \bigg\| \frac{\sigma}{\lambda_1 {n}} \wt \BX_{\bullet j}^{\top} \bigg( \mathrm{I} - \boldsymbol{\Delta}_n \bigg) \boldsymbol{z}_n  \bigg\|_{2} \ge \frac{ \xi C_{\min}}{3}\right) \\
	& \le 2 (p-q) \exp \left\{ -\frac{ \lambda_1^2 C_{\min}^2 \xi^2 n}{48 \sigma^2 \tau  } \right\} + (p-q) \exp \left( -\frac{n}{32} \right).
\end{split}
\end{align}
To show \eqref{equ:tail_bound_1}, first apply the union bound to get
\begin{align*}
	& \pr\Bigg(\mathrm{max}_{j \in \calS^c} \left\| \frac{\sigma}{\lambda_1 n} \wt \BX_{\bullet j}^{\top} \bigg( \mathrm{I} - \boldsymbol{\Delta}_n \bigg) \boldsymbol{z}_n  \right\|_{2} \ge \frac{ \xi C_{\min}}{3} \Bigg) \\
	& =  \pr\Bigg( \bigcup_{j \in \calS^c} \bigg\{ \left\| \frac{\sigma}{\lambda_1 {n}} \wt \BX_{\bullet j}^{\top} \bigg( \mathrm{I} - \boldsymbol{\Delta}_n \bigg) \boldsymbol{z}_n  \right\|_{2} \ge \frac{ \xi C_{\min}}{3} \bigg\} \Bigg) \\
	& \le \sum_{j \in \calS^c} \pr \Bigg( {1\over \sqrt{n} } \left\| \wt \BX_{\bullet j}^{\top} \bigg( \mathrm{I} - \boldsymbol{\Delta}_n \bigg) \boldsymbol{z}_n
\right\|_{2} \ge \frac{\lambda_1 \xi C_{\min} \sqrt{n}}{3\sigma} \Bigg).
\end{align*}
Write
\begin{align} \label{e:spectral_W}
 	\TnSS 
  = {1\over n}  \wt \BX_{\calS}^{\top} \otimes \wt\BX_{\calS} 
 	= \sum_{k=1}^{\infty} \widehat{\rho}_k \widehat{\boldsymbol{\phi}}_{k} \otimes \widehat{\boldsymbol{\phi}}_{k}^{\top},
\end{align}
where the $\widehat{\rho}_k$ are the (nonnegative) eigenvalues of $\TnSS$ arranged in descending order and $\widehat{\boldsymbol{\phi}}_{k}$ are the corresponding eigenfunctions.
Assume without loss of generality that $\{\widehat{\boldsymbol{\phi}}_{k}, k \ge 1\}$ is a CONS
of $\mathbb{L}_2^q[0,1]$. It follows that
\ben\label{eq:Xs_decomp}
 	\wt\BX_{\calS}(u)  = \sum_{k=1}^{\infty} \widehat{\boldsymbol{\zeta}}_k \widehat{\boldsymbol{\phi}}_{k}^{\top}(u),
\een
where $\widehat{\boldsymbol{\zeta}}_k := \int  \wt \BX_{\calS}(t) \widehat{\boldsymbol{\phi}}_{k}(t) dt$ satisfies
\begin{align}\label{eq:zeta_orthonorm}
	{1\over n} \widehat{\boldsymbol{\zeta}}_j^{\top} \widehat{\boldsymbol{\zeta}}_k 
	= \langle \widehat{\boldsymbol{\phi}}_{j}, \TnSS \widehat{\boldsymbol{\phi}}_{k} \rangle_2 
	= \widehat{\rho}_k \delta_{j,k}.
\end{align}
Since there are at most $n$ linearly independent $\widehat{\boldsymbol{\zeta}}_k$, $\widehat{\rho}_k = 0, k > n$ and higher order eigenfunctions
$\widehat{\boldsymbol{\phi}}_{k}$, for $ k>n$ can be obtained by the Gram-Schmidt orthogonalization. 
Thus, we can re-express $\boldsymbol{\Delta}_n$ as
\begin{align*}
	\boldsymbol{\Delta}_n &= {1\over n} \int \sum_{i=1}^n \widehat{\boldsymbol{\zeta}}_i \widehat{\boldsymbol{\phi}}_i^{\top}(u)
\left\{\sum_{k=1}^{\infty} (\widehat{\rho}_k + \lambda_2)^{-1}(\widehat{\boldsymbol{\phi}}_k\otimes
\widehat{\boldsymbol{\phi}}_k^{\top} ) \wt\BX_{\calS}\right\}(u) du \\
	& = {1\over n} \iint \sum_{i=1}^n \widehat{\boldsymbol{\zeta}}_i \widehat{\boldsymbol{\phi}}_i^{\top}(u)
\left\{\sum_{k=1}^{\infty} \widehat{\boldsymbol{\phi}}_k(u) (\widehat{\rho}_k + \lambda_2)^{-1} \widehat{\boldsymbol{\phi}}_k^{\top}(v) \right\} \sum_{j=1}^n \widehat{\boldsymbol{\phi}}_j(v) \widehat{\boldsymbol{\zeta}}_j^{\top}du dv \\
	& = {1\over n} \sum_{k=1}^{n} \frac{1}{\widehat{\rho}_k + \lambda_2} \widehat{\boldsymbol{\zeta}}_k \widehat{\boldsymbol{\zeta}}_k^{\top} 
	= \sum_{k=1}^{n} \frac{ \widehat{\rho}_k }{\widehat{\rho}_k + \lambda_2} \widehat{\boldsymbol{\zeta}}_k^* \widehat{\boldsymbol{\zeta}}_k^{* \top},
\end{align*}
where $\widehat{\boldsymbol{\zeta}}_k^* =(n \widehat{\rho_k})^{-1/2} \widehat{\boldsymbol{\zeta}}_k$ are $n$-dim orthonormal vectors. Clearly, $\mathrm{I} - \boldsymbol{\Delta}_n$ is a positive-definite matrix with all eigenvalues less or equal to $1$. 

Conditional on $\boldsymbol{X}_n$,  $Q_j(t):=n^{-1/2} \wt \BX_{\bullet j}^{\top}(t) ( \mathrm{I} - \boldsymbol{\Delta}_n ) \boldsymbol{z}_n$ is a rank $n$ Gaussian process with
\begin{align*}
	\E[Q_j(t)|\boldsymbol{X}_n]=0 \quad\mbox{and}\quad 
	\cov[Q_j(s), Q_j(t)|\boldsymbol{X}_n] = n^{-1} \wt \BX_{\bullet j}^{\top}(s) ( \mathrm{I} - \boldsymbol{\Delta}_n )^2 \wt \BX_{\bullet j}(t).
\end{align*}
Also, note that
\begin{align}
  	n^{-1} \int \wt \BX_{\bullet j}^{\top}(s) ( \mathrm{I} - \boldsymbol{\Delta}_n )^2 \wt \BX_{\bullet j}(s) ds 
	\le n^{-1} \int \wt \BX_{\bullet j}^{\top}(s)  \wt \BX_{\bullet j}(s) ds = \tr\left(\mathscr{T}_n^{(j,j)}\right).
   \end{align}
Define the event $\mathcal{D}_j(c_0) = \left\{\tr\left(\mathscr{T}_n^{(j,j)}\right) < c_0 \right\}$.
It follows that
\begin{align} \label{e:Q1}
\begin{split}
 & \pr \Bigg( \| Q_j \|_{2} \ge \frac{\lambda_1 \xi C_{\min} \sqrt{n}}{3\sigma} \Bigg)  \\
 & =  \E \Bigg[\pr \Bigg( \| Q_j \|_{2}^2 \ge \frac{\lambda_1^2 \xi^2 C_{\min}^2 n}{9\sigma^2} 
 \Bigg| \boldsymbol{X}_n \Bigg) \Bigg] \\
 & \le  \E \Bigg[ \pr \Bigg( \| Q_j \|_{2}^2 \ge \frac{\lambda_1^2 \xi^2 C_{\min}^2 n}{9\sigma^2} \Bigg| \boldsymbol{X}_n, \mathcal{D}_j(c_0) \Bigg) \mathrm{I}\left( \mathcal{D}_j(c_0) \right)  \Bigg]
+ \pr\left(\mathcal{D}_j^c(c_0)\right) \\
& \le  \E \Bigg[ \pr \Bigg( \| Q_j \|_{2}^2 \ge \frac{\lambda_1^2 \xi^2 C_{\min}^2 n}{9\sigma^2} \Bigg| \boldsymbol{X}_n, \mathcal{D}_j(c_0) \Bigg)   \Bigg]
+ \pr\left(\mathcal{D}_j^c(c_0)\right).
\end{split}
 \end{align}
By Lemma \ref{lm:large_deviation_2nd_ver} (i) with $L=1, K=n, s=4/3$,
\begin{align} \label{e:Q2}
 \pr \Bigg(  \| Q_j \|_{2}^2 \ge \frac{\lambda_1^2 \xi^2 C_{\min}^2 n}{9\sigma^2} \Bigg| \boldsymbol{X}_n, \mathcal{D}_j(c_0) \Bigg) 
\le 2 \exp\bigg\{ - \frac{\lambda_1^2 \xi^2 C_{\min}^2 n}{24 \sigma^2 c_0} \bigg\}.
\end{align}
Recall that $\Tnjj 
= \frac{1}{n} \sum_{i=1}^n \widetilde{X}_{ij}\otimes \widetilde{X}_{ij}$,
and $\widetilde{X}_{ij} \overset{indep}{\sim}  \mathcal{GP} \left(0, \mathcal{T}^{(j,j)} \right)$. Thus, $\tr\left(\mathscr{T}_n^{(j,j)}\right)  
=\frac{1}{n} \sum_{i=1}^n \| \widetilde{X}_{ij} \|_{2}^2
$ and
\begin{align*}
\pr\left(\mathcal{D}_j^c(c_0)\right) 
= \pr\left(\tr\left(\mathscr{T}_n^{(j,j)}\right) > c_0 \right)
= \pr\left(\sum_{i=1}^n \| \widetilde{X}_{ij} \|_{2}^2 > n c_0 \right).
\end{align*}
Thus, by Lemma \ref{lm:large_deviation_2nd_ver} (ii) with $L=n$, $s = 16/9$ and the assumption (C.\ref{ass:a2})
we obtain
\begin{align} \label{e:Q3}
	\pr\left(\mathcal{D}_j^c(c_0)\right) 
	\le  \exp\left\{  - \frac{ c_0 - \tr(\Tjj) }{ 32  } n \right\}
	\le  \exp\left\{  - \frac{ c_0 - \tau }{ 32  } n \right\},
\end{align}
for any $c_0 > (1 + s/2)\tau$.
It follows from \eqref{e:Q1}-\eqref{e:Q3}, with $c_0 = 2 \tau$, $\tau > 1$ and $\lambda_1 < D_{1,1}^*$, we obtain
\begin{align*}
	& \pr\left(\mathrm{max}_{j \in \calS^c} \left\|\frac{\sigma}{\lambda_1 \sqrt{n}} Q_j\right\|_{2} 
	\ge \frac{ \xi C_{\min}}{3} \right) \\
	& \le 2 (p-q) \exp \left\{ -\frac{ \lambda_1^2 C_{\min}^2 \xi^2 n}{48 \sigma^2 \tau  } \right\} + (p-q) \exp \left( -\frac{\lambda_1^2 n}{32 (D_{1,1}^*)^2 } \right).
\end{align*}
This proves \eqref{equ:tail_bound_1}.
Suppose for $d \in (0, 1)$,  we have
\begin{align*}
    \lambda_1 > \max\left( \sqrt{{48 \over d}} \cdot {\sigma \tau^{1/2} \over C_{\min} \xi} , \ \   \sqrt{\frac{32}{d}} \cdot D_{1,1}^* \right) \cdot \sqrt{{\log(p-q) \over n}},
\end{align*}
which is equivalent to
\begin{align*}
     \frac{  C_{\min}^2  \xi^2}{48 \sigma^2 \tau  } \cdot \lambda_1^2 n - \log(p-q) > (1-d) \frac{ C_{\min}^2 \xi^2}{48 \sigma^2 \tau  } \cdot \lambda_1^2 n,
\end{align*}
and
\begin{align*}
    \frac{ \lambda_1^2 n }{32 (D_{1,1}^*)^2 } - \log(p-q) > (1-d) \frac{ \lambda_1^2 n }{32 (D_{1,1}^*)^2 }.
\end{align*}
By Lemma \ref{lm:tail_sum} with $\xi = \gamma/3$ and $d = 1/2$, we have
\begin{align*}
	& \pr\left(\mathrm{max}_{j \in \calS^c} \bigg\| \frac{\sigma}{\lambda_1 {n}} \wt \BX_{\bullet j}^{\top} \bigg( \mathrm{I} - \boldsymbol{\Delta}_n \bigg) \boldsymbol{z}_n  \bigg\|_{2} \ge \frac{ \gamma C_{\min}}{9}\right) \le \exp\left( - D\lambda_1^2 n  \right)
\end{align*}
holds for any $D$ and $\lambda_1$ such that
\begin{align*}
    \lambda_1 >  D_1^* \cdot {  (\sigma+1) \tau^{1/2} \over C_{\min} \gamma} \cdot \sqrt{{\log(p-q) \over n}},
\end{align*}
and
\begin{align*}
    D < D_2^* \frac{ C_{\min}^2 \gamma^2}{ (\sigma + 1)^2 \tau  } < \min \left\{ \frac{ C_{\min}^2 \gamma^2}{ 864 \sigma^2 \tau  }, \frac{1}{32 ({D_{1,1}^*})^2 } \right\},
\end{align*}
where $D_1^*$ and $D_2^*$ are universal constants.
\end{proof}

\noindent{\bf Proof of Lemma \ref{lm:proof(ii)}}
\begin{proof}
Let constants $\xi,\delta$, and $\mu = \lambda_1 / \lambda_2$ satisfy
\begin{align} \label{e:xi_delta}
\xi\in (0,\gamma/2) \quad\mbox{and}\quad
\delta\in (0,(\gamma-2\xi)/(1-\gamma)) \quad\mbox{and}\quad \mu C_{\max} > (1-2\xi) / \xi.
\end{align}
We claim that
\begin{align} \label{equ:tail_bound_02}
\begin{split}
	& \pr\left(   \max_{j \in \calS^c}  \left\| \TnjS \left(\JnSS\right)^{-1} \left(\frac{\lambda_2}{\lambda_1} \fzeros+\ws\right) \right\|_{2}
\ge \left(1-{2 \xi \over 3}\right) C_{\min} \right)\\
	& \le  \exp\left\{ -\frac{\lambda_2^2 \varkappa^2 \delta^2 n}{4 C^2 \rho_1^2 q} \right\}
+ 2 (p-q) \exp \left\{ -\frac{ \lambda_2 (C_{\min}/C_{\max})^2 \xi^2 n}{24 (1 + \mu^{-1} C_{\max}^{-1})^2 \tau q}\right\},
\end{split}
\end{align}
for some constant $C >0$ and $q$ that satisfy
\begin{align} \label{e:C_q}
\frac{\lambda_2^2 \varkappa^2 \delta^2}{ 4 C^2 \rho_1^2} < q \le  \sqrt{\frac{\lambda_2^2  \varkappa^2 \delta^2 n }{ 4 C^2 \tau \rho_1}}.
\end{align}
To show \eqref{equ:tail_bound_02}, by Lemma \ref{lemma:cond_dist}, for any $j\in \calS^c$, 
\ben\label{eq:Xj_Xs}
	\wt \BX_{\bullet j}^{\top} \eqid \TjS (\TSS)^{-} \wt \BX_\calS^{\top} + \boldsymbol{E}_j^{\top},
\een
where $\boldsymbol{E}_j=(E_{1j},\ldots, E_{nj})^\top$ is a vector of iid zero-mean Gaussian processes independent of $\wt \BX_\calS$ with a covariance
operator 
\begin{align} \label{e:co_E_j}
	\mathscrbf{T}^{(j|\mathcal{S})} := \T^{(j,j)} - \TjS (\TSS)^{-} \T^{(\calS, j)} .
\end{align}
With (\ref{eq:Xj_Xs}) and Condition \ref{ass:a1},
%
\ben \label{e:main_1}
	&& \bigg\|  \TnjS \left(\JnSS\right)^{-1} \left(\frac{\lambda_2}{\lambda_1} \fzeros+\ws\right) \bigg\|_{2}  \\
	&& \hskip-5mm =  \bigg\| \int {1\over n} \wt \BX_{\bullet j}^{\top}(\cdot) \wt\BX_\calS (s) 
 \left(\JnSS\right)^{-1}  \left(\frac{\lambda_2}{\lambda_1} \fzeros+\ws\right)(s) ds \bigg\|_{2}  \nonumber\\
	&& \hskip-5mm =  \bigg\| \int   {1\over n} \{ \TjS (\TSS)^{-} \wt \BX_\calS^{\top} +  \boldsymbol{E}_j^{\top} \} (\cdot) \wt \BX_\calS(s)   \left(\JnSS\right)^{-1}  \left(\frac{\lambda_2}{\lambda_1} \fzeros+\ws\right)(s) ds \bigg\|_{2} \nonumber\\
	&& \hskip-5mm \le \vertiii{\TjS (\TSS)^{-}}_{\infty, 2} \vertiii{\TnSS (\JnSS)^{-1}}_{\infty, \infty} \left(\frac{\lambda_2}{\lambda_1} \|\fzeros\|_{\infty}  + C_{\max} \right) + \big\| \boldsymbol{E}_j^{\top} (\cdot) \boldsymbol{Z} \big\|_{2}, \nonumber 
\een
where
\begin{align} \label{e:Z}
\boldsymbol{Z}_{n\times 1}:={1\over n} \int \wt \BX_\calS (s)  \left(\JnSS\right)^{-1}  \left(\frac{\lambda_2}{\lambda_1} \fzeros+\ws\right)(s) ds.
\end{align}
Note that if 
\begin{align*} 
\vertiii{\TnSS (\JnSS)^{-1}}_{\infty, \infty} < \varkappa(1+\delta)
\quad\mbox{and}\quad
\max_{j \in \calS^c}\big\| \boldsymbol{E}_j^{\top} (\cdot) \boldsymbol{Z} \big\|_{2} < \frac{ \xi C_{\min}}{3},
\end{align*}
then \eqref{e:xi_delta}, \eqref{e:main_1} along with Condition C.\ref{ass:a3} and $\|\fzeros\|_{\infty} = 1$ give
\begin{align*}
& \max_{j \in \calS^c} \bigg\|\TnjS \left(\JnSS\right)^{-1} \left(\frac{\lambda_2}{\lambda_1} \fzeros+\ws\right) \bigg\|_{2} \\
& <  (1-\gamma)(1+\delta) \left(1 + C_{\max}^{-1}\frac{\lambda_2}{\lambda_1} \right) C_{\min} + {\xi C_{\min}\over 3}
< \left( 1-{2 \xi \over 3} \right) C_{\min},
\end{align*}
where the last inequality follows from the fact
$$
(1-\gamma)(1+\delta) \left(1 + C_{\max}^{-1}\frac{\lambda_2}{\lambda_1} \right) < 1-\xi
$$
by \eqref{e:xi_delta}.
In the following, we establish
\begin{align} \label{e:(i)0}
\pr\left(\vertiii{\TnSS (\JnSS)^{-1}}_{\infty, \infty} \ge \varkappa ( 1 + \delta ) \right)
\le \exp\left\{ -\frac{\lambda_2^2 \varkappa^2 \delta^2 n}{4 C^2 \rho_1^2 q} \right\}
\end{align}
and
\begin{align} \label{e:(ii)0}
\pr\left(\max_{j \in \calS^c}\big\| \boldsymbol{E}_j^{\top} (\cdot) \boldsymbol{Z} \big\|_{2} \ge \frac{ \xi C_{\min}}{3}\right) \le 2 (p-q) \exp \left\{ -\frac{ \lambda_2 (C_{\min}/C_{\max})^2  \xi^2 n }{24  (1 + \mu^{-1} C_{\max}^{-1})^2 \tau q }\right\}.
\end{align}
To show \eqref{e:(i)0}, first apply the triangle inequality to obtain
\begin{align} \label{e:(i)1}
\begin{split}
& \vertiii{\TnSS (\JnSS)^{-1}}_{\infty, \infty} \\
& \le \vertiii{\TSS (\JSS)^{-1}}_{\infty, \infty} + \vertiii{\TnSS \left\{ (\JnSS)^{-1} - (\JSS)^{-1} \right\} }_{\infty, \infty} \\
& \hspace{.3cm} + \vertiii{ (\TnSS - \TSS) (\JSS)^{-1} }_{\infty, \infty}.
\end{split}
\end{align}
Then, by Lemma \ref{lemma:norm_inequality}, 
\begin{align} \label{e:(i)2}
\begin{split}
&  \vertiii{\TnSS \left\{ (\JnSS)^{-1} - (\JSS)^{-1} \right\} }_{\infty, \infty} \\
& \le  \sqrt{q} \vertiii{\TnSS (\JnSS)^{-1} ( \TSS - \TnSS ) (\JSS)^{-1} }_{2,2} \\
& \le \sqrt{q} \vertiii{\TnSS (\JnSS)^{-1}}_{2,2} \vertiii{ \TSS - \TnSS}_{2,2} \vertiii{ (\JSS)^{-1} }_{2,2} 
\end{split}
\end{align}
and
\begin{align} \label{e:(i)3}
\begin{split}
& \vertiii{ (\TnSS - \TSS) (\JSS)^{-1} }_{\infty, \infty} \\
& \le \sqrt{q} \vertiii{ (\TnSS - \TSS) (\JSS)^{-1} }_{2,2} \\
& \le \sqrt{q} \vertiii{ \TnSS - \TSS }_{2,2} \vertiii{(\JSS)^{-1} }_{2,2}.
\end{split}
\end{align}
Since
\begin{align}\label{e:(i)4}
\vertiii{ (\JSS)^{-1} }_{2,2} \le \frac{1}{\lambda_2} \quad \mbox{and} \quad
\vertiii{\TnSS (\JnSS)^{-1}}_{2,2} \le 1, \end{align}
\eqref{e:(i)1}-\eqref{e:(i)4} together with the Condition C.\ref{ass:a3} give
\begin{align*}
\vertiii{\TnSS (\JnSS)^{-1}}_{\infty, \infty} \le \varkappa + \frac{2\sqrt{q}}{\lambda_2} \vertiii{ \TSS - \TnSS}_{2,2}.
\end{align*}
Thus, for $\delta > 0$, by Lemma \ref{lemma:tail_bound} we have
\begin{align*} 
\begin{split}
\pr \left( \vertiii{\TnSS (\JnSS)^{-1}}_{\infty, \infty} \ge \varkappa (1 + \delta) \right) 
& \le \pr \left( \vertiii{\TSS - \TnSS}_{2,2} \ge \frac{\lambda_2 \varkappa \delta}{2\sqrt{q}} \right) \\
& \le \exp\left\{ -\frac{\lambda_2^2 \varkappa^2 \delta^2 n}{4 C^2 \rho_1^2 q} \right\}
\end{split}
\end{align*}
for some constant $C >0$ and $q$ satisfies \eqref{e:C_q}.
This proves \eqref{e:(i)0}. 

To prove \eqref{e:(ii)0}, recall the definitions of $\boldsymbol{Z}$ in \eqref{e:Z} and let
$U_j(\cdot) =  \boldsymbol{E}_j^{\top} (\cdot) \boldsymbol{Z}$.
We have
\begin{align*}
	\pr\left(  \mathrm{max}_{j \in \calS^c}\|U_j(\cdot)\|_{2}\ge \frac{ \xi C_{\min}}{3}\right) 
	& \le \sum_{j \in \calS^c} \pr\left(\|U_j(\cdot)\|_{2}\ge \frac{ \xi C_{\min}}{3}\right)  \\
	& =  \sum_{j \in \calS^c} \E \left\{ \pr\left(\|U_j(\cdot)\|_{2}\ge \frac{ \xi C_{\min}}{3} | \wt\BX_\calS \right) \right\}. 
\end{align*}
Also, conditional on $\wt\BX_\calS$, $U_j$ is a zero-mean Gaussian process with covariance operator $\mathscr{H}_j$ with trace 
\begin{align} \label{e:tr_Hj}
	\tr(\mathscr{H}_j) = \|\boldsymbol{Z}\|^2 \tr(\mathscrbf{T}^{(j|\mathcal{S})}),
\end{align}
where $\mathscrbf{T}^{(j|\mathcal{S})}$ is defined in \eqref{e:co_E_j}. 
It remains to bound $\|\boldsymbol{Z}\|^2$ and $\tr(\mathscrbf{T}^{(j|\mathcal{S})})$.
First, by \eqref{e:co_E_j} and (C.\ref{ass:a2}),
\begin{align} \label{e:tr_jS}
	\tr(\mathscrbf{T}^{(j|\mathcal{S})}) \le  \tr(\mathscrbf{T}^{(j,j)}) \le \tau.
\end{align}
By the decompositions in \eqref{e:spectral_W} and \eqref{eq:Xs_decomp},
\begin{align*}
 	\left(\JnSS\right)^{-1}  \left(\frac{\lambda_2}{\lambda_1} \fzeros + \ws \right)
	= \sum_{k \ge 1} \frac{1}{\widehat{\rho}_k + \lambda_2} 
	\left\langle\widehat{\boldsymbol{\phi}}_k, \frac{\lambda_2}{\lambda_1} \fzeros + \ws \right\rangle_{2}  \widehat{\boldsymbol{\phi}}_k, 
\end{align*}
and therefore 
\begin{align} \label{e:Z_norm}
\begin{split}
	\|\boldsymbol{Z}\|^2 &={1\over n^2} \left\| \sum_{k \ge 1} \frac{1}{\widehat{\rho}_k + \lambda_2}  \left\langle\widehat{\boldsymbol{\phi}}_k, \frac{\lambda_2}{\lambda_1} \fzeros + \ws \right \rangle_{2} \wh{\boldsymbol{\zeta}}_k \right\|^2 \\
	&= {1\over n} \sum_{k \ge 1} \frac{\widehat{\rho}_k}{(\widehat{\rho}_k + \lambda_2)^2}  \left\langle\widehat{\boldsymbol{\phi}}_k, \frac{\lambda_2}{\lambda_1} \fzeros + \ws \right\rangle_{2}^2 \quad  \hbox{(by \eqref{eq:zeta_orthonorm})}\\
	& \le {1\over  n\lambda_2}\left\|\frac{\lambda_2}{\lambda_1} \fzeros + \ws\right\|_{2}^2  \\
	&\le {q\over  n \lambda_2} \left(\lambda_2/\lambda_1  + C_{\max} \right)^2.
\end{split}
\end{align}
By \eqref{e:tr_Hj}, \eqref{e:tr_jS}, \eqref{e:Z_norm}, and an application of Lemma \ref{lm:large_deviations} (i) with $s = 4/3$,
\begin{align*} 
	\sum_{j \in \calS^c} \E \Bigg\{ \pr \Bigg( \| U_j \|_{2} \ge \frac{\xi C_{\min}}{3} \Bigg| \wt\BX_\calS \Bigg) \Bigg\}
	\le 2 (p-q) \exp \left\{ -\frac{ \lambda_2 (C_{\min}/C_{\max})^2 \xi^2 n}{24 (1 + \mu^{-1} C_{\max}^{-1})^2 \tau q }\right\}.
\end{align*}
This concludes the proof of \eqref{e:(ii)0}. According to Lemma \ref{lemma:varkappa_lower}, we have $\varkappa \ge \rho_1 (\rho_1 + \lambda_2)^{-1}$. Let $\xi = \delta = \gamma/3$, we find \eqref{equ:tail_bound_02} is bounded by
\begin{align}
    \exp\left\{ -\frac{\lambda_2^2 \gamma^2 n}{36 C^2 (\rho_1 + \lambda_2)^2 q} \right\}
+ 2 (p-q) \exp \left\{ -\frac{ \lambda_2 (C_{\min}/C_{\max})^2 \gamma^2 n}{ 864 \tau q}\right\}.
\label{equ:tail_bound_21}
\end{align}

Note that $\rho_1$ must be bounded from below by a universal constant, denoted as $D_0^*$. Without this lower bound, the model will only contain noise and no meaningful signals.
Below, we will use $D^*$ to denote a universal constant in $(0,\infty)$ whose value changes from line to line.
Suppose $\lambda_2$ satisfies
\begin{align}
    \lambda_2 > \frac{6 \max(1,D_{2,1}^*)}{(D_0^*)^{1/2}} \frac{ C \tau^{1/2} \left( \rho_1 + 1 \right)}{ \gamma} \cdot \sqrt{\frac{q^2}{n}} > \frac{6 C \tau^{1/2} \left( \rho_1 + \lambda_2 \right)}{ \sqrt{\rho_1} \gamma} \cdot \sqrt{\frac{q^2}{n}},
\end{align}
which meets Condition \eqref{e:C_q}.
It can be shown that the first term of \eqref{equ:tail_bound_21} is bounded by $\exp\left( -D_a^{(2)} \frac{\lambda_2^2 n}{q} \right)$ for any $D_a^{(2)} \le D^* \gamma^2 (\rho_1 + 1)^{-2}$. 
Suppose for $d \in (0, 1)$, $\lambda_2$ also satisfies
\begin{align}
    \lambda_2 > \frac{864 \tau}{d (C_{\min}/C_{\max})^2 \gamma^2} \cdot \frac{q \log(p-q)}{n},
\end{align}
which is equivalent to
\begin{align*}
     \frac{ (C_{\min}/C_{\max})^2 \gamma^2}{864 \tau  } \cdot \frac{\lambda_2 n}{q} -\log(p-q) > (1-d) \cdot \frac{ (C_{\min}/C_{\max})^2 \gamma^2}{864 \tau  } \cdot \frac{\lambda_2 n}{q}.
\end{align*}
Then, the second term of \eqref{equ:tail_bound_21} is bounded by 
\begin{align*}
    &(p-q) \exp \left\{ -\frac{ \lambda_2 (C_{\min}/C_{\max})^2 \gamma^2 n}{ 864 \tau q}\right\} \\
    \le & \exp \left\{ -  \frac{(1-d) (C_{\min}/C_{\max})^2 \gamma^2}{864 \tau  } \cdot \frac{\lambda_2 n}{q}  \right\} \\
    \le & \exp \left\{ -  \frac{(1-d) (C_{\min}/C_{\max})^2 \gamma^2}{864 D_{2,1}^* \tau  } \cdot \frac{\lambda_2^2 n}{q}  \right\} \\
    \le & \exp \left( - D_b^{(2)} \frac{\lambda_2^2 n }{q} \right),
\end{align*}
where $D_b^{(2)} \le  D^* (1-d) (C_{\min}/C_{\max})^2 \gamma^2\tau^{-1}$. 
The second inequality uses the fact $\lambda_2 < D_{2,1}^*$.
It follows from Lemma \ref{lm:tail_sum} with $d = 1/2$
\begin{align*}
	& \pr\left(   \max_{j \in \calS^c}  \left\| \TnjS \left(\JnSS\right)^{-1} \left(\frac{\lambda_2}{\lambda_1} \fzeros+\ws\right) \right\|_{2} 
    \ge \left(1-{ 2\gamma \over 9}\right) C_{\min} \right) 
    \le \exp\left( - D^{(2)} \frac{\lambda_2^2 n}{q} \right)
\end{align*}
holds for any $D^{(2)}$ and $\lambda_2$ such that
\begin{align*}
    D^{(2)} \le D^* \frac{(C_{\min}/C_{\max})^2 \gamma^2}{(\rho_1+1)^2 \tau } \le \min\left\{D_a^{(2)}, D_b^{(2)} \right\},
\end{align*}
and
\begin{align*}
    \lambda_2 > D^* \frac{\tau (\rho_1 + 1)}{(C_{\min}/C_{\max})^2 \gamma^2}  \max\left({q\log(p-q)\over n},\sqrt{q^2\over n} \right).
\end{align*}

\end{proof}

\noindent{\bf Proof of Lemma \ref{lm:(iii)}}

\begin{proof}
Define $\ESS$ to be the operator that only contains the off-diagonal elements of $\TSS$, i.e.
$\ESS = \TSS - \QSS = \JSS - \QSS_{\lambda_2}$.
Then
\begin{align}
\begin{split}
  & \vertiii{ (\QSS)^{1/2}  (\JSS)^{-1}}_{\infty,\infty} \\
    & = \vertiii{ (\QSS)^{1/2}  (\QSS_{\lambda_2})^{-1} + (\QSS)^{1/2} \bigg\{ (\JSS)^{-1} - (\QSS_{\lambda_2})^{-1} \bigg\} }_{\infty,\infty} \\
    & = \vertiii{ (\QSS)^{1/2}  (\QSS_{\lambda_2})^{-1} - (\QSS)^{1/2} (\QSS_{\lambda_2})^{-1} \ESS (\JSS)^{-1} }_{\infty,\infty} \\
& \le \vertiii{ (\QSS)^{1/2}  (\QSS_{\lambda_2})^{-1} }_{\infty,\infty} \bigg( 1 +  \vertiii{\ESS (\JSS)^{-1} }_{\infty,\infty}\bigg).
\end{split}
\label{equ:bound_a}
\end{align}
Note that
\begin{align}
\begin{split}
& \vertiii{ (\QSS)^{1/2}  (\QSS_{\lambda_2})^{-1} }_{\infty,\infty} \\
& = \max_{j \in \calS} \vertiii{ (\Tjj)^{1/2} \left( \Tjj + \lambda_2 \mathscr{I} \right)^{-1} }_{2,2} \\
& = \max_{j \in \calS} \sup_{\| f_j\|_2 \le 1} \Bigg[ \sum_{k \ge 1} \frac{\nu_{jk}  }{(\nu_{jk} + \lambda_2)^2 } \langle f_j , \eta_{jk} \rangle_2^2 \Bigg]^{1/2} \le \frac{1}{2 \sqrt{\lambda_2}}.
\end{split}
\label{equ:bound_b}
\end{align}
The last inequality holds by observing that the maximum value of function $h(x) = x (x + \rho)^{-2}$ is $h(\rho) = (4\rho)^{-1}$. 
Meanwhile
\begin{align}
\begin{split}
& \vertiii{\ESS (\JSS)^{-1} }_{\infty,\infty} \\
 =&  \vertiii{\JSS (\JSS)^{-1} - \QSS_{\lambda_2} (\JSS)^{-1} }_{\infty,\infty} \\
    \le &  1 + \vertiii{\QSS_{\lambda_2} ( \QSS_{\lambda_2} + \ESS  )^{-1} }_{\infty,\infty} \\
     = & 1 + \vertiii{ \left\{ \mathscrbf{I} + \ESS (\QSS_{\lambda_2})^{-1} \right\}^{-1}  }_{\infty,\infty}.
\end{split}
\label{equ:bound_c}
\end{align}
By Theorem 3.5.5, \citesupp{hsing2015theoreticalsupp}, $\mathscrbf{I} + \ESS (\QSS_{\lambda_2})^{-1}$ is invertible if 
\begin{align*}
\aleph(\lambda_2) = \vertiii{\ESS (\QSS_{\lambda_2})^{-1}}_{\infty,\infty} < 1,
\end{align*}
which is warranted by Condition C.\ref{ass:a6}. In this case, 
\begin{align}
    \vertiii{ \left\{ \mathscrbf{I} + \ESS (\QSS_{\lambda_2})^{-1} \right\}^{-1}  }_{\infty,\infty} < \frac{1}{1-\aleph(\lambda_2)}.
\label{equ:bound_d}
\end{align}
Therefore, \eqref{equ:bound_4} holds by \eqref{equ:bound_a}-\eqref{equ:bound_d}.
\end{proof}

\noindent{\bf Proof of Lemma \ref{lm:(iv)}}

\begin{proof}
Recall that ${g}_{j} = n^{-1} \wt \BX_{\bullet j}^{\top} \boldsymbol{\epsilon}_n$.
Conditional on $\wt \BX_{\bullet j}$, ${g}_j$ is a rank $n$ Gaussian process with mean zero and covariance operator $\mathscr{R}_j  = n^{-1} \sigma^2 \Tnjj$, and $\tr(\mathscr{R}_j)  = n^{-1} \tr(\Tnjj) \sigma^2 $.
Define the event $\mathcal{D}_j(c_1) = \left\{\tr\left(\mathscr{T}_n^{(j,j)}\right) < c_1 \right\}$, it follows that
\begin{align*} 
\begin{split}
	\pr\left( \left\| \gns \right\|_{\infty} > \lambda_2 \right) 
  \le & \sum_{j \in \calS} \pr \Bigg( \| {g}_j \|_{2} \ge \lambda_2 \Bigg)  \\
    \le&   \sum_{j \in \calS} \E \Bigg[ \pr \Bigg( \| {g}_j \|_{2} \ge \lambda_2 \Bigg| \wt \BX_{\bullet j}, \mathcal{D}_j(c_1) \Bigg)  \Bigg]
    + \sum_{j \in \calS} \pr\left\{\mathcal{D}_j^c(c_1)\right\}.
    %
\end{split}
\end{align*}
Setting $c_1 = 2 \tau$ and applying Lemma \ref{lm:large_deviation_2nd_ver} (i) with $s = 4/3$, we get 
\bse
    \pr \Bigg( \| {g}_j \|_{2} \ge \lambda_2 \Bigg| \wt \BX_{\bullet j}, \mathcal{D}_j(c_1) \Bigg)  \le   2   \exp \left( -\frac{ 3 n \lambda_2^2 }{16 \sigma^{2}  \tau } \right). 
\ese
Given that $\widetilde{X}_{ij} \overset{indep}{\sim}  \mathcal{GP} \left(0, \mathcal{T}^{(j,j)} \right)$ and $\vertiii{\scrT^{(j,j)}}_{2,2}=1$ together with the facts $\tau > 1$ and $\lambda_2 < D_{2,1}^*$,
\begin{align*}
        \pr\left(\mathcal{D}_j^c(c_1)\right) &= \pr\left(\sum_{i=1}^n  \|\wt X_{ij}\|^2 > 2n \sigma_0^2\right)  \le   \exp \left( -\frac{\tau }{32 }  n  \right) \\
        & \le  \exp \left( -\frac{n}{32 }  \right)
        \le  \exp \left( -\frac{\lambda_2^2 n}{32 (D_{2,1}^*)^2 }  \right)
\end{align*}
by Lemma \ref{lm:large_deviation_2nd_ver} (ii) with $s = 16/9$.
Combining the two bounds entails
\begin{align*}
    \pr\left( \left\| \gns \right\|_{\infty} > \lambda_2 \right)  \le  2 q  \exp \left( -\frac{ 3 n \lambda_2^2 }{16 \sigma^{2}  \tau } \right)  +  q \exp \left( -\frac{\lambda_2^2 n}{32 (D_{2,1}^*)^2 }  \right).
\end{align*}
Suppose
\begin{align*}
    \lambda_2 >  D_1^* (\sigma + 1) \tau^{1/2} \cdot \sqrt{{\log q \over n}} > \max \left( {4 \over 3} \sqrt{6} \cdot \sigma \tau^{1/2}, \ 8 D_{2,1}^* \right) \cdot \sqrt{{\log q \over n}},
\end{align*}
which is equivalent to
\begin{align*}
     \frac{ 3 n \lambda_2^2 }{16 \sigma^{2}  \tau } - \log q > \frac{ 3 n \lambda_2^2 }{32 \sigma^{2}  \tau }
\end{align*}
and
\begin{align*}
     \frac{ n \lambda_2^2 }{32 (D_{2,1}^*)^2 } - \log q > \frac{ n \lambda_2^2 }{64 (D_{2,1}^*)^2 }.
\end{align*}
We have \eqref{equ:tail_bound_3} holds for some $D^{(3)} < D_2^* \left(  (\sigma + 1)^{2} \tau \right)^{-1} <  64^{-1} \min\left\{ 6 \sigma^{-2} \tau^{-1}, (D_{2,1}^{*})^{-2} \right\}$, where $D_1^*$ and $D_2^*$ are universal constants. 
\end{proof}

\noindent{\bf Proof of Lemma \ref{lemma:tail_bound}}

\begin{proof}
Recall $\vertiii{\TSS}_{2,2} = \rho_1$ and define $
t = \frac{u^2 n}{C^2 \rho_1^2 q}$ for some constant $C>0$, then by Corollary 2 in \citesupp{koltchinskii2017concentration},
\begin{align}\label{equ:lemma:tail_bound1}
\begin{split}
     &\pr\left(\sqrt{q} \vertiii{\TSS - \TnSS}_{2,2} \ge u \right)\\ 
     = & \pr\left( \vertiii{\TSS - \TnSS}_{2,2} \ge C \vertiii{\TSS}_{2,2} \sqrt{t\over n}\right) \\
     \le&  e^{-t}
\end{split}
\end{align}
as long as
\begin{align}
\sqrt{\frac{t}{n}} = \max \left( \sqrt{\frac{r(\TSS)}{n}}, \frac{r(\TSS)}{n}, \sqrt{\frac{t}{n}}, \frac{t}{n} \right),
\label{equ:lemma:tail_bound2}
\end{align}
where 
\begin{align*}
r(\TSS) = \frac{ (\mathbb{E} \|X_1 \|_2)^2 }{ \vertiii{\TSS}_{2,2} } \le
\frac{\mathbb{E} \|X_1\|_2^2 }{ \vertiii{\TSS}_{2,2} }
\le {q\tau \over \rho_1}
\end{align*}
by Jensen's inequality and Condition C.\ref{ass:a2}.
Hence \eqref{equ:lemma:tail_bound2} holds when
\begin{align*}
{q \tau \over \rho_1} < t < n,
\end{align*}
which amounts to \eqref{e:q_cond}.

\end{proof}

\subsection{Additional technical lemmas} \label{s:addi_lemmas}

\begin{lemma} \label{lemma:minimax_upper_1}
For any $0 < \nu < 1$,
\begin{align*}
    \left\|\left( \TSS \right)^{\nu}\left( \ftildets -\fzeros  \right)\right\|_{2} \leq(1-\nu)^{1-\nu} \nu^{\nu} \lambda_3 ^{\nu}\left\| \fzeros \right\|_{2}.
\end{align*}
\end{lemma}
\begin{proof}
    Write $\TSS = \sum_{k \ge 1} \rho_k \boldsymbol{\phi}_k \otimes \boldsymbol{\phi}_k$ and $\fzeros = \sum_{k \ge 1} f_k \boldsymbol{\phi}_k$. Then
    \begin{align*}
        \ftildets = \sum_{k \ge 1} \frac{\rho_k f_k}{\lambda_3 + \rho_k} \boldsymbol{\phi}_k.
    \end{align*}
    Therefore
    \begin{align*}
        \left\|\left( \TSS \right)^{\nu}\left( \ftildets -\fzeros  \right)\right\|_{2}^2 &= \sum_{k \ge 1} \rho_{k}^{2\nu} \left( \frac{\lambda_3 f_k}{\lambda_3 + \rho_k} \right)^2 \le \max_{k \ge 1} \frac{\lambda_3^2 \rho_k^{2\nu}}{(\lambda_3 + \rho_k)^2} \sum_{k \ge 1} f_k^2 \\
        & \le (1-\nu)^{2(1-\nu)} \nu^{2\nu} \lambda_3^{2\nu} \left\| \fzeros \right\|_{2}^2.
    \end{align*}
    The last inequality follows from Young's inequality: $\lambda_3 +\rho_{k} \geq(1-\nu)^{-(1-\nu)} \nu^{-\nu} \lambda_3^{1-\nu} \rho_{k}^{\nu}$. 
\end{proof}

\begin{lemma} \label{lemma:minimax_upper_2}
For $0 < \nu < 1$,
\begin{align*}
    \vertiii{ \left( \TSS \right)^{\nu} \left(\TSS_{\lambda} \right)^{-1}}_{2,2} \le (1-\nu)^{1-\nu} \nu^{\nu} \lambda^{\nu-1}.
\end{align*}
\end{lemma}
\begin{proof}
For any $\boldsymbol{f} \in \mathbb{L}_2^q$ such that $\| \boldsymbol{f} \|_2 \le 1$, write $\boldsymbol{f} = \sum_{k \ge 1} f_k \boldsymbol{\phi}_k$, we have
\begin{align*}
    \left\| \left( \TSS \right)^{\nu} \left(\TSS_{\lambda} \right)^{-1} \boldsymbol{f} \right\|_2 = \sqrt{\sum_{k \ge 1} \frac{\rho_k^{2\nu}}{(\rho_k + \lambda)^2} f_k^2} \le \max_{k \ge 1} \left\{ \frac{\rho_k^{\nu}}{\rho_k + \lambda} \right\} \le (1-\nu)^{1-\nu} \nu^{\nu} \lambda^{\nu-1}.
\end{align*}
\end{proof}

\begin{lemma} \label{lemma:minimax_upper_3}
Assume Condition C.\ref{ass:a8}-C.\ref{ass:a10} hold. For $0 < \nu \le 1/2$, $r > 1/2$
\begin{align*}
     \left\| \left( \TSS \right)^{\nu} \left(\TSS_{\lambda_3} \right)^{-1} \gns \right\|_2 = O_{p}\left(\left( \frac{n}{q} \cdot \lambda_3^{1-2 \nu + \frac{1}{2r}}\right)^{-\frac{1}{2}}\right).
\end{align*}
\end{lemma}
\begin{proof}
For $0 \le \nu \le 1/2$,
\begin{align*}
    \left\| \left( \TSS \right)^{\nu} \left(\TSS_{\lambda_3} \right)^{-1} \gns \right\|_2^2 &= \sum_{k \ge 1} \left\langle \left( \TSS \right)^{\nu} \left(\TSS_{\lambda_3} \right)^{-1} \gns, \  \boldsymbol{\phi}_k\right\rangle_2^2 \\
    &= \sum_{k \ge 1} \left\langle \left( \TSS \right)^{\nu} \left(\TSS_{\lambda_3} \right)^{-1} \boldsymbol{\phi}_k,  \ \gns, \right\rangle_2^2 \\
    &= \sum_{k \ge 1} \left\langle \frac{\rho_k^{\nu}}{\rho_k + \lambda_3} \boldsymbol{\phi}_k,  \ \frac{1}{n}\sum_{i=1}^n \epsilon_i \wt \BX_{i\calS}  \right\rangle_2^2 \\
    &= \sum_{k \ge 1} \frac{\rho_k^{2\nu}}{(\rho_k + \lambda_3)^2} \left\{ \frac{1}{n}\sum_{i=1}^n \epsilon_i \left\langle  \boldsymbol{\phi}_k,  \  \wt \BX_{i\calS}  \right\rangle_2 \right\}^2. 
\end{align*}
Therefore
\begin{align*}
    \mathbb{E}\left\| \left( \TSS \right)^{\nu} \left(\TSS_{\lambda_3} \right)^{-1} \gns \right\|_2^2 
    & = \frac{\sigma^2}{n} \sum_{k \ge 1} \frac{\rho_k^{2\nu}}{(\rho_k + \lambda_3)^2} \cdot \mathbb{E} \left\langle  \boldsymbol{\phi}_k,  \  \wt \BX_{1\calS}  \right\rangle_2^2  \\
    & = \frac{\sigma^2}{n} \sum_{k \ge 1} \frac{\rho_k^{2\nu + 1}}{(\rho_k + \lambda_3)^2} \\
    & \le \frac{\sigma^2}{n \lambda_3^{1-2\nu}} \sum_{k \ge 1} \frac{\rho_k^{2\nu + 1}}{(\rho_k + \lambda_3)^{1+2\nu}} \\
    & \le C \sigma^2 \left( (n/q) \cdot \lambda_3^{1-2\nu + \frac{1}{2r}}\right)^{-1}
\end{align*}
for some constant $C > 0$. The last inequality is obtained by Lemma \ref{lemma:minimax_upper_5}.
The proof can be completed by Markov inequality. 
\end{proof}
\begin{lemma} \label{lemma:minimax_upper_4}
Assume Condition C.\ref{ass:a8}-C.\ref{ass:a10} hold. Then for any $r>1/2$, $0 < \nu < 1/2 - 1/(4r)$,
\begin{align*}
    &(1). \ \ \vertiii{ \left( \TSS \right)^{\nu} \left(\TSS_{\lambda_3} \right)^{-1} \left( \TnSS - \TSS  \right) \left( \TSS \right)^{-\nu}  }_{2,2} = O_{p}\left( q^{\frac{1}{2}}\left( \frac{n}{q}  \lambda_3^{1-2 \nu + \frac{1}{2r} }\right)^{-\frac{1}{2}}\right). \\
    &(2). \ \ \vertiii{ \left( \TSS \right)^{1/2} \left(\TSS_{\lambda_3} \right)^{-1} \left( \TnSS - \TSS  \right) \left( \TSS \right)^{-\nu}  }_{2,2} = O_{p}\left( q^{\frac{1}{2}} \left( \frac{n}{q} \lambda_3^{\frac{1}{2r} }\right)^{-\frac{1}{2}}\right).
\end{align*}
\end{lemma}
\begin{proof}
    \noindent\textbf{(1).} Write $\boldsymbol{g} = \sum_{k \ge 1} g_k \boldsymbol{\phi}_k$, $\boldsymbol{h} = \sum_{k \ge 1} h_k \boldsymbol{\phi}_k$. We have
    \begin{align*}
        & \vertiii{ \left( \TSS \right)^{\nu} \left(\TSS_{\lambda_3} \right)^{-1} \left( \TnSS - \TSS  \right) \left( \TSS \right)^{-\nu}  }_{2,2} \\
        & = \sup_{\|g\| \le 1, \| h \| \le 1} \left| \left\langle  \boldsymbol{g}, \ \left( \TSS \right)^{\nu} \left(\TSS_{\lambda_3} \right)^{-1} \left( \TnSS - \TSS  \right) \left( \TSS \right)^{-\nu} \boldsymbol{h} \right\rangle_2 \right| \\
        & = \sup_{\|g\| \le 1, \| h \| \le 1} \left| \left\langle  \left(\TSS_{\lambda_3} \right)^{-1} \left( \TSS \right)^{\nu} \boldsymbol{g}, \   \left( \TnSS - \TSS  \right) \left( \TSS \right)^{-\nu} \boldsymbol{h} \right\rangle_2 \right| \\
        & = \sup_{\|g\| \le 1, \| h \| \le 1} \left| \left\langle  \sum_{k\ge 1} \frac{\rho_k^{\nu} g_k}{\rho_k + \lambda_3} \boldsymbol{\phi}_k, \  \sum_{l \ge 1} \rho_l^{-\nu} h_k \left( \TnSS - \TSS  \right)\boldsymbol{\phi}_l  \right\rangle_2 \right| \\
        & = \sup_{\|g\| \le 1, \| h \| \le 1} \left| \sum_{k,l \ge 1} \frac{\rho_k^{\nu} \rho_l^{-\nu} g_k h_l}{\rho_k + \lambda_3} \left\langle   \boldsymbol{\phi}_k, \  \left( \TnSS - \TSS  \right)\boldsymbol{\phi}_l  \right\rangle_2 \right| \\
        & \le \sup_{\|g\| \le 1, \| h \| \le 1} \left| \sum_{k,l \ge 1}  g_k^2 h_l^2 \right|^{1/2} \left| \sum_{k,l \ge 1} \frac{\rho_k^{2\nu} \rho_l^{-2\nu}}{(\rho_k + \lambda_3)^2} \left\langle   \boldsymbol{\phi}_k, \  \left( \TnSS - \TSS  \right)\boldsymbol{\phi}_l  \right\rangle_2^2 \right|^{1/2} \\
        & \le \left| \sum_{k,l \ge 1} \frac{\rho_k^{2\nu} \rho_l^{-2\nu}}{(\rho_k + \lambda_3)^2} \left\langle   \boldsymbol{\phi}_k, \  \left( \TnSS - \TSS  \right)\boldsymbol{\phi}_l  \right\rangle_2^2 \right|^{1/2}. 
    \end{align*}
    The second inequality from the bottom follows from the Cauchy-Schwarz inequality. By Jensen's inequality
\begin{align*}
    &\mathbb{E} \left| \sum_{k,l \ge 1} \frac{\rho_k^{2\nu} \rho_l^{-2\nu}}{(\rho_k + \lambda_3)^2} \left\langle   \boldsymbol{\phi}_k, \  \left( \TnSS - \TSS  \right)\boldsymbol{\phi}_l  \right\rangle_2^2 \right|^{1/2} \\
    &\le
     \left| \sum_{k,l \ge 1} \frac{\rho_k^{2\nu} \rho_l^{-2\nu}}{(\rho_k + \lambda_3)^2} \mathbb{E} \left\langle   \boldsymbol{\phi}_k, \  \left( \TnSS - \TSS  \right)\boldsymbol{\phi}_l  \right\rangle_2^2 \right|^{1/2}.
\end{align*}
Note that 
\begin{align*}
    & \mathbb{E} \left\langle   \boldsymbol{\phi}_k, \  \left( \TnSS - \TSS  \right)\boldsymbol{\phi}_l  \right\rangle_2^2 \\
    & = \mathbb{E} \left\langle   \boldsymbol{\phi}_k, \  \left( \frac{1}{n}\sum_{i=1}^n \wt\BX_{i\calS} \otimes \wt\BX_{i\calS}^{\top} - \mathbb{E} \wt\BX_{1\calS} \otimes \wt\BX_{1\calS}^{\top}  \right)\boldsymbol{\phi}_l  \right\rangle_2^2 \\
    & = \frac{1}{n} \mathbb{E} \left\langle   \boldsymbol{\phi}_k, \  \left( \wt\BX_{1\calS} \otimes \wt\BX_{1\calS}^{\top} - \mathbb{E} \wt\BX_{1\calS} \otimes \wt\BX_{1\calS}^{\top}  \right)\boldsymbol{\phi}_l  \right\rangle_2^2 \\
    & \le \frac{1}{n} \mathbb{E} \left\langle   \boldsymbol{\phi}_k, \  \left( \wt\BX_{1\calS} \otimes \wt\BX_{1\calS}^{\top}  \right)\boldsymbol{\phi}_l  \right\rangle_2^2  \\
    & = \frac{1}{n} \mathbb{E} \left\langle   \boldsymbol{\phi}_k, \   \wt\BX_{1\calS}  \right\rangle_2^2 \left\langle   \boldsymbol{\phi}_l, \   \wt\BX_{1\calS}  \right\rangle_2^2 \\
    & \le \frac{1}{n} \mathbb{E}^{1/2} \left\langle   \boldsymbol{\phi}_k, \   \wt\BX_{1\calS}  \right\rangle_2^4  \mathbb{E}^{1/2} \left\langle   \boldsymbol{\phi}_l, \   \wt\BX_{1\calS}  \right\rangle_2^4 \\
    & = \frac{3}{n} \mathbb{E} \left\langle   \boldsymbol{\phi}_k, \   \wt\BX_{1\calS}  \right\rangle_2^2 \mathbb{E} \left\langle   \boldsymbol{\phi}_l, \   \wt\BX_{1\calS}  \right\rangle_2^2\\
    & = \frac{3}{n} \rho_k \rho_l. 
\end{align*}
The last inequality follows from the Cauchy-Schwarz inequality. The second-to-last equality from the bottom is derived from the property of Gaussian kurtosis, $\mathbb{E} \langle \boldsymbol{\phi}_k,  \wt\BX_{1\calS} \rangle_{2} ^4 = 3 \left(\mathbb{E} \langle \boldsymbol{\phi}_k,  \wt\BX_{1\calS} \rangle_{2}^2 \right)^2$, where $\langle \boldsymbol{\phi}_k,  \wt\BX_{1\calS} \rangle_{2}$ follows a Gaussian distribution with mean $0$ and variance smaller than $\rho_1$.
Therefore
\begin{align}
    & \mathbb{E}\vertiii{ \left( \TSS \right)^{\nu} \left(\TSS_{\lambda_3} \right)^{-1} \left( \TnSS - \TSS  \right) \left( \TSS \right)^{-\nu}  }_{2,2} \nonumber \\
    & \le \left( \frac{3}{n} \sum_{k\ge 1} \frac{\rho_k^{1+2\nu} }{(\rho_k + \lambda_3)^2} \sum_{l \ge 1} \rho_l^{1-2\nu}  \right)^{1/2} \nonumber \\
    &\le  \left(3\sum_{l \ge 1} \rho_l^{1-2\nu} \right)^{1/2} \left( \frac{1}{n \lambda_3^{1-2\nu}} \sum_{k\ge 1} \frac{\rho_k^{1+2\nu} }{(\rho_k + \lambda_3)^{1+2\nu}} \right)^{1/2} . \label{equ:lemma_s4_expectation}
\end{align}
By Corollary \ref{corollary:eigenvalue_rela}, we have 
\begin{align*}
    \sum_{l \ge 1} \rho_l^{1-2\nu} = \sum_{j=1}^q \sum_{k \ge 1} \left(\rho_{q(k-1)+j} \right)^{1-2\nu}  \le (bc)^{1-2\nu} q \sum_{k \ge 1}  k^{-2r(1-2\nu)} = O(q).
\end{align*}
The last equation holds because $1-2\nu > 1/(2r)$.
By Lemma \ref{lemma:minimax_upper_5}, the expression \eqref{equ:lemma_s4_expectation} can be bounded by $C q^{1/2} \left( (n/q) \cdot \lambda_3^{1-2 \nu + \frac{1}{2r} }\right)^{-1 / 2}$ for some $C>0$. The proof is completed by applying the Markov inequality. 

\noindent\textbf{(2).} Similarly, we can show that
\begin{align*}
    & \mathbb{E}\vertiii{ \left( \TSS \right)^{1/2} \left(\TSS_{\lambda_3} \right)^{-1} \left( \TnSS - \TSS  \right) \left( \TSS \right)^{-\nu}  }_{2,2} \\
    &\le  \left(3 \sum_{l \ge 1} \rho_l^{1-2\nu} \right)^{1/2} \left( \frac{1}{n} \sum_{k\ge 1} \frac{\rho_k^{2} }{(\rho_k + \lambda_3)^{2}} \right)^{1/2} \\
    &\le C' q^{1/2} \left( (n/q) \cdot \lambda_3^{\frac{1}{2r} }\right)^{-1 / 2} 
\end{align*}
for some $C'>0$. The proof is completed by applying the Markov inequality. 
\end{proof}

\begin{lemma} \label{lemma:minimax_upper_5}
For $\lambda < 1$, suppose $\TSS$ satisfies Condition C.\ref{ass:a8}, $\{ \rho_l \}_{l \ge 1}$ are the eigenvalues of $\TSS$. Then there exist constants $c' > 0$ depending only on $b, c, r, \nu$ such that
\begin{align*}
    \sum_{l \geq 1} \frac{\rho_{l}^{1+2 \nu}}{\left(\lambda+\rho_{l}\right)^{1+2 \nu}} \leq c' q \left(1+\lambda^{-1 /(2 r)}\right),
\end{align*}
where $b, c$ are defined in Condition C.\ref{ass:a8}.
\end{lemma}
\begin{proof}
Let $C = bc$, according to Corollary \ref{corollary:eigenvalue_rela}, it is straightforward that
\begin{align*}
    \sum_{l \geq 1} \frac{\rho_{l}^{1+2 \nu}}{\left(\lambda+\rho_{l}\right)^{1+2 \nu}} & = 
    \sum_{j = 1}^q \sum_{k \ge 1} \left(\frac{\rho_{q(k-1)+j}}{\lambda + \rho_{q(k-1)+j} } \right)^{1+2\nu} \\
    & \le q\sum_{k \geq 1} \left(\frac{C k^{-2r}}{\lambda + C k^{-2r} }\right)^{1+2\nu} \\
    & = q C^{1+2\nu} \sum_{k \geq 1} \frac{1}{(\lambda k^{2r}  + C)^{1+2\nu} } \\
    & \le q C^{1+2\nu} \left( C^{-(1+2\nu)} + \int_{1}^{\infty} \frac{dx}{(\lambda x^{2r}  + C)^{1+2\nu}} \right) \\
    & \le qC^{1+2\nu} \left( C^{-(1+2\nu)} + \lambda^{-\frac{1}{2r}} \int_{0}^{\infty} \frac{dy}{(y^{2r}  + C)^{1+2\nu}} \right) \\
    & < q c' \left(1 + \lambda^{-\frac{1}{2r}} \right).
\end{align*}
The last inequality holds because for $r > 1/2$, 
\begin{align*}
    \int_{0}^{\infty} \frac{dy}{ (y^{2r} + C)^{1+2\nu}} < \sum_{k=0}^{\infty} \frac{1}{ (k^{2r} + C)^{1+2\nu}} < C^{-(1+2\nu)} + \sum_{k=1}^{\infty} k^{-2r(1+2\nu)}  < \infty.
\end{align*}
\end{proof}

\begin{lemma}
Suppose that $\boldsymbol{U}_1, \boldsymbol{U}_2$ are jointly Gaussian processes with means $\boldsymbol{\mu}_1,
\boldsymbol{\mu}_2$,
(auto) covariance operators $\boldsymbol{\mathcal{G}}_{11}, \boldsymbol{\mathcal{G}}_{22}$ and
cross covariance operator $\boldsymbol{\mathcal{G}}_{12} = \boldsymbol{\mathcal{G}}_{21}^*$.
Then, conditional on $\boldsymbol{U}_1$, $\boldsymbol{U}_2$ is a Gaussian process with 
mean $\boldsymbol{\mu}_2 + \boldsymbol{\mathcal{G}}_{21} \boldsymbol{\mathcal{G}}_{11}^{-} (\boldsymbol{U}_1 -  \boldsymbol{\mu}_1)$\
and covariance operator $\boldsymbol{\mathcal{G}}_{22} -  \boldsymbol{\mathcal{G}}_{21} \boldsymbol{\mathcal{G}}_{11}^{-} \boldsymbol{\mathcal{G}}_{12}$, where $\boldsymbol{\mathcal{G}}_{11}^{-}$ is the Moore-Penrose generalized 
inverse of $\boldsymbol{\mathcal{G}}_{11}$, and therefore
\begin{align*}
\boldsymbol{U}_2 \eqid \boldsymbol{\mu}_2 + 
\boldsymbol{\mathcal{G}}_{21} \boldsymbol{\mathcal{G}}_{11}^{-} (\boldsymbol{U}_1 -  \boldsymbol{\mu}_1) + \boldsymbol{Z}
\end{align*}
where $\boldsymbol{Z}$ is a zero-mean process independent of $U_1$ and has covariance operator 
$\boldsymbol{\mathcal{G}}_{22} -  \boldsymbol{\mathcal{G}}_{21} \boldsymbol{\mathcal{G}}_{11}^{-} \boldsymbol{\mathcal{G}}_{12}$.
\label{lemma:cond_dist}
 \end{lemma}

\begin{lemma}\label{lm:large_deviation_2nd_ver}
Suppose $U_l \overset{iid}{\sim} \GP(0, \mathscr{G})$, $l=1, \ldots, L $, with $\tr(\mathscr{G}) < \infty$, then for any $s > 1$,
\begin{enumerate}
\item[(i)] 
\begin{align*}
\pr\left(\sum_{l=1}^L \| U_l \|_2^2 > x\right) \le \left(\frac{s}{s-1}\right)^{L/2} \exp\left(-{x\over 2s \cdot \tr(\mathscr{G}) }\right);
\end{align*}
\item[(ii)]
if we further have $x > (1+s/2) L \cdot\tr(\mathscr{G})$, then
\begin{align*}
\pr\left(\sum_{l=1}^L \| U_l \|_2^2 > x\right) \le  \exp\left(  - \frac{ (1 - s^{-1/2})^2 }{2 \| \mathscr{G} \|_2 } (x - L \cdot \tr(\mathscr{G}) ) \right).
\end{align*}
\end{enumerate}
\end{lemma}
The proof of this result is a straightforward application of the following Lemma \ref{lm:large_deviations}.

\begin{lemma} \label{lm:large_deviations}
Suppose that $\xi_{lk}, 1\le m\le L, 1\le k\le K$, are independent random variables 
where $L<\infty, K\le \infty$, $\xi_{lk} \sim N(0,\theta_k)$ for all $l,k$ with 
$\| \theta \|_1 < \infty$, where $\| \theta \|_1 = \sum_{k=1}^K \theta_k$, further define $\| \theta \|_{\infty} = \max_{ \{k=1,\dots,K\} } \theta_k$, then for any $s > 1$,
\begin{enumerate}
\item[(i)]
\begin{align}
\pr\left(\sum_{l=1}^L\sum_{k=1}^K \xi_{lk}^2 > x\right) \le \left(\frac{s}{s-1}\right)^{L/2} \exp\left(-{x\over 2s \| \theta \|_1 }\right); \label{equ:LM_inequality_0}
\end{align}
\item[(ii)]
if we further have $x > (1+s/2) L \| \theta \|_1$, then
\begin{align}
\pr\left(\sum_{l=1}^L\sum_{k=1}^K \xi_{lk}^2 > x\right) \le  \exp\left(  - \frac{ (1 - s^{-1/2})^2 }{2 \| \theta \|_{\infty} } (x - L \| \theta \|_1 ) \right). \label{equ:LM_inequality_1}
\end{align}
\end{enumerate}

\end{lemma}

\begin{proof}
For (i), by Markov's inequality,
\begin{align*}
\pr\left(\sum_{l=1}^L\sum_{k=1}^K \xi_{lk}^2 > x\right)  \le e^{-tx}  
\left\{\prod_{k=1}^K \E\left(e^{t \xi_{1k}^2 }\right)\right\}^L
= e^{-tx} \prod_{k=1}^K \left( 1 - 2t \theta_{k} \right)^{-L/2}.
\end{align*}
Letting $t = (2s \sum_{k=1}^\infty \theta_{k})^{-1}$, $s > 1$, we obtain
\begin{align*}
\prod_{k=1}^K \left( 1 - 2t \theta_{k} \right)^{-L/2}=\prod_{k=1}^K \left( 1 - \frac{\theta_k}
{s \sum_{k=1}^K \theta_k}\right)^{-L/2} \le \left(\frac{s}{s-1}\right)^{L/2},
\end{align*}
where 
the maximum is attained when $\theta_1\not = 0, \theta_2 = \theta_3 = \cdots = 0$. 
To see why the above statement is true, define $r_k = \theta_k
(\sum_{k=1}^K \theta_k)^{-1}$, then we have $0 \le r_k \le 1$, $\sum_{k=1}^K r_k = 1$, denote $\boldsymbol{r}_K = (r_1, \dots, r_K)^{\top}$, define
\begin{align*}
   g_K(\boldsymbol{r}_K) = -\frac{L}{2} \sum_{k=1}^K \log \left( 1 - \frac{r_k}{s} \right). 
\end{align*}
It is straightforward to determine that the function $g_K$ has a compact support and is differentiable. By setting the gradient of $g_K$ with respect to $\boldsymbol{r}_K$ equal to zero, we obtain $r_k \equiv 1/K$, $k=1,\dots, K$, and this leads to the attainment of the function's minimum value. Note that function $g_K$ only have one critical point, as a result, the maximum value must be attained at the boundary of the support of $\boldsymbol{r}_K$. Without loss of generality, we have $r_K = 0$, then the minimum value of $g_{K-1}$ is attained at $r_k \equiv 1/(K-1)$, $k=1,\dots, K-1$, the maximum value of $g_{K-1}$ must be attained at the boundary of $\boldsymbol{r}_{K-1}$. Recursively using this fact, we have $r_1=1$, $r_2 = \dots = r_K = 0$. 

For (ii), the proof utilizes a modified version of the Laurent-Massart inequality \citepsupp{laurent2000adaptive}, as follows. Suppose $Z_j \overset{i.i.d.}{\sim} N(0, 1)$, $a_j \ge 0$ $(j=1, \dots, n)$, define $c = 2 \|a\|_{\infty}$ and $v^2 = 2 \|a\|_2^2$. Then, for any $y > 0$,
\begin{align*}
    \pr \left( \sum_{j=1}^n a_j (Z_j^2 - 1) > y \right) \le \exp\left\{ -\frac{v^2}{2c^2} \left((1+2v^{-2}cy)^{1/2} - 1\right)^2 \right\}.
\end{align*}
Back to our setting, letting $\xi_{lk} = \theta_{k}^{1/2} Z_{lk}$, $v^2 = 2L \| \theta \|_2^2$, $c = 2 \| \theta \|_{\infty}$, and assuming $y > 2^{-1} s L \| \theta \|_1$ $(s > 1)$, we have $2cy / s > 2 L \| \theta \|_1 \| \theta \|_{\infty}  \ge 2 L \| \theta \|_2^2 = v^2$. Then, $2 v^{-2} cy > s > 1$, and in this case
\begin{align*}
    \frac{v^2}{2c^2} \left((1+2v^{-2}cy)^{1/2} - 1\right)^2 > \frac{v^2}{2c^2} \left((2v^{-2}cy)^{1/2} - 1\right)^2  > \frac{ \left( 1 - s^{-1/2} \right)^{-2}}{c} y.
\end{align*}
Subsequently, 
\begin{align*}
\pr\left(\sum_{l=1}^L\sum_{k=1}^K (\xi_{lk}^2 - \theta_k) > y \right) \le  \exp\left(  - \frac{ (1 - s^{-1/2})^2 }{2 \| \theta \|_{\infty} } y \right).
\end{align*}
Let $x = y + L \| \theta \|_1$. Then, for $x > (1+s/2) L \| \theta \|_1$, \eqref{equ:LM_inequality_1} holds.

\end{proof}

The proofs of the following lemmas are straightforward and are omitted.

\begin{lemma} \label{lemma:norm_inequality}
For operator-valued matrices $\boldsymbol{A}$ and $\boldsymbol{B}$,
\begin{enumerate}
\item[(i)]
$\vertiii{\boldsymbol{A}\boldsymbol{B}}_{\alpha,\beta} \le \vertiii{\boldsymbol{A}}_{\eta, \beta}\vertiii{\boldsymbol{B}}_{\alpha, \eta}$ for $\alpha, \beta,\eta \in \{2, \infty \}$;
\item[(ii)] if $\boldsymbol{A}$ has dimension $q \times q$, then
$\frac{1}{\sqrt{q}} \vertiii{\boldsymbol{A}}_{2,2} \le \vertiii{\boldsymbol{A}}_{\infty, \infty} \le \sqrt{q} \vertiii{\boldsymbol{A}}_{2,2}$.
\end{enumerate}

\end{lemma}

\begin{lemma} \label{lemma:varkappa_lower}
For a $q \times q$ operator-valued covariance matrix $\boldsymbol{R}$, suppose $\rho_1$ is the largest eigenvalue of $\boldsymbol{R}$, then for any $\lambda > 0$
\begin{align*}
    \vertiii{ \boldsymbol{R} (\boldsymbol{R} + \lambda \mathscrbf{I})^{-1} }_{\infty, \infty} \ge \frac{\rho_1}{\rho_1 + \lambda}.
\end{align*}
\end{lemma}

\section{Substantiating examples for the technical conditions} \label{s:MA_AR_cases}
We now provide examples of functional predictors that satisfy technical conditions such as C.\ref{ass:a3} and C.\ref{ass:a6}. As described in Remark 2, we consider functional predictors with partially separable covariance structure \citepsupp{Zapata2021supp} such that 
\ben
    \scrbfT^{(\calS,\calS)}=\sum_{k=1}^\infty \BA_k \psi_k \otimes \psi_k,
\een
where $\{\psi_k, k\ge 1\}$ are orthonormal functions in $\bbL_2[0,1]$ and $\{\BA_k, k\ge 1\} $ are a sequence of $q \times q$ covariance matrices. Further, consider $\BA_k = \nu_k \BR$, where $\nu_1 \ge \nu_2 \ge \cdots >0$ are a sequence of eigenvalues and $\BR$ is a $q\times q$ correlation matrix, e.g. a $MA(1)$ correlation matrix. In this setting, $\{X_j, j\in \calS\}$ share the same eigenvalues and eigenfunctions, and their principal component scores have the same correlation structure across different order $k$. To satisfy Condition C.\ref{ass:a2}, $\nu_1=1$ and $\{\nu_k \}$ decay to $0$ fast enough such that $\sum_{k\ge 1} \nu_k <\infty$. To verify C.\ref{ass:a3},
\bse
    \scrbfT^{(\calS,\calS)} (\scrbfT^{(\calS,\calS)}_\lambda)^{-1} = \sum_{k=1}^\infty \BA_k (\BA_k +\lambda \BI)^{-1} \psi_k \otimes \psi_k \equiv \sum_{k=1}^\infty \BB_k \psi_k \otimes \psi_k.
    %
%
\ese
%
Under the setting considered, $\BB_k = \BR (\BR +\vartheta_k \BI)^{-1}$, where $\vartheta_k= \lambda/ \nu_k \to \infty$ as $k\to \infty$. 

\subsection{MA(1) correlation}
We first focus on MA(1) correlation 
\bse
    \BR=\left( \begin{array}{cccccc}
    1 & \rho & 0 & 0& \cdots & 0\\
    \rho & 1 & \rho  & 0 & \cdots & 0\\
    0 & \rho & 1 & \rho   & \cdots & 0\\
    \vdots & 0 & \ddots & \ddots & \ddots & \vdots \\
     0 & \vdots &  &  \rho & 1 & \rho\\
     0& 0 & \cdots & 0 & \rho & 1      
    \end{array}\right).
\ese
In order for $\BR$ to be a legitimate correlation matrix, we need $|\rho| < 1/2$. We will focus on the case $0\le \rho < 1/2$; the same conclusion can be reached for $\rho\in (-1/2, 0)$ using similar arguments. We have
\bse
    \BB_k= \BI - \vartheta_k (\BR +\vartheta_k \BI)^{-1} = \BI - {\vartheta_k \over 1+\vartheta_k}  \wt \BR_k ^{-1} 
\ese
where 
\bse
    \wt \BR_k =\left( \begin{array}{cccccc}
    1 & \wt \rho_k & 0 & 0& \cdots & 0\\
    \wt \rho_k & 1 & \wt \rho_k  & 0 & \cdots & 0\\
    0 & \wt \rho_k & 1 & \wt \rho_k   & \cdots & 0\\
    \vdots & 0 & \ddots & \ddots & \ddots & \vdots \\
     0 & \vdots &  &  \wt\rho_k & 1 & \wt \rho_k\\
     0& 0 & \cdots & 0 & \wt \rho_k & 1      
    \end{array}\right), 
\ese
with $\wt \rho_k= \rho/ (1+\vartheta_k)$. 
Note that both $\BB_k$ and $\wt \BR_k^{-1}$ are positive definite, all diagonal values for both matrices should be greater than $0$, hence $|B_{k, jj}| < 1$ for all $k, j$.
Let $\wt R^{jj'}$ be the $(j,j')$th element of $\wt \BR^{-1}$ and denote 
\bse
    \theta_k={1- \sqrt{1-4\wt \rho_k^2} \over 2 \wt \rho_k}= {2 \wt \rho_k \over 1+ \sqrt{1-4\wt \rho_k^2}  }. 
\ese
One can easily verify that $\theta_k$ is an increasing function of $\wt \rho_k$ and $|\theta_k|<1$. Hence, $\theta_k$ decreases to 0 as $\vartheta_k \to \infty$ with $k$.

By \citesupp{Shaman1969},
\ben
    |\wt R_k^{jj'} | &\le& {1\over \sqrt{1-4\wt \rho_k^2}} \theta_k^{|j-j'|} \nonumber \\
    &\le & {1\over \sqrt{1-4\wt \rho_1^2}} \theta_1^{|j-j'|}\nonumber \\
    &\le & {1\over \sqrt{1-4 \rho^2}} \theta^{|j-j'|},   
\een
where $\theta= {1- \sqrt{1-4 \rho^2} \over 2 \rho} \in [0,1)$. Hence, for $j\neq j'$, $|B_{k, jj'}| \le |\wt R_k^{jj'} |\le  {1\over \sqrt{1-4 \rho^2}} \theta^{|j-j'|} $ uniformly for all $k$. 
By \eqref{eq:infty_norm_bound}
\ben\label{eq:infty_norm_bound_ma1}
    \varkappa = \vertiii{\scrbfT^{(\calS,\calS)} (\scrbfT^{(\calS,\calS)}_\lambda)^{-1}}_{\infty, \infty} \le 1 + {1\over \sqrt{1-4 \rho^2}} {2\theta \over 1- \theta},
\een
which is a constant not depending on $\lambda$ or $q$. We continue to verify C.\ref{ass:a6} in this example:
\bse
    (\scrbfT^{(\calS,\calS)} - \QSS) (\QSS_\lambda)^{-1} = \sum_{k=1}^{\infty} {\nu_k \over \nu_k +\lambda} (\BR-\BI) \psi_k\otimes \psi_k,
\ese
where $\BR$ is the MA(1) correlation matrix above. 
Using the same argument as for \eqref{eq:infty_norm_bound},
\bse
    \vertiii{ (\scrbfT^{(\calS,\calS)} - \QSS) (\QSS_\lambda)^{-1}}_{\infty,\infty}  = 2\rho \max_k {\nu_k \over \nu_k +\lambda} \le 2\rho <1, 
\ese
which satisfies Condition C.\ref{ass:a6}.

\subsection{AR(1) correlation}
We shift our focus towards AR(1) correlation 
\bse
    \BR=\left( \begin{array}{cccccc}
    1 & \rho & \rho^2 & \rho^3& \cdots & \rho^{q-1} \\
    \rho & 1 & \rho  & \rho^2 & \cdots & \rho^{q-2}\\
    \rho^2 & \rho & 1 & \rho   & \cdots & \rho^{q-3}\\
    \vdots & \rho^2 & \ddots & \ddots & \ddots & \vdots \\
     \rho^{q-2} & \vdots & \ddots &  \rho & 1 & \rho\\
     \rho^{q-1} & \rho^{q-2} & \cdots & \rho^2 & \rho & 1      
    \end{array}\right),
\ese
and we will focus on the case $0\le \rho < 1$. Similarly, because $\BB_k= \BI - \vartheta_k (\BR +\vartheta_k \BI)^{-1}$, we have $|B_{k, jj}| < 1$ for all $k, j$. Define $\wt \BR_k = \BR +\vartheta_k \BI$, let $\wt R_k^{jj'}$ be the $(j,j')$th element of $\wt \BR_k^{-1}$, we have $|B_{k, jj'}| \le \vartheta_k |\wt R_k^{jj'}|$ for all $j' \neq j$. 

Consider stochastic process $Y_t$ with AR(1) mean and Gaussian white noise, i.e.
\bse
\left\{
\begin{array}{ll}
    Y_t = \mu_t + W_t, & W_t \overset{i.i.d}{\sim} \mathcal{N}(0, \vartheta) \\
    \mu_t = \rho \mu_{t-1} + V_t, & V_t \overset{i.i.d}{\sim} \mathcal{N}(0, 1 - \rho^2) 
\end{array}
\right.
\ese
then $\boldsymbol{Y}_{[1:q]} \sim \mathcal{N}(\boldsymbol{0}, \wt \BR)$, where $\wt \BR = \BR +\vartheta \BI$. It can be shown that $Y_t$ is an ARMA(1,1) process 
\bse
Y_t = \rho Y_{t-1} + U_t - \theta U_{t-1}, \quad U_t \overset{i.i.d}{\sim} \mathcal{N}(0, \kappa), 
\ese
where $0 \le \theta < 1$, and $(\theta, \kappa)$ satisfies
\bse
\left\{
\begin{aligned}
    & \mbox{Var}(Y_t) =  1 + \vartheta = \frac{1 - 2 \rho \theta + \theta^2}{1 - \rho^2} \kappa\\
    & \mbox{Cov}(Y_t, Y_{t-h}) =  \rho^{|h|} = \frac{(\rho - \theta)( 1 - \rho \theta)}{1 - \rho^2} \rho^{|h|-1} \kappa,
\end{aligned}
\right.
\ese
then 
\begin{align*}
    \theta \le \rho, \quad \frac{\rho}{1 + \vartheta} = \frac{(\rho - \theta)(1-\rho\theta)}{1-2\rho\theta+\theta^2}, \quad {\vartheta \over \kappa} = {\theta \over \rho}.
\end{align*}
According to \citesupp{tiao1971analysis}, for $j' \neq j$, we have
\begin{align*}
    \kappa |\wt R^{jj'}| &\le C \left\{(1 - \rho\theta)^2 \theta^{|j-j'|-1} + (\rho - \theta)^2 \theta^{2q - |j-j'|-1} + (1 - \rho\theta)(\rho - \theta) \left(\theta^{j+j'-2} + \theta^{2q -j-j'} \right) \right\},
\end{align*}
where
\begin{align*}
    C &= \left\{1 + \frac{(\rho - \theta)^2(1-\theta^{2q})}{(1-\rho^2)(1-\theta^2)} \right\}^{-1} \frac{(\rho - \theta)(1-\rho\theta)}{(1-\rho^2)(1-\theta^2)^2} \\
    &\le \frac{1}{(1-\theta^2)^2} \frac{(\rho - \theta)(1-\rho\theta)}{1-2\rho\theta+\theta^2} \\
    & = \frac{\rho}{(1+\vartheta)(1-\theta^2)^2} \\
    & \le \frac{\rho}{(1-\theta^2)^2}.
\end{align*}
Also, note that 
\begin{align*}
   {\vartheta \over \kappa} C \le {\theta \over (1-\theta^2)^{2}}; \quad
   1-\rho\theta \le 1-\theta^2; \quad 
   \rho - \theta \le 1 - \theta \le 1 - \theta^2,
\end{align*}
we have
\begin{align*}
    |B_{k,jj'}| & \le \theta_k^{|j-j'|} + \theta_k^{2q - |j-j'|} + \theta_k^{j+j'-1} + \theta_k^{2q-j-j'+1} \\
    & \le \rho^{|j-j'|} + \rho^{2q - |j-j'|} + \rho^{j+j'-1} + \rho^{2q-j-j'+1} .
\end{align*}
Applying some algebra, we have
\begin{align*}
    & \max_{j}\sum_{j\neq j'} \rho^{|j-j'|} \le \frac{2\rho}{1-\rho}(1-\rho^{q-1}), \quad \max_{j}\sum_{j\neq j'} \rho^{2q-|j-j'|} = \sum_{k=q+1}^{2q-1} \rho^k \le  \frac{\rho^{q+1}}{1-\rho}, \\
    & \max_j \sum_{j'\neq j} \rho^{j+j'-1} + \rho^{2q-j-j'+1} \le \max_j \left( \rho^{j-1}  + \rho^{q-j} \right) \sum_{k=1}^q \rho^k \le \frac{\rho}{1-\rho} (1 + \rho^{q-1}).
\end{align*}

By \eqref{eq:infty_norm_bound} and the above derivation, 
\begin{align}
   \varkappa = \vertiii{\scrbfT^{(\calS,\calS)} (\scrbfT^{(\calS,\calS)}_\lambda)^{-1}}_{\infty, \infty} \le 1 + \frac{3\rho}{1-\rho} \label{eq:infty_norm_bound_ar1}
\end{align}
which is a constant not depending on $\lambda$ or $q$. We continue to verify C. \ref{ass:a6}. Using the same argument as for \eqref{eq:infty_norm_bound}, 
\begin{align*}
     \vertiii{    \sum_{k=1}^{\infty} {\nu_k \over \nu_k +\lambda} (\BR-\BI) \psi_k\otimes \psi_k }_{\infty, \infty} & \le \max_{1\le j \le q} \sum_{j' \neq  j} \max_k \frac{\nu_k}{\nu_k + \lambda} \rho^{|j-j'|} \\
    & \le \max_{1\le j \le q} \sum_{j' \neq  j} \rho^{|j-j'|} \\
    & = \frac{\rho}{1 - \rho} \left(2 - \rho^{\lceil (q-1)/2 \rceil} - \rho^{\lfloor (q-1)/2 \rfloor} \right) \\
    & \le \frac{2\rho}{1-\rho}.
\end{align*}
Hence, for large $q$, we need $\rho \le 1/3$ in order that C. \ref{ass:a6} holds. 

\section{Additional Simulation Results}\label{sec:add_simulation}

\begin{table}[ht]
\caption{Simulation Scenario \II: summary of estimation, prediction, and variable selection performance of the proposed fEnet versus FLR-SCAD under different problem sizes.}
\centering
\scriptsize
\begin{tabular}{cccccccc}
\toprule
$n$ & $p$ & $q$ & Method & FPR $(\%)$ & FNR $(\%)$ & MND & RER \\
\midrule
\multicolumn{8}{c}{$\rho=0$} \\
500 & 50 & 5 & fEnet & 0 (0, 0) & 0 (0, 0) & 1.11 (0.61, 1.82) & 0.0009 (0.0005, 0.0015) \\
& & & FLR-SCAD & 0 (0, 0) & 0 (0, 0) & 1.80 (0.90, 3.59) & 0.0014 (0.0008, 0.0028) \\ 
200 & 100 & 5 & fEnet & 0 (0, 0) & 0 (0, 0) & 1.57 (0.81, 2.37) & 0.0025 (0.0015, 0.0040) \\
& & & FLR-SCAD & 0 (0, 0) & 0 (0, 0) & 2.16 (1.18, 3.71) & 0.0048 (0.0025, 0.0111) \\
100 & 200 & 10 & fEnet & 0 (0, 0.5) & 0 (0, 0) & 3.23 (2.01, 5.05) & 0.0252 (0.0124, 0.0611) \\
& & & FLR-SCAD & 5.8 (1.1, 13.2) & 10 (0, 30)  & 7.49 (4.90, 15.18) & 0.4896 (0.2332, 0.8809) \\
\midrule
\multicolumn{8}{c}{$\rho=0.3$} \\
500 & 50 & 5 & fEnet & 0 (0, 0) & 0 (0, 0) & 1.11 (0.68, 2.05) & 0.0011 (0.0007, 0.0017) \\
& & & FLR-SCAD & 0 (0, 0) & 0 (0, 0) & 1.96 (0.93, 4.11) & 0.0016 (0.0009, 0.0033) \\ 
200 & 100 & 5 & fEnet & 0 (0, 0) & 0 (0, 0) & 1.66 (0.90, 2.52) & 0.0028 (0.0016, 0.0049) \\
& & & FLR-SCAD & 0 (0, 0) & 0 (0, 0) & 2.18 (1.03, 3.60) & 0.0054 (0.0025, 0.0132) \\
100 & 200 & 10 & fEnet & 0 (0, 1.1) & 0 (0, 0) & 3.15 (1.95, 4.97) & 0.0230 (0.0110, 0.0735) \\
& & & FLR-SCAD & 8.4 (4.2, 14.2) & 10 (0, 30)  & 7.60 (4.95, 12.37) & 0.4162 (0.2522, 0.7676) \\
\midrule
\multicolumn{8}{c}{$\rho=0.75$} \\
500 & 50 & 5 & fEnet & 0 (0, 0) & 0 (0, 0) & 1.61 (0.82, 2.63) & 0.0013 (0.0008, 0.0021) \\
& & & FLR-SCAD & 0 (0, 0) & 0 (0, 0) & 3.08 (1.38, 6.41) &  0.0018 (0.0010, 0.0040) \\
200 & 100 & 5 & fEnet & 0 (0, 0) & 0 (0, 0) & 1.95 (0.99, 3.25) & 0.0032 (0.0018, 0.0055)  \\
& & & FLR-SCAD & 0 (0, 2.1) & 0 (0, 0) & 2.93 (1.41, 6.34) & 0.0060 (0.0030, 0.0140) \\
100 & 200 & 10 & fEnet & 0 (0, 3.7) & 0 (0, 10) & 4.15 (2.73, 6.55) & 0.0184 (0.0084, 0.0914)  \\
& & & FLR-SCAD & 4.7 (1.6, 10.6) & 50 (30, 70) & 8.16 (4.95, 16.04)  & 0.2345 (0.1581, 0.3791)  \\
\bottomrule
\end{tabular}
\end{table}

\begin{figure}[ht]
   \centering
    \includegraphics[width=0.87\linewidth]{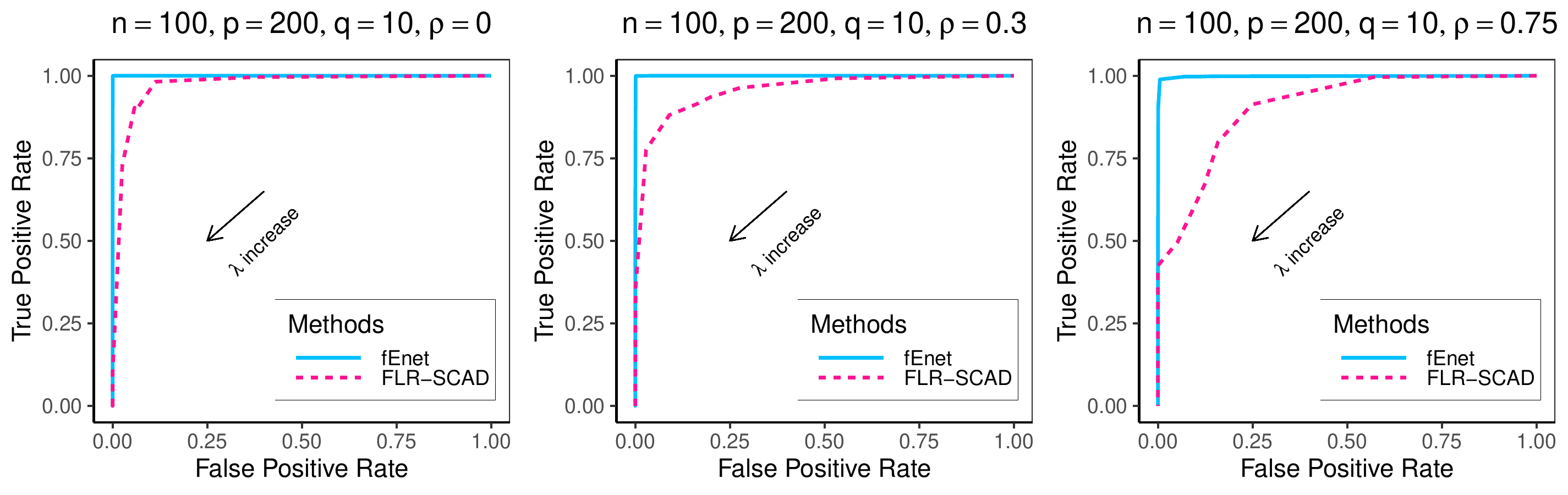}
    \caption{Simulation Scenario \II: the ROC curves of fEnet and FLR-SCAD under the ultra high-dimensional case. The ROC curves are obtained by changing the value of $\lambda$ and holding other hyperparameters as optimal.}
\end{figure}

\begin{figure}[ht]
    \centering
    \includegraphics[width=0.95\linewidth]{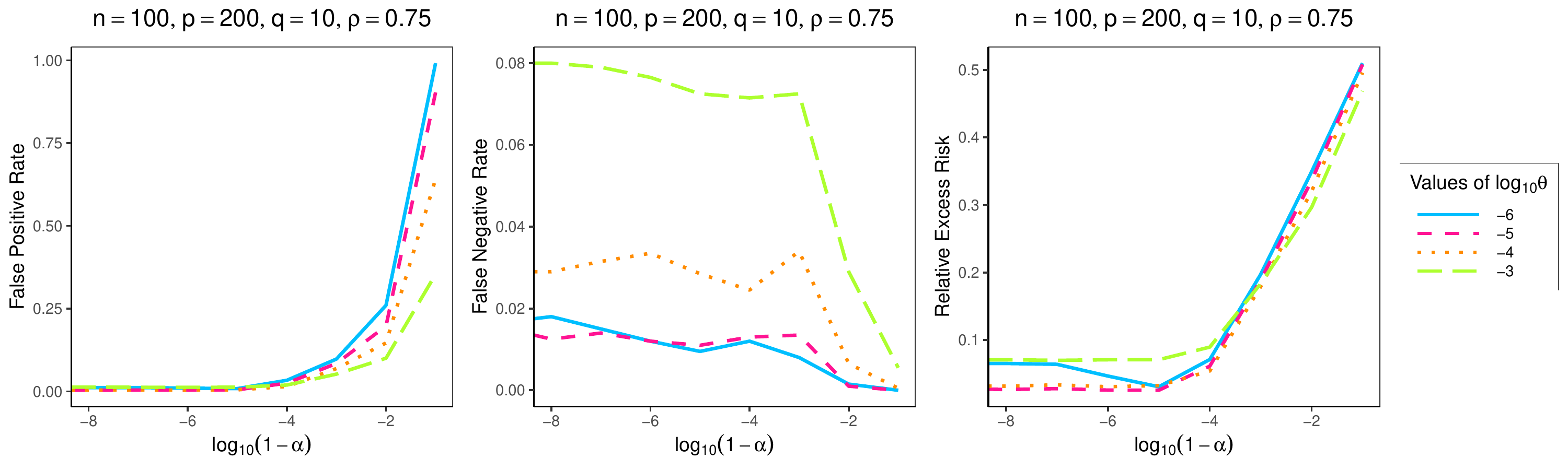}
    \caption{Simulation Scenario \II: the plots of FPR, FNR, and RER versus $\log_{10}(1-\alpha)$ for different values of $\theta$ under the ultra high-dimensional case.}
\end{figure}




%
\begin{table}[ht]
\caption{Simulation Scenario \III: summary of estimation, prediction, and variable selection performance of the proposed fEnet method versus FLR-SCAD under different problem sizes. }
\centering
\scriptsize
\begin{tabular}{cccccccc}
\toprule
$n$ & $p$ & $q$ & Method & FPR $(\%)$ & FNR $(\%)$ & MND & RER \\
\midrule
\multicolumn{8}{c}{$\rho=0$} \\
500 & 50 & 5 & fEnet & 0 (0, 0) & 0 (0, 0) & 0.55 (0.41, 0.82) & 0.0213 (0.0113, 0.0340) \\
& & & FLR-SCAD & 0 (0, 0) & 0 (0, 0) & 0.65 (0.46, 1.09) & 0.0381 (0.0245, 0.0604) \\ 
200 & 100 & 5 & fEnet & 0 (0, 0) & 0 (0, 0) & 0.86 (0.57, 1.39) & 0.0413 (0.0248, 0.0702) \\
& & & FLR-SCAD & 9.5 (4.2, 17.9) & 0 (0, 0) & 0.92 (0.65, 1.37) & 0.0612 (0.0405, 0.1034) \\
100 & 200 & 10 & fEnet & 0 (0, 0.5) & 0 (0, 10) & 1.49 (0.95, 4.18) & 0.0784 (0.0429, 0.2346) \\
& & & FLR-SCAD & 6.8 (2.6, 11.6) & 0 (0, 30)  & 4.01 (2.86, 4.18) & 0.4616 (0.2127, 0.7290) \\
\midrule
\multicolumn{8}{c}{$\rho=0.3$} \\
500 & 50 & 5 & fEnet & 0 (0, 0) & 0 (0, 0) & 0.58 (0.40, 0.89) & 0.0274 (0.0172, 0.0491) \\
& & & FLR-SCAD & 0 (0, 2.2) & 0 (0, 0) & 0.65 (0.48, 0.87) & 0.0528 (0.0353, 0.0830) \\ 
200 & 100 & 5 & fEnet & 0 (0, 0) & 0 (0, 0) & 0.95 (0.61, 1.39) & 0.0562 (0.0338, 0.1042) \\
& & & FLR-SCAD & 9.5 (4.2, 15.8) & 0 (0, 0) & 0.96 (0.66, 1.41) & 0.0797 (0.0503, 0.1410) \\
100 & 200 & 10 & fEnet & 0 (0, 1.1) & 0 (0, 20) & 1.84 (1.32, 4.18) & 0.1048 (0.0618, 0.3288) \\
& & & FLR-SCAD & 8.4 (3.7, 13.2) & 20 (0, 50)  & 4.18 (3.88, 4.18) & 0.5074 (0.3487, 0.7764) \\
\midrule
\multicolumn{8}{c}{$\rho=0.75$} \\
500 & 50 & 5 & fEnet & 2.2 (0, 6.7) & 0 (0, 0) & 0.86 (0.62, 1.42) & 0.0504 (0.0276, 0.0926) \\
& & & FLR-SCAD & 26.7 (13.3, 37.8) & 0 (0, 0) & 1.05 (0.73, 3.59) &  0.0870 (0.0506, 0.1701) \\
200 & 100 & 5 & fEnet & 1.1 (0, 4.2) & 0 (0, 20) & 1.45 (0.90, 4.18) & 0.1411 (0.0603, 0.3734)  \\
& & & FLR-SCAD & 9.5 (3.2, 16.8) & 20 (0, 40) & 4.18 (1.29, 4.18) & 0.3056 (0.1227, 0.5523) \\
100 & 200 & 10 & fEnet & 0.5 (0, 1.6) & 40 (20, 50) & 4.18 (4.18, 4.18) & 0.1518 (0.0878, 0.2769)  \\
& & & FLR-SCAD & 5.3 (2.1, 9.0) & 60 (40, 70) & 4.19 (4.18, 6.16)  & 0.2467 (0.1616, 0.3688)  \\
\bottomrule
\end{tabular}
\end{table}

\begin{figure}[ht]
    \centering
    \includegraphics[width=0.87\linewidth]{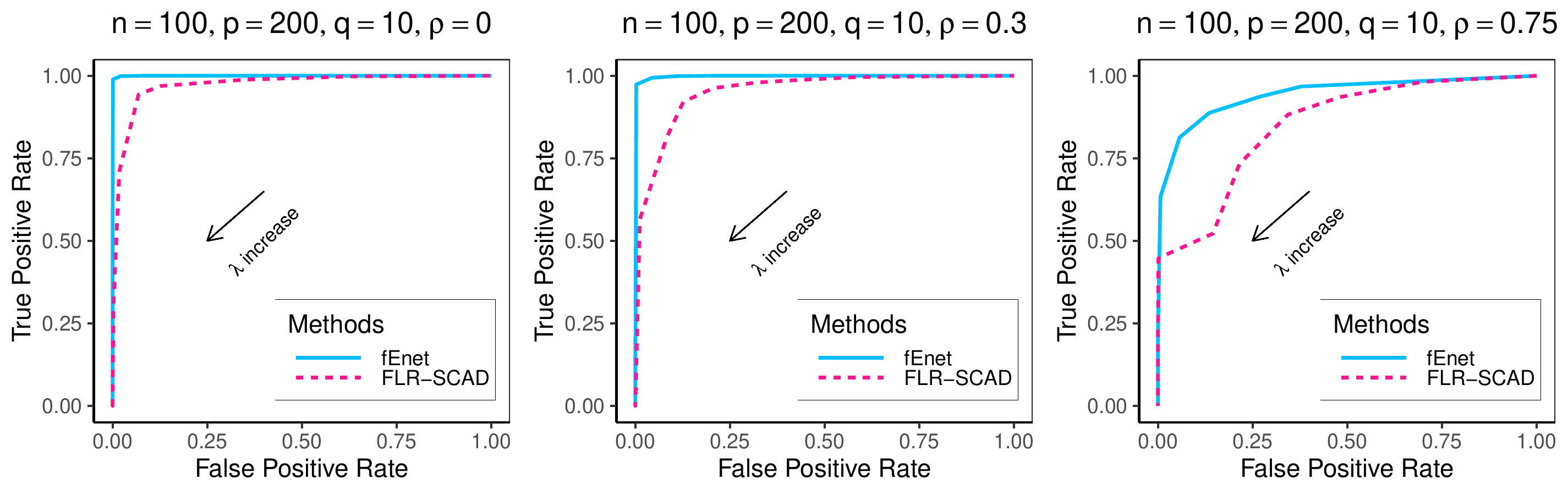}
    \caption{Simulation Scenario \III: the ROC curves of fEnet and FLR-SCAD under the ultra high-dimensional case. The ROC curves are obtained by changing the value of $\lambda$ and holding other hyperparameters as optimal.}
\end{figure}

\begin{figure}[ht]
    \centering
    \includegraphics[width=0.95\linewidth]{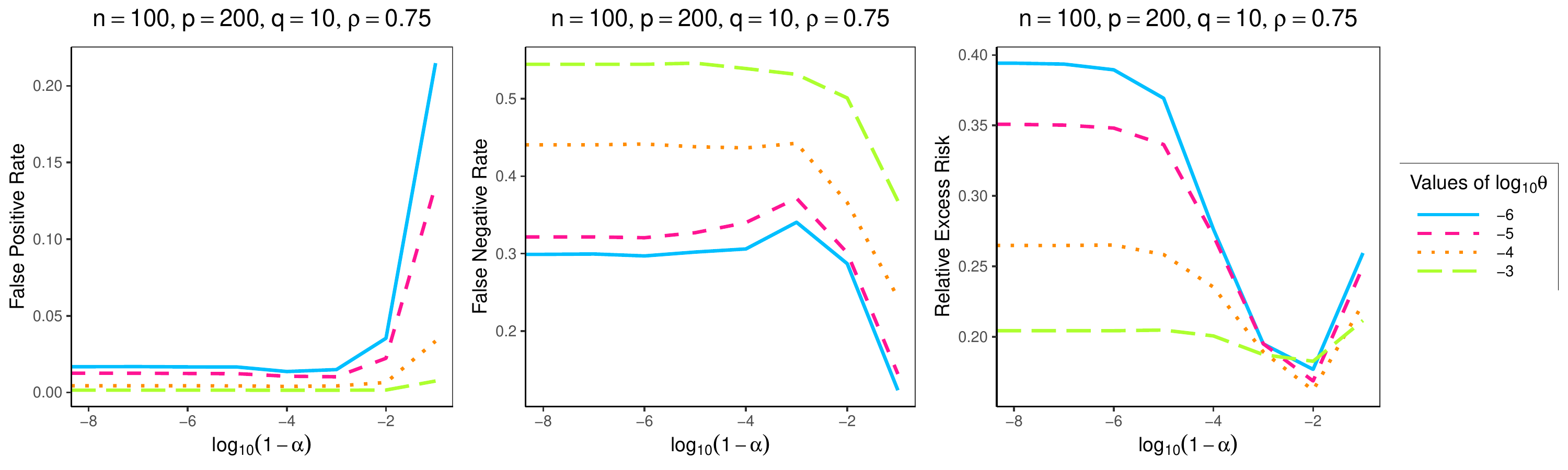}
    \caption{Simulation Scenario \III: the plots of FPR, FNR, and RER versus $\log_{10}(1-\alpha)$ for different values of $\theta$ under the ultra high-dimensional case.}
\end{figure}

\clearpage\pagebreak\newpage

\bibliographystylesupp{apalike}
\bibliographysupp{flm_vs_ref}

\end{document}